\documentclass[12pt]{article}
\usepackage{jheppub}

\usepackage{tikz}
\usetikzlibrary{arrows,decorations.markings,shapes.arrows,patterns,calc}

\usepackage{mathrsfs}
\usepackage{amsthm, amscd,amssymb, amsfonts, verbatim,subfigure, enumerate}
\usepackage[mathcal]{eucal}
\usepackage[super]{nth}

\newcommand{\lvac}{\langle \emptyset \rvert}
\newcommand{\rvac}{\lvert \emptyset \rangle}
\newcommand{\CP}{\mathbb{CP}}
\newtheorem{lemma}{Lemma}
\newtheorem{theorem}{Theorem}
\newtheorem{proposition}{Proposition}
\newtheorem{conjecture}{Conjecture}
\newcommand{\wbar}{\br{w}}
\newcommand{\tr}{\triangle}
\newcommand{\cinfty}{C^{\infty}}
\newcommand{\iso}{\cong}
\newcommand{\PV}{\op{PV}}
\newcommand{\dbar}{\br{\partial}}
\newcommand{\gl}{\mf{gl}}
\renewcommand{\c}{\mathsf{c}}
\renewcommand{\b}{\mathsf{b}}
\newcommand{\zbar}{\br{z}}

\newcommand{\eps}{\epsilon}

\newcommand{\xto}{\xrightarrow}

\newcommand{\what}{\widehat}

\newcommand{\til}{\widetilde}
\newcommand{\mscr}{\mathscr}
\renewcommand{\det}{\operatorname{det}}

\newcommand{\br}{\overline}

\newcommand{\C}{\mathbb C}

\newcommand{\norm}[1]{\left\| #1 \right\|}
\newcommand{\Oo}{\mscr O}
\newcommand{\Z}{\mathbb Z}

\newcommand{\op}{\operatorname}
\newcommand{\mbf}{\mathbf}
\newcommand{\mbb}{\mathbb}
\newcommand{\mf}{\mathfrak}
\newcommand{\mc}{\mathcal}

\newcommand{\ip}[1]{\left\langle #1 \right\rangle}
\newcommand{\abs}[1]{\left| #1 \right|}

\newcommand{\R}{\mbb R}
\renewcommand{\d}{\mathrm{d}}

\DeclareMathOperator{\Sym}{Sym} \DeclareMathOperator{\Hom}{Hom}

\DeclareMathOperator{\Tr}{Tr}

\title{Twisted Holography}
\abstract{We derive and test a novel holographic duality in the B-model topological string theory.
The duality relates the B-model on certain Calabi-Yau three-folds to two-dimensional chiral algebras defined as gauged $\beta \gamma$ 
systems. The duality conjecturally captures a topological sector of more familiar AdS$_5$/CFT$_4$ holographic dualities.}

\author[1]{Kevin Costello,}
\author[1]{Davide Gaiotto}
\affiliation[1]{Perimeter Institute for Theoretical Physics, Waterloo, Ontario, Canada N2L 2Y5}

\begin{document}
\maketitle

\section{Introduction and Conclusions}
In this paper we introduce a general family of ``topological'' holographic dualities. The ``gravitational'' side of the duality 
involves B-model topological string theory \cite{Bershadsky:1993ta, Bershadsky:1993cx} on Calabi-Yau manifolds with an SL$_2(\C)$ isometry group. The 
``gauge theory'' side of the duality involves certain vertex operator algebras defined as gauged $\beta \gamma$ systems.

The dualities are based on a standard holographic dictionary \cite{Maldacena:1997re,Witten:1998qj}, with a
one-to-one correspondence between the insertion of local operators on the gauge theory side and 
local modifications of the asymptotic boundary conditions on the gravitational side. 
\footnote{By the state-operator map, this is dictionary can be also interpreted as a 
direct identification of the spaces of states on the two sides of the duality. This can be a useful perspective 
for making sense of the notions of ``local operator'' and of ``local modification of the boundary conditions''.}
Both sides enjoy a local conformal invariance. 

Our main example is the holographic duality between 
\begin{itemize}
\item A gauged $\beta \gamma$ system $\mathcal{A}_N$ transforming in the adjoint of $\mathfrak{u}(N)$. 
\item B-model topological strings on an SL$_2(\C)$ background (also called ``deformed conifold''), with $N$ being the period
of the holomorphic three form. \footnote{We absorb the topological string coupling into the normalization of the three form.}
\end{itemize} 

The existence of some open-closed duality involving the deformed conifold and a stack of topological D2 branes, say wrapping the $\CP^1$ cycle in the resolved conifold
$O(-1) \oplus O(-1) \to \CP^1$, has been suggested several times in the literature, e.g. in \cite{Gopakumar:1998ki}, and was an important ingredient in the development of the 
Dijkgraaf-Vafa duality \cite{Dijkgraaf:2002dh,Dijkgraaf:2002fc,Aganagic:2003qj,Dijkgraaf:2007sx}, which in particular relates
the Gaussian matrix model to a certain subsector of the B-model on the deformed conifold. 

Our proposal sharpens these duality statements into a fully-fledged holographic duality and appears to be compatible with these earlier proposals. 
In particular, the $\mathcal{A}_N$ vertex algebra includes a small ${\cal N}=4$ super-Virasoro algebra. The topological twist of this super-conformal algebra 
selects a topological subsector of $\mathcal{A}_N$ which roughly matches the Dijkgraaf-Vafa subsector of the B-model.

Notice that the gauge theory side is well-defined for all finite values of $N$, while the gravitational side is 
only defined as an asymptotic series in $N^{-1}$. Our proposal is thus testable all orders in string perturbation theory and can be taken as a non-perturbative definition of B-model topological strings on certain backgrounds. 

In this paper we will limit ourselves to tree level tests. 
\begin{enumerate} 
	\item We show that single-trace operators in the large $N$ chiral algebra are in bijection with modifications of the boundary conditions of the holographic dual gravitational theory. In particular, the Witten indices on both sides coincide.
    \item 		 We use Witten diagrams to define the two- and three-point functions of the holographic dual theory,  and the OPE. We prove algebraically that all two- and three-point functions match.  
\end{enumerate}

Our basic duality can be modified in a variety of ways by adding extra D-branes on the gravity side. 
We will consider, in particular, space-filling branes. 

We also propose a broader family of dualities, involving
\begin{itemize}
\item A gauged $\beta \gamma$ system $\mathcal{A}^\Gamma_N$ based on an affine ADE quiver of type $\Gamma$. 
\item B-model topological strings on an SL$_2(\C)/\Gamma$ background.
\end{itemize} 

\subsection{Relations to physical dualities}
We will present a ``derivation'' of our dualities which is completely analogous to the original derivation of the 
canonical example of holography: the duality between maximally supersymmetric, four-dimensional 
Yang-Mills theory with $U(N)$ gauge group in conformally flat spacetime and Type IIB superstring theory in an 
AdS$_5 \times$S$^5$ background \cite{Maldacena:1997re,Witten:1998qj}.\footnote{We should also recall previous examples of topological open-closed dualities in topological strings, such as the geometric transition \cite{Gopakumar:1998ki,Ooguri:1999bv}, relating A-model topological strings on the resolved conifold and Chern-Simons gauge theory.  A holographic analysis, more in the spirit of the present work, was applied to a different topological string theory setup in the interesting recent paper \cite{IMZ2018}.} 

In both cases, the two sides of the duality arise respectively as 
\begin{itemize}
\item The low-energy limit of the world-volume theory of a certain stack of D-branes in flat space.
\item The near-horizon limit of the back-reacted geometry sourced by the same stack of D-branes. 
\end{itemize}

The original example of holography \cite{Maldacena:1997re} involves a stack of $N$ D3 branes in Type IIB string theory wrapping 
$\R^{3,1} \subset \R^{9,1}$. 
The conformal symmetry group emerges on both sides as a low-energy/near horizon enhancement 
of the Lorentz group. \footnote{Recall that the D3 brane back-reacted geometry is
\begin{equation}
ds^2 = \left(1 + \frac{g_s N}{r^4} \right)^{-\frac12} \sum_{i=1}^4 dx_i^2 +  \left(1 + \frac{g_s N}{r^4} \right)^{\frac12} \left(dr^2 + r^2 ds^2_{S^5} \right)
\end{equation}
After a scaling limit $r \to 0$ recovers AdS$_5 \times$S$^5$. The change of coordinates $y = \sqrt{g_s N}r^{-1}$ gives
\begin{equation}
ds^2 = \sqrt{g_s N} \frac{dy^2 + \sum_{i=1}^4 dx_i^2}{y^2} + \sqrt{g_s N} ds^2_{S^5}
\end{equation}
and shows that the $y \to 0$ asymptotic boundary region of AdS$_5$ replaces the flat space asymptotic region of the original geometry. } Analogously, our  proposal involves a stack of $N$ topological B-branes wrapping $\C \subset \C^3$. 
No extra scaling limits are necessary and the enhanced conformal symmetry group emerges immediately
in the back-reacted geometry. 

We believe that our topological duality actually coincides with a topological sub-sector of 
the physical holographic duality. Indeed, this paper began as an effort to 
identify the holographic dual description of a certain class of protected correlation functions 
in ${\mathcal N}=4$ Super Yang Mills with gauge group $U(N)$. 
These protected correlation functions are encoded in a 2d chiral algebra introduced in \cite{Beem:2013sza}
as the cohomology of a certain generator of the ${\mathcal N}=2$ super-conformal algebra.
The 2d chiral algebra for ${\mathcal N}=4$ $U(N)$ SYM precisely coincides with $\mathcal{A}_N$. 

The operation of cohomological twist of the gauge theory is expected to correspond to twisting supergravity in the sense of \cite{Costello:2016mgj}. Explicitly computing  twisted supergravity is difficult. However, in \cite{Costello:2016mgj},  conjectural descriptions of certain twists of type IIB supergravity were given, in terms of the topological $B$-model on certain $5$-dimensional Calabi-Yau's.  We show, assuming the conjectures of \cite{Costello:2016mgj}, that twisting type IIB supergravity on $AdS_5 \times S^5$ using the supercharge of \cite{Beem:2013sza}  localizes to the topological $B$-model on the deformed conifold.  

A direct localization analysis restricted to the first KK mode in $AdS_5$ can be found in \cite{Bonetti:2016nma}. 

\subsection{Global symmetry algebras}
The large $N$ chiral algebra admits an infinite-dimensional Lie algebra of global symmetries which we call $\mbf{a}_{\infty}$. This consists of all modes of single trace operators which preserve the vacuum at $0$ and $\infty$. This includes the $SL_2$ global conformal symmetries algebra and the $SU(2)_R$ global $R$-symmetries, as well as an infinite tower of higher-spin global symmetries.

The holographic chiral algebra also has an infinite-dimensional Lie algebra of global symmetries,  $\mf{a}_{\infty}^{hol}$.  This is the Lie algebra of infinitesimal gauge symmetries which preserve the vacuum solution to the equations of motion on the deformed conifold $SL_2(\C)$.  It's bosonic part is $\op{Vect}_0(SL_2(\C))$, the Lie algebra of holomorphic vector fields on this manifold which preserve the holomorphic volume form.    This Lie algebra acts by global symmetries on the holographic chiral algebra, defined in terms of Witten diagrams. 

In each case, these global symmetries are strong enough to fix all two- and three-point functions  except the $2$-point functions of a small number of operators of dimension $\le 2$.  

One of our main results is that these global symmetry algebras are isomorphic as super-Lie algebras: $\mf{a}_{\infty} \iso \mf{a}_{\infty}^{hol}$.  This is what allows us to match the holographic and large $N$ two and three point functions. 

We view the identification of the large $N$ and holographic global symmetry algebras as an important part of the topological-holographic dictionary.  The holographic global symmetry algebra has a very clean geometric and physical interpretation, and is easy to compute.  This is in contrast to the holographic chiral algebra, whose structure constants are encoded in difficult Witten diagram computations.

The holographic global symmetry algebra can also be interpreted in string-theory terms.  For any topological or physical string theory the BRST cohomology of the space of closed-string\footnote{This structure is also present once we include open strings.} states is a graded Lie algebra. The Lie bracket is givne by a string scattering processes where $2$ closed strings merge to $1$. The Jacobi identity is a consequence of the classical master equation satisfied by tree-level string vertices in either physical or topological strings, \cite{SenZwi96,Cos05,Sen:2015uaa}.  

The second cohomology group $H^2$ of the space of closed string states describes consistent modifications of the closed string background; and $H^1$ describes the infinitesimal symmetries of the closed string background. \footnote{This statement is closely related to the more familiar physical statements that BRST-closed vertex operators of ghost number $2$ define infinitesimal fluctuations of the background and BRST-closed vertex operators 
of ghost number $1$ define infinitesimal symmetries of the background \cite{Witten:1991zd}.}  The Lie bracket is defined on the cohomology of the space of closed string states, with the ghost number shifted by $1$, putting symmetries in degree $0$. 

In our context, the holographic global symmetry algebra is isomorphic to the cohomology of the space of closed-string states, with a shift of the ghost number by one, and with this Lie bracket. 

We conjecture that in any topological-holography setting, the large $N$ CFT has a Lie algebra of global symmetries, built from modes (and modes of descendents) which preserve the vacuum at $0$ and $\infty$.   This Lie  algebra should be isomorphic to that of closed-string states on the holographic dual closed-string background. 

\subsection{Open questions and future directions}

Our proposal leads to several interesting open problems, both of perturbative and non-perturbative nature. 
Here are some notable examples:
\begin{itemize}
\item In this paper we set up the basic ingredients for a systematic perturbative expansion of Kodaira-Spencer theory in the deformed conifold background.
It would be very interesting to do concrete loop calculations and match them to the $N^{-1}$ expansion of $\mathcal{A}_N$ correlation functions.
It would be even more interesting to make manifest the interpretation of these calculations as finite, well-defined loop calculations in IIB supergravity.
\item It would be very interesting to study the topological analogue of ``giant gravitons'' \cite{McGreevy:2000cw,Grisaru:2000zn}, i.e. local operators of determinant form which are holographic dual to probe D-branes rather than strings. A full match would considerably extend and sharpen the semiclassical comparison proposed in \cite{Biswas:2006tj}. 
\item It would be very interesting to consider the holographic dual of conformal blocks and correlation functions of $\mathcal{A}_N$
on a general Riemann surface. A precise match will likely require a sum over semiclassical saddles on the ``gravitational'' side 
of the duality, perhaps analogous to past AdS$_3$/CFT$_2$ proposals ( e.g. \cite{Maloney:2007ud,Dijkgraaf:2000fq,Yin:2007gv}).
It may give general lessons about the choice of saddles which may contribute to a quantum gravity calculation. 
\item A full classification of modules for $\mathcal{A}_N$ may lead to a non-perturbative classification of dynamical objects which can be found in the B-model topological string. 
\item It should be straightforward to extend this work to chiral algebras associated to gauge theories with orthogonal or symplectic gauge groups and Type IIB holographic duals \cite{Witten:1998xy}. 
\item It should be possible to find holographic dual descriptions for all chiral algebras associated to gauge theories of class ${\mathcal S}$ with known M-theory holographic duals \cite{Gaiotto:2009gz}. 
\item It would be interesting to fully uncover the relationship between the topological twist of $\mathcal{A}_N$, the Dijkgraaf-Vafa duality and its relationship with integrable hierarchies \cite{Dijkgraaf:2002dh,Dijkgraaf:2002fc,Aganagic:2003qj,Dijkgraaf:2007sx}. 
\item The B-model topological strings have a somewhat mysterious relationship with the $c=1$ non-critical string theory at the self-dual radius \cite{Witten:1991zd,Ghoshal:1995wm}. It would be nice to clarify the connections between our holography proposal and non-critical string theory. 
\end{itemize}

\section{The ${\mathcal N}=4$ chiral algebra}

The definitions in \cite{Beem:2013sza} associate a vertex algebra to any 
4d ${\mathcal N}=2$ superconformal quantum field theory. 

If the super-conformal quantum field theory is a gauge theory with exactly marginal couplings and a weakly coupled limit, 
the vertex algebra can be defined as a BRST reduction of a free vertex algebra. 
Namely, consider a gauge theory with gauge group $G$, acting on a symplectic representation 
$M$. There are free vertex algebras $\mathrm{Sb}^M$ and $\mathrm{bc}^{\mathfrak{g}}$
of symplectic bosons valued in $M$ and bc system valued in the Lie algebra $\mathfrak{g}$.

The symplectic bosons are defined by an OPE 
\begin{equation}
Z_a(z) Z_b(w) \sim \frac{\omega_{ab}}{z-w}
\end{equation}
and the bc system by 
\begin{equation}
b_I(z) c^J(w) \sim \frac{\delta_I^J}{z-w}
\end{equation}

The BRST operator takes the schematic form
\begin{equation}
Q_{\mathrm{BRST}} = \oint t_I^{ab} c^I Z_a Z_b + f^I_{JK} b_I c^J c^K
\end{equation}
and is nilpotent iff the second Casimir of $M$ has the correct value to cancel the beta function of the original gauge theory. 
The BRST-reduced vertex algebra $\mathrm{Sb}^M//G$ is defined as the $Q_{\mathrm{BRST}}$ cohomology of G-invariant operators, with the $c$ ghost zeromode removed.

In the specific case of maximally symmetric $U(N)$ gauge theory, the representation $M$ is the sum of two copies of 
$\mathfrak{g}$. The corresponding vertex algebra $\mathcal{A}_N$ is thus built from a pair $Z = (Z^1,Z^2)$ of adjoint fields,
rotated by an $\widehat{\mathfrak{su}}(2)$ current algebra with generators $J^{ij} = \mathrm{Tr} Z^i Z^j$, which is in the BRST cohomology. 

The BRST cohomology contains generators in non-zero ghost number. In particular, we have ``supercurrents'' of dimension $3/2$ and 
ghost number $\pm 1$: 
\begin{equation}
G^i_+ = \mathrm{Tr} \,b Z^i  \qquad \qquad G^i_- = \mathrm{Tr} \,\partial c Z^i 
\end{equation}
They combine with the $SU(2)$ currents and the stress tensor into a (small) ${\mathcal N}=4$ super-Virasoro algebra. 
\footnote{The (small) ${\mathcal N}=4$ super-Virasoro algebra has an $SU(2)_o$ automorphism which rotates $G^i_\pm$ into each other
and extends the ghost number symmetry. It is not immediately obvious if this is a symmetry of the full chiral algebra $\mathcal{A}_N$ (either at finite $N$ or in the large $N$ limit). It is definitely not a symmetry of the underlying four-dimensional gauge theory, as it mixes components of fermions of different chirality. It may have an interesting physical interpretation if the chiral algebra is constructed as a boundary vertex algebra for a three-dimensional gauge theory, as in \cite{Costello:2018fnz}, as an IR symmetry enhancement. We will see a similar emergent symmetry in the B-model dual.}

\subsection{The large $N$ limit}
In the large $N$ limit, the gauge-invariant operators built from adjoint-valued fields organize themselves according to some well-known combinatorics. 
Operators whose dimension remains finite in the large $N$ limits can be written as polynomials of single-trace (a.k.a.\ ``closed string'') operators, with no relations among themselves. 

The computation of a correlation function or OPE involves some collection of Wick contractions, which can be organized by the degree of planarity of the corresponding ribbon graph. For example, if we do not insert powers of $N$ in the normalization of single-trace operators, the two-point functions of an operator of length $\ell$ 
(i.e. a trace of a product of $L$ adjoint fields) will have a leading order $N^{\ell}$. That means an operator of 
length $\ell$ should be thought of as a quantity of order $N^{\frac{\ell}{2}}$. The same normalization works for 
multi-trace operators, while two-point functions between operators with a different number of traces have sub-leading powers of $N$.

If we normalize operators of length $\ell$ by a factor of $N^{-\frac{\ell}{2}}$, the planar contributions to the OPE will give 
the identity operator at order $1$ and single-trace operators at order $N^{-1}$. This is consistent with $N^{-1}$ being the 
closed string coupling of an holographic dual description of the chiral algebra. \footnote{The combinatorics is easily understood:
when computing the OPE of two single-trace operators of lengths $\ell_1$ and $\ell_2$, 
we need to contract at least a pair of fields. This will give a single-trace operator of length $\ell_1 + \ell_2 -2$ 
and produce no extra factors of $N$. Contracting $k$ pairs of fields in a planar way will produce
a single-trace operator of length $\ell_1 + \ell_2 -2k$ and produce $k-1$ extra factors of $N$.
Non-planar contractions will have sub-leading powers of $N$, and will also yield operators which are no longer single-trace.}
For later convenience we will not include any powers of $N$ in the normalization of the single trace operators. 

The planar OPE is linear in the single-trace operators, which means that the 
Fourier modes of the chiral algebra form a Lie algebra in the planar limit.
It is important to remember, however, that the non-singular part of the OPE of two single-trace operators is 
dominated by double-trace operators, which are the terms with no Wick contractions.  

Correspondingly, the decomposition of the space of operators into single-traces, double-traces,
etc. does not play well with the full vertex operator algebra structure: the action of Fourier modes on
the space of states or space of local operators typically adds or removes traces. 

\subsection{Large $N$ single-trace operators}

In the large $N$ limit, the BRST reduction can be executed directly on single-trace operators. 
The full algebra will consist of polynomials in the BRST-closed single-trace operators.
In \cite{Beem:2013sza} it is conjectured that the single-trace BRST cohomology is generated by symmetrized traces
\begin{equation}
A^{(n)} = \Tr Z^{(i_1} Z^{i_2} \cdots Z^{i_n)}
\end{equation}
and their global super-conformal partners
\begin{align}
B^{(n)} &= \Tr b Z^{(i_1} Z^{i_2} \cdots Z^{i_n)} \cr
C^{(n)} &= \Tr \partial c Z^{(i_1} Z^{i_2} \cdots Z^{i_n)} \cr
D^{(n)} &= \frac12 \epsilon_{ij} \Tr \partial Z^{i} Z^{(j} \cdots Z^{i_n)} + \cdots
\end{align}
where the ellipsis omits terms involving the ghosts. 
This conjecture can be readily tested at the level of the character of the chiral algebra:  see Appendix \ref{app:index}.  Further, in Appendix \ref{app:coho} we give a proof of this conjecture, by using spectral sequences to reduce it to some classic results in homological algebra \cite{Tsy83,LodQui84}.

\subsection{Adding Flavor}
\label{subsection_matter}
The choice of matter representation in an ${\mathcal N}=2$ gauge theory must be carefully tuned in order for the theory to 
be super-conformal: the overall beta function receives contributions of opposite sign (proportional to the second Casimir invariant) 
from the gauge fields and the matter fields, which must cancel. Precisely the same condition is required in order to avoid a BRST anomaly
in the construction of the corresponding chiral algebra.

There is an interesting loophole: the contribution of some matter fields to the beta function will switch sign if we switch their Grassmann
parity. This is a somewhat unphysical choice, as it breaks unitarity of the physical gauge theory. It is perfectly fine, though, at the level of the chiral algebra.

In particular, we may consider a chiral algebra $\mathcal{A}^{k|k}_N$ where we take the $U(N)$ BRST reduction of a set of symplectic bosons in the 
adjoint representation, $k$ sets of symplectic bosons in the fundamental plus anti-fundamental representation, and $k$ sets of complex fermions in the fundamental plus anti-fundamental representation.  
We can denote the adjoint fields again as $Z^i$, the fundamental fields as $J$ and the anti-fundamental fields as $I$.

The resulting chiral algebra is rather interesting. The extra terms in the BRST operator break some cancellations which were necessary for the presence in $\mathcal{A}_N$ 
of cohomology in non-zero ghost number. At large $N$, half of the fermionic super-conformal generators are gone from the cohomology \footnote{Physically, this is closely related to the 
fact that the four-dimensional gauge theory has only ${\mathcal N}=2$ supersymmetry. }. The flip side of this disappearance is that one of the naive gauge-invariant $\mathfrak{u}(k|k)$ 
affine currents $IJ$ are not in the BRST cohomology, so that $\mathcal{A}^{k|k}_N$ has a $\mathfrak{pgl}(k|k)$ current algebra\footnote{This discussion is valid in the planar limit. At order $1/N$ the other fermionic super-conformal generators also cancel and we have a $\mf{psl}(k \mid k)$ current algebra.}.   

In the large $N$ limit, the gauge-invariant operators organize themselves into polynomials of closed string operators and ``open string'' operators,
i.e. matrix products of adjoint fields sandwiched between an anti-fundamental and a fundamental field. 

An open string operator with $\ell-1$ adjoint fields has a two-point function of order $N^\ell$ and should be thought of as an object of 
order $N^{\frac{\ell}{2}}$. The OPE of open string operators of lengths $\ell_1$ and $\ell_2$ gives open string operators of length 
$\ell_1 + \ell_2 - 2k-1$ at order $N^k$ and closed string operators of length 
$\ell_1 + \ell_2 - 2k-2$ at order $N^k$. The latter contribution to the OPE is thus sub-leading and open string operators 
have OPE which close to open string operators at the leading order in $N$. The OPE of an open and a closed string operator 
is dominated by open string operators. 

 There are terms in the BRST operator which mix the open and closed string operators. Let us first describe the BRST cohomology where we drop these terms.   
 
 The BRST cohomology of closed string operators 
is generated, as before, by the bosonic $A^{(n)}$ and $D^{(n)}$ towers, and the fermionic $B^{(n)}$ and $C^{(n)}$ towers. The BRST cohomology of open string operators is generated by 
\begin{equation}
E_\mathfrak{t}^{(n)} = \Tr \mathfrak{t}\, I  Z^{(i_1} Z^{i_2} \cdots Z^{i_n)} J
\end{equation}
where $\mathfrak{t}$ is a constant element of the $\mathfrak{gl}(k|k)$ current algebra. We will also write $E_{ij}^{(n)}$ where $i,j$ run over a basis of $\C^{k \mid k}$. 

The operators $E_{ij}^{(n)}$ transform in the adjoint representation of $GL(k \mid k)$, and representation of spin $n/2$ of $SU(2)_R$, and are of spin $n/2+1/2$.  

This description has neglected the term in the BRST operator which mixes open and closed string states. The relevant term in the BRST current is $\oint \op{Tr} I \c J $.  As we have seen, at leading order in the $1/N$ expansion, the OPE between $\op{Tr} (I\c J) $ and an open-string operator is dominated by the term which has a single Wick contraction between open-string fields, yielding again an open-string operator.  This term we have already accounted for in our computation of open-string BRST cohomology.

The OPE of $\op{Tr}( I \c J)$ with a closed-string operator is dominated by the term where we Wick contract the closed-string field $\c$ to yield an open-string field.  This can only happen if the closed-string operator contains a $\b$-ghost, and so is an element of the $B^{(n)}$ or $D^{(n)}$ towers. Applying this term in the BRST operator to an operator in the $B^{(n)}$ or $D^{(n)}$ towers will yield a $GL(k \mid k)$ invariant open-string operator, which is of the form $E^{(n)}_{\op{Id}}$.  Such operators are bosonic. We conclude that this term in the BRST operator only yields a non-zero anser when applied to $B^{(n)}$.  Since the operator $B^{(n)}$ contains a single $\b$-ghost, and has no space-time derivatives, we find
\begin{equation} 
	Q_{BRST} B^{(n)} = E^{(n)}_{\op{Id}}. 
\end{equation}
Thus, the identity component of the $E^{(n)}$ tower cancels with the the $B^{(n)}$ tower.  

It is interesting to note that beyond the planar limit, there is a further cancellation.  The application of the BRST operator $\op{Tr} \oint I \c J$ to $E^{(n)}_{\mf{t}}$ has a sub-leading term which involves two Wick contractions of matter fields $I,J$. The result of this Wick contraction is the fermionic tower $C^{(n)}$, times the trace of $\mf{t}$: 
\begin{equation} 
	Q_{BRST} E_{\mf{t}}^{(n)} = N^{-1} \op{Tr}(\mf{t}) C^{(n)} . 
\end{equation}
Thus, beyond the planar limit, both the towers $B^{(n)}$, $C^{(n)}$, whose lowest-lying states are super-conformal generators, cancel with the open-string operators; and the open-string current algebra is reduced  \footnote{The case of $k=1$ is somewhat special. The resulting algebra is suspiciously similar to $\mc{A}_N$, as in this case the $E$-tower of oprators has the same quantum numbers as the $B$- and $C$-towers in the $\mc{A}_N$ algebra. This fact has no obvious four-dimensional interpretation, but it may have an interesting physical interpretation if the chiral algebra is constructed as a boundary vertex algebra for a three-dimensional gauge theory, as in \cite{Costello:2018fnz}, as an IR mirror symmetry. We will see a similar phenomenon in the B-model dual.} to $\mf{psl}(k \mid k)$.

We will show that the $\mathcal{A}^{k|k}_N$ chiral algebra is dual to B-model topological strings on the deformed conifold, in the presence of $k$ space-filling branes and $k$ space-filling ghost branes. This is a topological analogue of standard holographic statements involving a D3-D7 system \cite{Karch:2002sh}.

\subsection{The affine quivers}
A simple variant of the standard holographic duality involves 4d ${\mathcal N}=2$
gauge theories associated to an affine ADE quiver $\Gamma$, with $U(N)$ gauge groups at all nodes \cite{Kachru:1998ys,Lawrence:1998ja}. 
The holographic dual is expected to be AdS$_5 \times$S$^5/\Gamma$, with $\Gamma$ 
embedded in the $SU(2)$ subgroup of the $SO(6)$ isometry group 
corresponding to the decomposition $\mathbb{R}^6 \to \mathbb{C}^2 \oplus \mathbb{R}^2$. 

The chiral algebra $\mathcal{A}^{\Gamma}_N$ is now built from symplectic bosons $Z^i$ associated to individual edges of the quiver. Almost by construction, we can embed $\Gamma$ in the $SU(2)_R$ acting on the $Z^i$, so that single-trace operators of the affine quiver 
are in one-to-one correspondence with operators in the affine one-node quiver which are invariant under the action of $\Gamma$. 
At the leading order in $N^{-1}$, correlation functions and OPE will essentially coincide with those of the corresponding sub-algebra of 
$\mathcal{A}_{N'}$ for an appropriate $N'$ \cite{Bershadsky:1998cb}. Dually, this corresponds to the statement that an orbifold operation 
does not affect tree-level topological string correlation functions. This supports our proposal of holographic duality with 
the B-model on SL$_2(\C)/\Gamma$.  

We can also add flavor to the affine quiver, with a small twist: we can vary the ranks of the nodes to some $N + n_a$ 
while at the same time adding $(k_a+m_a|k_a)$ flavors at each node, as long as the $n_a$ and $m_a$ mismatches 
balance out, i.e. $\vec m= C_\Gamma \cdot \vec n$ with $C_\Gamma$ being the Cartan matrix.   

On the holographic side, addding flavour amounts to intoducing certain space-filling $B$-branes on $SL_2(C) / \Gamma$. Let us focus on type $A$ quivers, so $\Gamma = \Z/N$.  Then each node of the affine quiver corresponds to a representation of $\Z/N$, and so to a flat $U(1)$ bundle on $SL_2(\C)/(\Z/N)$.  Adding flavour to a node amounts to introducing a space-filling brane of rank $(k_a +m_a \mid k_a)$ twisted by this flat $U(1)$ bundle.  We expect that varying the rank of a node to $N+n_a$ changes the back-reacted geometry to a non-commutative deformation of $SL_2(\C) / (\Z/N)$.

We will leave the detailed analysis of this class of models to future work.

\section{A quick review of B-model topological strings}
The $B$-model topological string is defined on a Calabi-Yau $3$-fold $X$.  It has the feature that there are no world-sheet instantons, so that the entire topological string theory can be captured by a field theory on the space-time manifold $X$.  We will use exclusively this space-time description, and will not discuss the world-sheet.

Let us first describe the closed-string field theory living on the space-time manifold. This is the Kodaira-Spencer theory of gravity introduced in \cite{Bershadsky:1993ta,Bershadsky:1993cx}.   The fields of this theory are expressed in terms of poly-vector fields on $X$, which we now introduce.

We let
\begin{equation} 
	\PV^{j,i}(X) = \Omega^{0,i}(X, \wedge^j T X) 
\end{equation}
be the space of $(0,i)$ forms with coefficients in the $j$th exterior power of the holomorphic tangent bundle of $X$.  We let 
\begin{equation} 
	\PV^{\ast,\ast}(X) = \oplus_{i,j} 	\PV^{j,i}(X)  
\end{equation}
where we place $\PV^{j,i}(X)$ in cohomological degree $i+j$.  Wedging of Dolbeault forms and of sections of the exterior powers of $T X$ makes $\PV^{\ast,\ast}(X)$ into a graded-commutative algebra.

It is equipped with several additional operations. The first is the $\dbar$-operator
\begin{equation} 
	\dbar : \PV^{j,i}(X) \to \PV^{j, i+1}(X) 
\end{equation}
which is the usual $\dbar$ operator on Dolbeault forms with coefficients in a holomorphic vector bundle. The operator $\dbar$ is a derivation for the wedge product on $\PV^{\ast,\ast}(X)$. 

Next, we have an operator 
\begin{equation} 
	\partial : \PV^{j,i}(X) \to \PV^{j-1,i}(X). 
\end{equation}
This is defined by observing that there is an isomorphism
\begin{align} 
	\PV^{j,i}(X) & \iso \Omega^{3-j,i}(X) \\
	\alpha & \mapsto \alpha \vee \Omega_X
\end{align}
where $\Omega_X$ is the holomorphic volume form.  Using this isomorphism, the $\partial$ operator on the de Rham complex becomes the operator  we call $\partial$ acting on polyvector fields.

In local coordinates, the algebra of polyvector fields can be written as 
\begin{equation} 
	\PV^{\ast,\ast}(\C^3) = \cinfty(\C^3)[\eta^i, \d \zbar_i].  
\end{equation}
It is obtained by adjoining to the ring $\cinfty(\C^3)$ of smooth functions on $\C^3$ odd variables $\d \zbar_i$ and $\eta^i$, both of ghost number $1$.  These odd variables both live in the fundamental of $SU(3)$.  We have denoted the elements of $\PV^{0,1}$ by $\eta^i$ rather than by $\partial_{z_i}$ so as not to cause confusion. 

The operators $\dbar$ and $\partial$ are given by the expressions
\begin{align} 
	\dbar &= \sum \d \zbar_i \partial_{\zbar_i} \\
	\partial &= \sum\frac{\partial}{\partial \eta^i } \partial_{z_i}   
\end{align}
Because the isomorphism
\begin{equation} 
	\PV^{\ast,\ast}(X) \iso \Omega^{\ast,3-\ast}(X) 
\end{equation}
does not respect the algebra structure present on both sides, the operator $\partial$ is not a derivation for the wedge product on $\PV^{\ast,\ast}(X)$. Instead, we have the identity
\begin{equation} 
	\partial(\alpha \wedge \beta) = (\partial \alpha) \wedge \beta + (-1)^{\abs{\alpha}} \alpha \wedge \partial \beta + \{\alpha,\beta\} 
\end{equation}
where $\{-,-\}$ is an operation known as the Schouten bracket.   

In local coordinates, if $\alpha,\beta \in\cinfty(\C^3)[\eta^i, \d \zbar_i]$, we have
\begin{equation} 
	\{\alpha,\beta\} = (\partial_{\eta^i} \alpha)(\partial_{z_i} \beta) + (-1)^{\abs{\alpha}} (\partial_{z_i} \alpha) (\partial_{\eta^i} \beta).    
\end{equation}

The final operation we need on polyvector fields is the integration map.  This is defined as the composition
\begin{equation} 
	\int : \PV^{3,3}(X) \xto{-\vee \Omega_X} \Omega^{0,3}(X) \xto{- \wedge \Omega_X} \Omega^{3,3}(X) \xto{\int} \C. 
\end{equation}
We extend this integration map to all polyvector fields, with the convention that it is zero except on $\PV^{3,3}(X)$.  In local coordinates, this integration map is defined by integrating the function that is the coefficient of $\d \zbar_1 \d \zbar_2 \d \zbar_3 \eta^1 \eta^2 \eta^3$.

The Kodaira-Spencer theory on a Calabi-Yau three-fold $X$ has fields
\begin{equation} 
	\op{Ker} \partial \subset \oplus \PV^{i,j}(X)[2-i-j]. 
\end{equation}
The symbol $[2-i-j]$ indicates that elements of $\PV^{i,j}(X)$ are placed in ghost number $i+j-2$.  The fields are those polyvector fields which are killed by the operator $\partial$.

This desribes the fields in the BV formalism, including antifields, ghosts, etc.  The fields of ghost number zero consist of those elements in $\PV^{0,2}(X)$, $\PV^{1,1}(X)$ and $\PV^{2,0}(X)$ which are in the kernel of the operator $\partial$.

Elements of $\PV^{1,1}(X)$ are Beltrami differentials which describe deformations of $X$ as an almost complex manifold. Beltrami differentials which are $\dbar$-closed correspond to integrable deformations of the complex structure.  Those $\dbar$-closed Beltrami differentials which are in the kernel of $\partial$ correspond to deformations of $X$ as a complex manifold equipped with a holomorphic volume form. 

Fields in $\PV^{2,0}(X)$ are holomorphic Poisson tensors, which describe deformations of $X$ into a non-commutative complex manifold.  Elements of $\PV^{0,2}(X)$ are interpreted as deformations of the trivial gerbe on $X$.

The Lagrangian is \cite{Bershadsky:1993ta}
\begin{equation} 
	\tfrac{1}{2}\int_X  (\partial^{-1} \alpha) \dbar \alpha + \tfrac{1}{6} \int_X \alpha^3. 
\end{equation}
The equations of motion are
$$
\dbar \alpha + \tfrac{1}{2} \{\alpha,\alpha\} = 0. 
$$
For instance, if $\alpha \in \PV^{1,1}(X)$, the equations of motion tell us that $\alpha$ gives an infinitesimal deformation of $X$ as a Calabi-Yau manifold.   

It is unfortunate that the fields of the Kodaira-Spencer theory are required to satisfy the differential equation $\partial \alpha = 0$, and that the action has a kinetic term which contains $\partial^{-1}$.  However, the theory has a perfectly nice propagator which allows one to analyze the theory in perturbation theory.  

The propagator is defined as follows. Choose a K\"ahler metric on $X$. Let $\tr$ be the corresponding Hodge-de Rham Laplacian on $\PV^{i,j}(X)$, defined by $\tr = [\dbar, \dbar^\ast]$.  Define a Green's kernel
\begin{equation} 
	G(z,z') \in \PV^{3,3}(X \times X) 
\end{equation}
by asking that, for all $\beta \in \PV^{i,j}(X)$, we have
\begin{equation} 
	\tr_{z} \int_{z'} G(z,z') \wedge \beta(z')  = \beta(z) - \Pi \beta(z) 
\end{equation}
where $\Pi$ is the projection onto the harmonic polyvector fields.  
Then, the propagator is
\begin{equation} 
	P = \partial_z \dbar^\ast_z G(z,z') \in \PV^{2,2}(X \times X). 
\end{equation}
Using the facts that $\partial^2 = 0$ and $(\partial_z + \partial_{z'})G(z,z') = 0$, we see that $\partial_z P = 0$, $\partial_{z'} P = 0$.  

In local coordinates $z_i$ on flat space, the propagator is
\begin{multline} 
	P = \frac{36}{\pi^3} \eps_{ijk} \eps_{abc} (\d \zbar_i - \d \zbar'_i)(\d \zbar_j - \d \zbar'_j)(\zbar_k - \zbar'_k)(\eta_a - \eta'_a) 	(\eta_b - \eta'_b) (\zbar_c - \zbar'_c) \norm{z - z'}^{-8}.  
\end{multline}
Here, as before, we are viewing the space of polyvector fields as being generated over $\cinfty(\C^3)$ by the odd variables $\eta_i$, $\d \zbar_i$. The variable $\eta_i$ corresponds to $\partial_{z_i}$. 

The theory can be analyzed in perturbation theory using this propagator and the $\alpha^3$ interaction. Using a slight variant of Kodaira-Spencer theory, in which the requirement that the field satisfies $\partial \alpha = 0$ is imposed homologically, perturbation theory was analyzed in \cite{CosLi15, Costello:2016mgj}. There, it was shown that despite being non-renormalizable, all counter-terms for the theory are fixed at the quantum level by the quantum master equation and the requirement that a consistent coupling to holomorphic Chern-Simons theory exists. 

\subsection{Coupling to holomorphic Chern-Simons theory}
Kodaira-Spencer theory can be coupled to holomorphic Chern-Simons theory, and this plays an important role in our story.  Holomorphic Chern-Simons theory for the group $\mf{gl}_K$ is the theory on 
$K$ space-filling branes in the topological $B$-model.   The fundamental field of this theory is a partial connection
\begin{equation} 
A \in \Omega^{0,1}(X) \otimes \mf{gl}_K 
 \end{equation}
and the Lagrangian is
\begin{equation} 
\int \Omega_X CS(A) 
 \end{equation}
 where $\Omega_X$ is the holomorphic volume form on $X$ and $CS(A)$ is the Chern-Simons $3$-form.  This action is invariant under gauge transformations which at the infinitesimal level take the form
 \begin{equation} 
 A \mapsto A + \eps \dbar \c + \eps [\c,A]. 
  \end{equation}

Kodaira-Spencer theory can be coupled to holomorphic Chern-Simons theory.  The components 
\begin{align} 
\alpha^{2,0} &\in \PV^{2,0}(X) \\
\alpha^{1,1} & \in \PV^{1,1}(X) \\
\alpha^{0,2} & \in \PV^{0,2}(X) 
 \end{align}
are, to leading order, coupled to the fields of holomorphic Chern-Simons theory by the terms
\begin{align} 
& \tfrac{1}{3} \Omega_X \alpha^{2,0}_{ij} \op{Tr} A \partial_{z_i} A \partial_{z_j} A \\
& \tfrac{1}{2}\Omega_X \alpha^{1,1}_{i j } \op{Tr} A \d \zbar_i \partial_{z_j} A \\
& \Omega_X \alpha^{0,2}_{ij}\op{Tr} A \d \zbar_i \d \zbar_j \label{formulae_coupled} 
 \end{align}
There are complicated higher-order terms which involve both $\alpha$ and the open-string fields, but we will not need them for our analysis.

More generally, one can consider $K$ branes and $L$ anti-branes, giving rise to holomorphic Chern-Simons for the supergroup $\mf{gl}(K \mid L)$.  It turns out that the only situation in which the coupled holomorphic Chern-Simons theory and Kodaira-Spencer theory can be consistent at the quantum level is when $K = L$.

Let us explain why this is the case. A standard descent argument tells us that anomalies\footnote{For a theory like holomorphic Chern-Simons, not all anomalies can be described in this way. This does, however, describe all anomalies that can appear at one-loop.} for a gauge theory in $6$ real dimensions can be written as quartic polynomials in the field strength.  The anomaly to quantizing holomorphic Chern-Simons theory on flat space is of the form 
\begin{equation} 
	\op{Tr}_{Ad} F^4 =	\sum_{r+s = 4} (-1)^r \op{Tr}_{Fun} F^r \op{Tr}_{Fun} F^s . 
\end{equation}
On the right hand side, traces are taken in the fundamental representation.

It was shown in \cite{CosLi15} that a version of the Green-Schwartz mechanism for the topological string cancels all terms in the anomaly except that from $r = 0$ or $s = 0$. This term is $(K-L) \op{Tr}_{Fun} F^4$, which only cancels if $K = L$.   

Because of this, we will focus on the case when we have the same number of space-filling branes and anti-branes. 

\subsection{An alternative formulation of Kodaira-Spencer theory, and a hidden $SL_2(\C)$ symmetry}
\label{sec:KS_alternative}
The fundamental fields of Kodaira-Spencer theory satisfy the differential equation $\partial \alpha = 0$.  For the field $\alpha^{2,0}$, this implies that locally there is some $\gamma \in \PV^{3,0}(X)$ such that $\alpha^{2,0} = \partial \gamma$. Further, $\gamma$ is unique up to the addition of an anti-holomorphic function on $X$; and the equations of motion for $\alpha^{2,0}$ imply that this anti-holomorphic function must be constant. \footnote{We can easily convince ourselves that 
such constant shifts for $\gamma$ indeed happen: a topological brane of complex co-dimension $1$ will source an 
$\alpha^{2,0}$ with a first order pole at the location of the brane of quantized residue. This gives $\gamma$ a logarithmic monodromy.} 

We conclude that we can take our fundamental fields to be
\footnote{The reader may be concerned about the fact that $\gamma$ is defined only locally, and only up to some constant shifts. 
Within perturbation theory, this is definitely not a concern and one can readily verify that replacing $\alpha^{2,0}$ 
with $\gamma$ will not change the propagator. Beyond perturbation theory, we may observe that already in the original KS theory  the field $\alpha^{0,2}$ has a suspicious nature: it controls a gerbe, or B-field, which is not a true 2-form, but rather a 2-form $U(1)$ connection.} $\gamma, \alpha^{1,1}, \alpha^{0,2}$ . In this presentation, the action takes the form
\begin{equation} 
 \tfrac{1}{2}\int_X  (\partial^{-1} \alpha^{1,1}) \dbar \alpha^{1,1} + \tfrac{1}{6} \int_X (\alpha^{1,1})^3
 + \int_X \gamma \dbar \alpha^{0,2} + \int_X \partial(\gamma) \alpha^{1,1} \alpha^{0,2}. \label{lagrangian_gamma} 
 \end{equation}
Recall that in the BV formalism, the field $\alpha^{0,2}$ is the $(0,2)$ component of a field $\alpha^{0,\ast}  \in \Omega^{0,\ast}(X)[2]$.  Similarly, $\gamma$ -- which we view as an element of $\Omega^{0,0}(X)$ -- is the $(0,0)$ component of a field which we call $\gamma^{0,\ast} \in \Omega^{0,\ast}(X)$.  The Lagrangian \eqref{lagrangian_gamma} has a $U(1)$ action under which $\gamma$ has charge $1$ and $\alpha^{0,2}$ has charge $-1$. The corresponding multiplets $\gamma^{0,\ast}$, $\alpha^{0,\ast}$ also have charges $1$ and $-1$ under this $U(1)$ action.  

We can use this $U(1)$ action to twist: we will shift the ghost number of the fields, and the fermion number, by the $U(1)$ charge.  If we do this, we find that the fields of ghost number zero become $\alpha^{1,1}$, as before, together with two fermionic $(0,1)$ forms, $\gamma^{0,1}$ and $\alpha^{0,1}$.  The Lagrangian is
\begin{equation} 
 \tfrac{1}{2}\int_X  (\partial^{-1} \alpha^{1,1}) \dbar \alpha^{1,1} + \tfrac{1}{6} \int_X (\alpha^{1,1})^3
 + \int_X \Omega_X \gamma^{0,1} \dbar \alpha^{0,1} + \int_X \Omega_X \gamma^{0,1} \alpha^{1,1}_{i,j} \d \zbar_i \partial_{z_j} \alpha^{0,1}.\label{lagrangian_psu} 
 \end{equation}
We can view $\alpha^{0,1}$ and $\gamma^{0,1}$ as together being a gauge field $A$ for holomorphic Chern-Simons theory for the Abelian superalgebra $\Pi \C^2$. 

The field $\alpha^{1,1}$ is the fundamental field of a theory described in \cite{CosLi15} as $(1,0)$ Kodaira-Spencer theory\footnote{The terminology is because the field content of $(1,0)$ Kodaira-Spencer theory is related to that of a $(1,0)$ tensor multiplet in $6$ dimensions, whereas ordinary Kodaira-Spencer theory is related to a $(2,0)$ tensor multiplet.} .  This is simply the theory obtained by dropping the fields in $\PV^{0,2}$ and $\PV^{2,0}$ from usual Kodaira-Spencer theory. 

What we have shown is that ordinary Kodaira-Spencer theory is equivalent to $(1,0)$ Kodaira-Spencer theory coupled to holomorphic Chern-Simons theory for the Abelian Lie algebra $\Pi \C^2$. This Lie algebra is purely fermionic, and has inner product given by the symplectic form on $\C^2$.  (One can also interpret holomorphic Chern-Simons for the Lie algebra $\Pi \C^2$ as being holomorphic Rozansky-Witten theory with target the holomorphic symplectic manifold $\C^2$).   

In this formulation, a little care is need with the gauge symmetry.  There are gauge transformations $\alpha^{1,0}$, $\gamma^{0,0}$ and $\alpha^{0,0}$ (the latter two of which are fermionic) under which the fields transform as 
\begin{align} 
	\delta \gamma^{0,1} &= \dbar_{\alpha^{1,1}}  \gamma^{0,0} + \alpha^{1,0}_{z_i} \partial_{z_i}  \gamma^{0,1} . \\
	\delta \alpha^{0,1} &= \dbar_{\alpha^{1,1}}  \alpha^{0,0} + \alpha^{1,0}_{z_i} \partial_{z_i}  \alpha^{0,1} . \\
	\delta \alpha^{1,1} &= \dbar_{\alpha^{1,1}} \alpha^{1,0} + \eps_{ijk} \partial_{z_j}  \alpha^{0,1} \partial_{z_k}  \gamma^{0,0}\partial_{z_i} -    \eps_{ijk} \partial_{z_j}  \alpha^{0,0} \partial_{z_k}  \gamma^{0,1}\partial_{z_i}    
	\label{eqn_gauge_transformations_KS}
\end{align}
These gauge symmetries are written in local coordinates in which the holomorphic volume form is $\d z_1 \d z_2 \d z_3$.  These formulae imply that the generators of the gauge transformations $\alpha^{0,0}$ and $\gamma^{0,0}$ must have non-trivial commutation relations, which yield the gauge transformation in $\alpha^{1,0}$.  The formula is
\begin{equation} 
	[\alpha^{0,0}, \gamma^{0,0} ] = \eps_{ijk} \partial_{z_j} \alpha^{0,0} \partial_{z_k} \gamma^{0,0} \partial_{z_i}
\label{eqn_KS_gauge_commutator}
\end{equation}
where the right hand side is viewed as a polyvector field of type $(1,0)$ on $SL_2(\C)$.

This formulation exhibits a symmetry that was hidden in the usual formulation. The algebra $\Pi \C^2$ has an outer action of $SL_2(\C)$.  This action preserves the Lagrangian \eqref{lagrangian_psu}, and permutes the fields we denoted $\alpha^{0,1}$ and $\gamma^{0,1}$. \footnote{Non-perturbatively, we expect these fields to be both defined up to integral shifts and this $SL_2(\C)$ should be broken to $SL_2(\Z)$ duality. If we view the fields $\alpha^{0,\ast}$, $\gamma^{0,\ast}$ as describing holomorphic Rozansky-Witten theory for $\C^2$, we expect that the non-perturbative formulation will give us holomorphic Rozansky-Witten theory with target $\C^\times \times \C^\times$. This has an evident $SL_2(\Z)$ symmetry.  If we embed the B-model into Type IIB string theory in a standard way, this should coincide with the $SL_2(\Z)$ duality group of the physical theory. See also \cite{Nekrasov:2004js,Dijkgraaf:2007sx}.}

\subsection{Coupling to holomorphic Chern-Simons revisited}
\label{secn_cs_coupling}
Let us now consider how this new formulation of Kodaira-Spencer theory couples to holomorphic Chern-Simons theory. One can use  (\ref{formulae_coupled}) to derive the formula for the Lagrangian.  We will write the Lagrangian in the BV formalism.  Let $A^{0,\ast} \in \Omega^{0,\ast}(X) \otimes \mf{gl}_{K \mid K}[1]$ denote the field for holomorphic Chern-Simons theory in the BV formalism, and as above we have the fields $\gamma^{0,\ast}$ and $\alpha^{0,\ast}$ which are the fields for $\Pi \C^2$ holomorphic Chern-Simons.  At the quadratic level, these are coupled by the Lagrangian
\begin{equation} 
	\lambda^{-1} \int \alpha^{0,\ast} \op{Tr} A^{0,\ast}  \label{quadratic_term}  
\end{equation}
where $\lambda$ is the string coupling constant.

We will need to be careful about the powers of the string coupling constant $\lambda$ that appear, in order to match with the large $N$ chiral algebra in the planar limit.  Any term in the Lagrangian involving only the closed-string fields $\alpha^{0,\ast}$, $\gamma^{0,\ast}$, $\alpha^{1,\ast}$ will be accompanied by $\lambda^{-2}$.  Any term which involves open-string fields $A^{0,\ast}$, or a mixture of open-string and closed-string fields, will be accompanied by $\lambda^{-1}$.  

Because of this, the BV odd symplectic form for closed string fields will have a coefficient of $\lambda^{-2}$, whereas the open-string odd symplectic form  will involve $\lambda^{-1}$.   

Let us compute the linearized BRST operator $Q$ on the space of fields. In the BV formalism, this has the feature that $\tfrac{1}{2} \ip{\phi ,Q \phi}$ is the quadratic term in the action, where $\ip{-,-}$ is the BV odd symplectic pairing.  This is just the statement that the quadratic term in the action functional is the Hamiltonian function which generates the linear operator $Q$. 

The operator $Q$ is the $\dbar$ operator on each multiplet $\alpha^{1,\ast}$, $A^{0,\ast}$, $\alpha^{0,\ast}$, $\gamma^{0,\ast}$; together with two more terms, coming from the quadratic term \eqref{quadratic_term} in the Lagrangian.  The extra two terms are  

This quadratic term contributes a differential where
\begin{align} 
	Q \alpha^{0,\ast} &= \alpha^{0,\ast} \op{Id} \in \Omega^{0,\ast}(X) \otimes \mf{gl}(k \mid k) \\
        Q A^{0,\ast} &= \lambda \op{Tr} A^{0,\ast}  . 
\end{align} 
The first term maps a $\alpha^{0,\ast}$ closed string field to an open string fields, and the second an open string field to a $\gamma^{0,\ast}$ to closed string field. The powers of $\lambda$ are necessary so that these operators contribute the term \eqref{quadratic_term}, bearing in mind that the BV symplectic form on closed string fields has a coefficient of $\lambda^{-2}$ and on open string fields has a coefficient of $\lambda^{-1}$.

We can interpret these new terms in the BRST operator as saying that we are considering holomorphic Chern-Simons for a super Lie algebra $\Pi \C^2 \oplus \mf{gl}(K \mid K)$ with a differential.  Let us take a basis $a_i$ for $\Pi \C^2$ so that the field $\alpha^{0,\ast}$ corresponds to $a_1$ and $\gamma^{0,\ast}$ corresponds to $a_2$, and let $M$ denote an element of $\mf{gl}(K \mid K)$. The differential on  $\Pi \C^2 \oplus \mf{gl}(K \mid K)$ takes the form
\begin{align} 
	Q ( r_i a_i  )&= \op{Id} r_1 \in \mf{gl}(K \mid K) \\
	Q ( \sum m_{ij} M_{ij} ) &= \lambda (\sum \pm m_{ii} ) a_2 
\end{align}
If we set the string coupling constant $\lambda$ to zero, then the cohomology of this differential Lie super-algebra is $\mf{pgl}(K \mid K) \oplus \Pi \C$, where $\mf{pgl}(K \mid K)$ is the quotient of $\mf{gl}(K \mid K)$ by the identity matrix.   This implies that holomorphic Chern-Simons for this differential super-algebra is equivalent to holomorphic Chern-Simons for $\mf{pgl}(K \mid K)\oplus \C$.  The equivalence is realized by integrating out the fields that cancel by the new BRST operator. 

We conclude that, at $\lambda = 0$, the Kodaira-Spencer theory coupled to $\mf{gl}(K \mid K)$ holomorphic Chern-Simons theory is equivalent to $(1,0)$ Kodaira-Spencer theory, coupled to $\mf{pgl}(K \mid K)$, holomorphic Chern-Simons theory, together with a single additional field $\gamma^{0,\ast}$. There is differential of order $\lambda$ which cancels $\gamma^{0,\ast}$ with the trace of $A^{0,\ast} \in \Omega^{0,\ast}(X, \mf{pgl}(K \mid K))$.

There are further terms in the Lagrangian involving $\gamma^{0,\ast}$. The first such term is
\begin{equation} 
	\gamma^{0,\ast} \op{Tr} (\partial A^{0,\ast})^3. 
\end{equation}
This comes from the term involving $\alpha^{2,0}$ in \eqref{formulae_coupled}, recalling that $\partial \gamma^{0,0} = \alpha^{2,0}$.   It turns out that terms such as this of order $>3$ in the fields will not affect our holographic computations \footnote{They will potentially modify the holographic global symmetry algebra into a $L_{\infty}$ algebra, and similarly they may modify the holographic chiral algebra by higher homotopical operations. After taking BRST cohomology, these higher homotopical operations will not be visible.   }.

\section{The B-brane backreaction}

The starting point of our analysis of holography is a stack of $N$ topological B-branes wrapping a complex line $C \equiv \C$ in the $\C^3$ Calabi-Yau manifold. 
We choose coordinates $z,w_1,w_2$ so that the brane wraps the curve $w_i = 0$ and the holomorphic three-form is 
\begin{equation}
\Omega_{\C^3} = \d z \d w_1 \d w_2.
\end{equation} 
The world-volume theory on the stack of branes is precisely the gauged $\beta \gamma$ system which defines the ${\mathcal A}_N$ chiral algebra. 
This is a well known fact \cite{Witten:1992fb}, which we will review in Appendix \ref{app:chiral}.

In this section we aim to compute the back-reaction of such a stack of branes. A brane supported on a complex sub-manifold $C$ generally adds a localized source to the Kodaira-Spencer equation of motion:
\begin{equation} 
	\dbar \alpha + \tfrac{1}{2}\{\alpha,\alpha\} + N \delta_{C} = 0. \label{eqn_source} 
\end{equation}
Here we are treating $\delta_C$ as an element of $\PV^{1,2}(X)$, using the isomorphism between $\PV^{1,2}(X)$ and $\Omega^{2,2}(X)$. 
This equation tells us that the obstruction to $\alpha$ satisfying the Maurer-Cartan equation defining an integrable deformation of complex structure of $X$ is given by $-N\delta_C \in \PV^{1,2}$. We refer the reader to Appendix \ref{app:source} for more details about this. \footnote{There may be more than one solution to equation \eqref{eqn_source}. The moduli of such solutions is simply the moduli of solutions to the equations of motion of KS theory in the presence of the brane.  However, we would expect that, in nice situations, there will be a unique solution to the equation \eqref{eqn_source} which is compatible with the symmetries of the situation and with any boundary conditions present.}

Let us see how this works when $X = \C^3$ and $C= \C$.  We choose coordinates $z,w_1,w_2$ so that the brane wraps the curve $w_i = 0$.     
\begin{lemma} 
	In this situation, the Beltrami differential
	\begin{equation} 
	\beta =	\frac{2N}{8 \pi^2}  \frac{\wbar_1 \d \wbar_2 - \wbar_2 \d \wbar_1 } {\norm{w}^4} \partial_z 
	\end{equation}
	satisfies equation \eqref{eqn_source}. Moreover, this is the unique solution to this equation which in addition :
	\begin{enumerate} 
		\item Satisfies the gauge condition			 
\begin{equation} 
	\dbar^\ast \alpha = 0. 
\end{equation}
\item Is invariant under the $SU(2)_R$ rotating the $w_1,w_2$ complex plane.
\item Is of charge $-2$ under the $U(1)$ action under which the $w_i$ have charge $1$ and $z$ has charge $0$, and is of charge $-1$ under the $U(1)$ action where $z$ has charge $1$ and the $w_i$ have charge $0$.
	\end{enumerate}
\end{lemma}
The reason these particular $U(1)$ charge appear is because $\delta_C$, when viewed as an element of $\PV^{1,2}$ instead of $\Omega^{2,2}$, has these $U(1) \times U(1)$ charges.
\begin{proof} 
We note that $\{\beta,\beta\} = 0$, because the coefficient of $\partial_z$ in $\beta$ is independent of $z$.  Therefore equation \ref{eqn_source} reduces to the equation
\begin{equation} 
\dbar \beta = - N \delta_C. 
 \end{equation}
 Let us treat $\beta$ as a $(2,1)$-form
 \begin{equation} 
 \beta' = \beta \vee \Omega_{\C^3} =  \frac{2N  }{8 \pi^2} \d w_1 \d w_2  \frac{\wbar_1 \d \wbar_2 - \wbar_2 \d \wbar_1 }{\norm{w}^4} . 
  \end{equation}
  It is easy to check that $\d \beta' = 0$ away from $w_i = 0$.  Therefore $\d \beta'$ can be expressed in terms of $\delta_C$ and its derivatives. Since $\beta'$ is of dimension $0$ under scaling the $w_i$ and $\br{w}_i$, we find that $\d \beta'$ can not involve any normal derivatives of $\delta_C$.  Therefore, $\d \beta'$ must be some multiple of $\delta_C$.  It remains to calculate which multiple, which is equivalent to calculating
  \begin{equation} 
  \int_{\norm{w} = 1} \beta'. 
   \end{equation}
   When restricted to the $3$-sphere $\norm{w} = 1$, $\beta'$ agrees with
   \begin{equation} 
    \beta''  = \frac{2N  }{8 \pi^2} \d w_1 \d w_2 \left( \wbar_1 \d \wbar_2 - \wbar_2 \d \wbar_1 \right).
    \end{equation}
    By Stoke's theorem, the integral of $\beta''$ over $S^3$ is the same as the integral of 
    \begin{equation} 
    \d \beta'' = \frac{2 N}{8 \pi^2} 2 \d w_1 \d w_2 \d \wbar_1 \d \wbar_2  
     \end{equation}
over the $4$-ball.  Since
\begin{equation} 
\d w_1 \d w_2 \d \wbar_1 \d \wbar_2 = 4 \d x_1 \d x_2 \d x_3 \d x_4 
 \end{equation}
we find
\begin{equation} \int_{\norm{w} \le 1} \d \beta'' = \frac{2 N}{ \pi^2} \op{Vol}(B^4) = N.  
\end{equation}

It is not difficult (or all that important) to verify that there are no other solutions.  

\end{proof}

\subsection{The deformed geometry}
Now we have determined the Beltrami differential sourced by the sourced by the brane wrapping $\C$ inside $\C^3$. This Beltrami differential defines a new complex structure on $\C^3 \setminus \C$. We will see that the new complex structure makes $\C^3 \setminus \C$ into the deformed conifold $SL_2(\C)$.   In what follows we will absorb the factors of $2 \pi$ in the Beltrami differential into the normalization of the holomorphic volume form on $SL_2(\C)$. 

To describe the complex manifold in the deformed complex structure, we need to determine the algebra of holomorphic functions. A function $F(z,\zbar,w_i,\wbar_i)$ is holomorphic in the new complex structure if it satsfies the differential equations
\begin{align} 
 \partial_{\zbar} F &= 0\\ 
	\partial_{\wbar_1} F - 	N \frac{ \wbar_2 } {\norm{w}^4} \partial_z F &= 0 \\
\partial_{\wbar_2} F + 	N  \frac{ \wbar_1 } {\norm{w}^4} \partial_z F &= 0 . 
\end{align}
If $F$ is a polynomial in $w_1,w_2$ then it satisfies these equations and defines a holomorphic function.  There are two other basic solutions to this equation given by
\begin{align} \label{eq:deformation}
	u_1 &= 	w_1 z - N  \frac{\wbar_2}{\norm{w}^2} \\
u_2 &= 	w_2 z + N  \frac{\wbar_1}{\norm{w}^2}.   
\end{align}
All holomorphic functions can be expressed in terms of $u_i,w_i$, and these satisfy the relation
\begin{equation} 
	u_2 w_1-u_1 w_2 = N 
\end{equation}

This is the equation defining the deformed conifold $X_N$ as a submanifold of $\C^4$. Notice that  \ref{eq:deformation} gives a diffeomorphism between the deformed conifold 
and $\C^3 \setminus \C$. We conclude that the deformed conifold is the correct back-reacted geometry. The natural holomorphic three-form on the deformed conifold is 
\begin{equation}
\Omega_{X_N}  |_{w_1 \neq 0} = \frac{\d u_1 \d w_1 \d w_2}{w_1}  \qquad \qquad \Omega_{X_N}  |_{w_2 \neq 0} = \frac{\d u_2 \d w_1 \d w_2}{w_2} 
\end{equation}
This is also the image of $\Omega_{\C^3}$ under the deformation. 

The deformed conifold has a single compact three-cycle, a three-sphere. One can check that the period of $\Omega_{X_N}$ on the compact three-cycle is $N$,
as expected. 

\section{Holography and the deformed conifold}

The deformed conifold has a much larger symmetry group than the original $\C^3 \setminus \C$. If we collect the coordinates into a single matrix
\begin{equation}
g = \begin{pmatrix}w_1 & w_2 \cr u_1 & u_2  \end{pmatrix}
\end{equation}
then the deformed conifold equation is $\det g = N $. We claim that $\Omega_{X_N}$ is invariant under an $SL(2,\C)_L \times SL(2,\C)_R$ 
symmetry acting on $g$ by matrix multiplication from the left and from the right. Indeed, we can write \footnote{For later reference, we record 
another expression which only makes the $SL(2,\C)_R$ symmetry manifest:
\begin{equation}
\Omega_{X_N} = \frac{\left(u_2 \d u_1 - u_1 \d u_2\right) \d w_1 \d w_2}{u_2 w_1- u_1 w_2}
\end{equation}
}
\begin{equation}
\Omega_{X_N} = N \Tr (g^{-1} dg)^3
\end{equation}
The $SL(2,\C)_L$ symmetry group includes the original Lorentz group symmetry along the D-brane worldvolume and enhances it to the full two-dimensional conformal 
group. This is a striking confirmation that we are on the right track: the back-reacted geometry has the correct symmetry enhancement to describe the holographic dual of
a two-dimensional conformal field theory. The $SL(2,\C)_R$ symmetry complexifies holomorphic rotational symmetry in the plane perpendicular to the D-brane world-volume. 

In order to define the Kodaira-Spencer theory on the deformed conifold, we will need to pick a Kahler metric to use in the gauge-fixing condition and the definition of the  
propagator. In order to avoid the risk of anomalies, it would be best to pick a gauge-fixing condition which preserves as many of the desired symmetries as possible. 
There is a very nice choice of K\"ahler form which preserves precisely the $SL(2,\C)_L \times SU(2)_R$ symmetries we need for our holographic duality proposal: 
\begin{equation}
\omega_{X_N} = \Tr (g^{-1} dg)(g^{-1} dg)^\dagger
\end{equation}

The action of $SU(2)_R$ on the deformed conifold has orbits which are a particularly nice family of non-trivial $S^3$ representatives for the basic compact three-cycle. 
In particular, we can look at the representative defined by the equations $u_2 = \bar w_1$, $u_1 = - \bar w_2$, $\norm{w}^2 = N$. 
This is the same as setting $g = N^{1/2}  \gamma$ with $\gamma \in SU(2)$. Then $\Omega_{X_N}$ is proportional to the standard volume form on 
$SU(2)$ and the period can be readily computed.

The quotient $\frac{X_N}{SU(2)_R}$ is precisely hyperbolic space, i.e. global Euclidean $AdS_3$: the projection map $g \to \rho \equiv g g^\dagger$
maps the deformed conifold to the space of Hermitean matrices of fixed positive determinant and positive trace, 
which presents hyperbolic space as a connected component of a space-like hyperboloid in $R^{1,3}$. It is easy to verify that this is precisely the 
$AdS_3$ which is relevant for holography, built from the world-volume coordinate $z$ and the radial direction away from the brane: 
\begin{equation}
\rho = g g^\dagger= \begin{pmatrix}\norm{w}^2 & \bar z \norm{w}^2 \cr z \norm{w}^2 & \frac{N^2}{16 \pi^4} \frac{1}{\norm{w}^2} + \norm{z}^2 \norm{w}^2  \end{pmatrix}
\end{equation}
In Poincare' coordinates $(z, \bar z, y = \norm{w}^{-2})$, the boundary of $AdS_3$ lies at $y \to 0$. 

\subsection{Kaluza-Klein reduction and higher spin theories}
In the remainder of the paper, we will do calculations in the full three-dimensional complex geometry. 
In this section we will briefly  discuss an alternative, physically instructive perspective: the $S^3$ Kaluza-Klein reduction of the Kodaira-Spencer of holomorphic Chern-Simons theories on the deformed conifold.  We leave a detailed analysis of this approach to future work.  

The result of this reduction will be some theory on AdS$_3$ with an infinite tower of fields of increasing spin. 
In order for our holographic proposal to work, these fields should have a very special Chern-Simons-like action, 
so that appropriate boundary conditions at $y \to 0$ will result in the appearance of a boundary chiral algebra
with a structure compatible with that of ${\mathcal A}_N$ or ${\mathcal A}^{(k|k)}_N$.

Let us first  consider the lowest Kaluza-Klein mode for an holomorphic $\mathfrak{psu}(k|k)$ Chern-Simons theory.  
We parameterize this mode of the connection as 
\begin{equation}
A = a_i(\rho) \Tr \sigma^i \rho^{-1} \bar \partial \rho = a_i(\rho) \Tr \sigma^i (g^\dagger)^{-1} \bar \partial g^\dagger
\end{equation}
where $\sigma^i$ are the three Pauli matrices. Then flatness of $A$ is equivalent to 
flatness of the AdS$_3$ connection $a^{(0)} = a_i(\rho) \Tr \sigma^i \rho^{-1} d \rho$. The 
holomorphic Chern-Simons action evaluates to the standard Chern-Simons action for $a^{(0)}$,
with level proportional to $N$. 

Standard WZW boundary conditions for Chern-Simons gauge theory 
on AdS$_3$ will result into a $\widehat{\mathfrak{pu}}(k|k)$ boundary chiral current algebra at level $N$.
This is precisely the correct central charge for the simplest open string operator in the chiral algebra, the traceless part of the $IJ$ currents.
The boundary condition sets to zero $i_{\partial_{\bar z}} a^{(0)}$ at the boundary. In terms of the six-dimensional connection 
this sets to zero the component of $A$ along the $SL(2,\C)_L$ generator $\bar w_1 \partial_{\bar u_1} + \bar w_2 \partial_{\bar u_2}$.
We will later propose the full holomorphic generalization of this zero-mode boundary condition.

The analysis of the KK modes of Kodaira-Spencer theory is more challenging, and we will not give all details. We will focus on the $(1,0)$ part\footnote{The $\Pi \C^2$ part of Kodaira-Spencer theory can be 
analyzed as before, leading to $\Pi \C^2$ Chern-Simons theory on AdS$_3$ and tentatively to the chiral algebra generators associated to 
$\Tr b$ and $\Tr \partial c$. These fields are decoupled from the rest of the algebra, though, and it is not obvious they should be included. The answer may depend on non-perturbative considerations which are beyond the scope of this paper.}of the KS theory, whose fundamental field is a Beltrami differential.

We expect that the lowest-lying modes of Kodaira-Spencer are of two types. We can take Beltrami differentials which are proportional to the vector fields $R_a$, generating the $SU(2)_R$ action; or to $L_a$, generating the $SL_2(\R)$ global conformal symmetry. In the first case, we can parameterize a KK mode of the form 
\begin{equation}
\alpha = \alpha^a_i(\rho) R_a \Tr \sigma^i \rho^{-1} \bar \partial \rho 
\end{equation}
The bracket of two such vector fields reduces to the bracket $\{R_a, R_b \} = \epsilon_{abc} R_c$,
as $R_a$ acts trivially on $\rho$. The equations of motion reduce to a standard Chern-Simons action for
an $SU(2)_R$ gauge field and natural boundary conditions would lead to boundary chiral operators matching the 
$\Tr Z^{(i} Z^{j)}$ currents in the chiral algebra. 

A KK mode coming from a Beltrami differential proportional to the vector fields $L_a$ generating global conformal symmetry will be of the form 
\begin{equation}
\alpha = \beta^k_i(\rho) L_k \Tr \sigma^i \rho^{-1} \bar \partial \rho 
\end{equation}
This field should, abstractly, be another copy of $SL_2$ Chern-Simons theory, but with an ``oper'' boundary condition at which the Virasoro algebra lives.  It is tempting to think of these fields as giving a chiral $3$-dimensional gravity theory. We will leave a detailed analysis of the KK reduction to others.

\section{Global symmetry algebras}
\label{sec:global_symmetry_isomorphism}
In this section we will introduce the global symmetry algebras of the large $N$ CFT and of its holographic dual.  In each case, these are infinite-dimensional Lie algebras which act as symmetries of the theory in the planar limit.

We will show, by a combination of algebra and explicit computations of OPEs, that  holographic global symmetry algebra is isomorphic to the global symmetry algebra of the large $N$ chiral algebra.   This provides very strong constraints on the planar two- and three-point functions,  and we will see later that the global symmetry algebra is almost enough to completely fix them. 

\subsection{The algebra of global symmetries of a chiral CFT}
In any two-dimensional chiral CFT, there is a nice associative algebra of \emph{global symmetries}.  This is a sub-algebra of the algebra of modes of the theory,  consisting of all modes which annihilate the vacuum both at the origin and at infinity.   Concretely, for each operator $\mc{O}(z)$
of dimension $\Delta$, the algebra of global symmetries includes the modes $\oint z^k \mc{O}(z) \d z$ for $0 \le k \le 2\Delta -2$. These are the modes with dimension ranging from $1 - \Delta$ to $\Delta - 1$.  In particular, this includes zeromodes of Kac-Moody currents and the global conformal generators inside the Virasoro algebra of modes of the stress tensor.

We call this the algebra of global symmetries because these modes act on local operators in a way which preserves correlation functions.   Suppose we have a collection of local operators $\mc{O}_i$ which we place at points $z_i$ on a cylinder. Then the action of a global symmetry $\alpha$ on a local operator $\mc{O}(z_i)$ is the same as the commutator $[\alpha,\mc{O}(z_i)]$ in the algebra of modes: it is the difference between placing the contour integral defining $\alpha$ to the left or right of $z$.  Then we have
\begin{equation} 
	\lvac \left[\alpha , \mc{O}_1(z_1), \dots, \mc{O}_n(z_n)\right] \rvac = 0  
\end{equation}
because $\alpha$ preserves the vacua at $0$ and $\infty$. The derivation property for $[\alpha,-]$ tells us 
\begin{equation} 
	\lvac \left[\alpha, \mc{O}_1(z_1), \dots, \mc{O}_n(z_n)\right] \rvac = \sum \lvac \mc{O}_1(z_1) \dots [\alpha,\mc{O}_i(z_i)] \dots \mc{O}_n (z_n) \rvac. 
\end{equation}
from which we conclude that $\alpha$ preserves correlation functions. 

The reason global symmetries form an associative algebra is the following.  If $\alpha_1, \alpha_2$ are global symmetries, so that $\lvac \alpha_i = 0$ and $\alpha_i \rvac = 0$, then $\lvac \alpha_1 \alpha_2 = 0$ and $\alpha_1 \alpha_2 \rvac = 0$.  This is a non-unital algebra, because the unit does not annihilate the vacuum.

\subsection{Global symmetries of the large $N$ chiral algebra} 
For a general CFT, the algebra of global symmetries is simply an associative algebra.  It turns out that when we work in the planar limit of the large $N$ CFT, the algebra of global symemtries is the universal enveloping algebra of a Lie algebra.  This is because we can write all modes as a product of the modes of single-trace operators.  The commutator of modes of single-trace operators yields a combination of a single trace operator and the identity operator.  The identity operator can not appear in the commutator of single-trace modes which leave invariant the vacuum at $0$ and $\infty$.    

We will let $\mf{a}_{\infty}$ denote the Lie algebra of single-trace global symmetries.  This turns out to be an infinite-dimensional Lie algebra which extends the global super-conformal algebra inside the ${\mathcal N}=4$ super-Virasoro
algebra.

The global symmetries of the large $N$ chiral algebra are given by the expressions
\begin{align} 
	\oint z^k \d z A^{(n)}(z) & \text{ for } 0 \le k \le n-2 \\ 	
	\oint z^k \d z B^{(n)}(z) & \text{ for } 0 \le k \le n \\
	\oint z^k \d z C^{(n)}(z) & \text{ for } 0 \le k \le n \\ 
	\oint z^k \d z D^{(n)}(z) & \text{ for } 0 \le k \le n+2 .  
\end{align}
The operators coming from the $B,C$ towers are fermionic, and the others are bosonic. 

The most basic bosonic symmetry generators are the $\mathfrak{su}(2)_R$ generators 
defined by the zeromodes of the dimension $1$ currents
\begin{equation}
J^{ij}_0 = \oint \d z \Tr Z^i Z^j(z)
\end{equation}
and the $\mathfrak{su}(2)_L$ global conformal generators
\begin{equation}
L_{-1} = \oint \d z T(z) \qquad L_{0} = \oint \d z z T(z)\qquad L_{1} = \oint \d z z^2 T(z)
\end{equation}

The $A^{(n)}$ operators give bosonic global symmetry generators $J^{(n),m}_r$ transforming in a finite-dimensional representation of 
spin $\frac{n}{2}$ for $\mathfrak{su}(2)_R$ and spin $\frac{n}{2}-1$ for $\mathfrak{su}(2)_L$. The $D^{(n)}$
operators give bosonic global symmetry generators $L^{(n),m}_r$ transforming in a finite-dimensional representation of 
spin $\frac{n}{2}$ for $\mathfrak{su}(2)_R$ and spin $\frac{n}{2}+1$ for $\mathfrak{su}(2)_L$.
The $B^{(n)}$ and $C^{(n)}$ operators give fermionic global symmetry generators $b^{(n),m}_r$ and $c^{(n),m}_r$ transforming 
in a finite-dimensional representation of 
spin $\frac{n}{2}$ for $\mathfrak{su}(2)_R$ and spin $\frac{n}{2}$ for $\mathfrak{su}(2)_L$.

One can compute the commutator in the global symmetry algebra using the OPE. We will perform some explicit computations later. 

\subsection{The holographic global symmetry algebra}
The holographic dual to the chiral algebra is the topological $B$-model on the deformed conifold.  The OPEs of general operators in the chiral algebra can be computed using Witten diagrams on the deformed conifold. In general, these calculations can be quite difficult to perform on the holographic side.  

The global symmetry algebra, however, has a very simple description from the holographic point of view:  the global symmetry algebra is the Lie algebra of gauge symmetries of the vacuum field configuration for Kodaira-Spencer theory on the deformed conifold.  It is intuitively clear that this Lie algebra should function as the global symmetries of the holographic chiral algebra.  We will give in section \ref{sec:modes} a more formal derivation of this statement. We will see how this algebra is given by certain modes of the holographic chiral algebra, which is defined using Witten diagrams. 

The Lie algebra $\mf{a}^{hol}$ of holographic global symmetries is very easy to describe.   The bosonic part of $\mf{a}^{hol}$ is the Lie algebra 
\begin{equation} 
	\op{Vect}_0 (SL_2(\C)), 
\end{equation}
the Lie algebra of divergence free holomorphic vector fields on the complex manifold $SL_2(\C)$.  These are the symmetries of $SL_2(\C)$ as a Calabi-Yau manifold, and hence the gauge symmetries preserving the vacuum field configuration. 

The fermionic part of $\mf{a}^{hol}$ consists of two fermionic copies of holomorphic functions on $SL_2(\C)$, $\alpha,\gamma \in \Oo(SL_2(\C))$. These are the gauge symmetries of the trivial gauge field for holomorphic Chern-Simons theory for the Lie algebra $\Pi \C^2$, which provide the remaining fields of holomorphic Chern-Simons theory. 

The fermionic part is a module over the bosonic part in an evident way, and the commutator between two fermionic elements is
\begin{equation} 
	[\alpha,\gamma] =  \partial \alpha \partial \gamma.  
\end{equation}
On the right hand side, we are identifying $(2,0)$ forms on $SL_2(\C)$ with vector fields using the holomorphic volume form.  This equation follows from the form of the commutation relations for gauge transformations in Kodaira-Spencer theory as given in equation \eqref{eqn_KS_gauge_commutator}.

Let us describe $\mf{a}^{hol}$ as a representation of $SL_2(\C)$.  The fermionic part is easily described using the Peter-Weyl theorem, which says that the algebra of polynomial functions on $SL_2(\C)$ is 
\begin{equation} 
	\Oo(SL_2(\C)) = \oplus_{j \ge 0, 2 j \in \Z} V_j^L \otimes V_j^R 
\end{equation}
where $V_j^L$, $V_j^R$ indicate the spin $j$ representation of $SL_2(\C)^L$ and $SL_2(\C)^R$.  

To describe the bosonic part, note that we can write vector fields on $SL_2(\C)$ as the tensor product of $\Oo(SL_2(\C))$ with the right-invariant vector fields, which transform in $V_1^L$.  Therefore, 
\begin{align} 
	\op{Vect}(SL_2(\C)) &= \oplus_{j} (V_1^L \otimes V_j^L) \otimes V_j^R \\
	&= \oplus_j \left(V_{j+1}^L + V_j^L + V_{j-1}^L    \right)  \otimes V_j^R. 
\end{align}
The divergence map 
\begin{equation} 
	\op{Div} : \op{Vect}(SL_2(\C)) \to \Oo(SL_2(\C))  
\end{equation}
is a map of $SL_2(\C)$ representations, and is surjective.  It follows that the divergence free vector fields are 
\begin{equation} 
	 \op{Vect}_0(SL_2(\C)) = \oplus_j \left(V_{j+1}^L +  V_{j-1}^L    \right)  \otimes V_j^R. 
\end{equation}

Including the fermionic elements, we find that the entire global symmetry algebra is
\begin{equation} 
	\mf{a}_{\infty}^{hol} =   \bigoplus_j \left(V_{j+1}^L +  V_{j-1}^L    \right)  \otimes V_j^R \oplus \bigoplus_{j}\left( \Pi \C^2 \otimes  V_j^L \otimes V_j^R \right). 
\end{equation}
This matches exactly with the large $N$ global symmetry algebra, as a representation of $SL_2(\C)_L \times SL_2(\C)_R$.

The main theorem we will prove in this section is that:
\begin{theorem} 
	There is an isomorphism of super Lie algebras
\begin{equation} 
	\mf{a}_{\infty} \iso \mf{a}_{\infty}^{hol}. 
\end{equation}
\end{theorem}
Our proof of this is a little indirect.  We will first introduce matter to both the large $N$ chiral algebra and its holographic dual, and verify the isomorphism of global symmetry algebras with matter.  We will use this to constrain the global symmetry algebras without matter sufficiently that we can prove they are isomorphic. 

Along the way, we will find a beautiful manifestation of holography: we will see how to recover the algebra $\Oo(SL_2(\C))$ of polynomial functions on the deformed conifold directly from the structure constants of the global symmetry algebra with matter.

\subsection{The holographic global symmetry algebra including matter}
We can include matter into the chiral algebra, as discussed in section \ref{subsection_matter}.  This corresponds to introducing $k$ bosonic and $k$ fermionic space-filling probe branes into the topological $B$-model. This changes the global symmetry algebra into an algebra $\mf{a}_{(k \mid k)}^{hol}$. 

This contains,  as before, the Lie algebra $\op{Vect}_0(SL_2(\C))$ of divergence-free vector fields on the deformed conifold, together with the fermionic generators $\alpha,\gamma$.  We also have open-string symmetries, which are the gauge transformations of the trivial $GL(k \mid k)$ bundle on $SL_2(\C)$.  This Lie algebra is $\Oo(SL_2(\C)) \otimes \mf{gl}(K \mid K)$.  

The bosonic part of the closed string algebra acts on the open-string algebra in an evident way.  As we saw in section \ref{secn_cs_coupling}, there is a non-trivial BRST differential which sends the fermionic closed string generator $\alpha$ to the open-string generator $\alpha \otimes \op{Id}$.  Taking the cohomology with respect to this operator cancels the coefficient of the identity in $\mf{gl}(k \mid k) \otimes \Oo(SL_2(\C))$, reducing the open-string algebra to $\mf{pgl}(k \mid k) \otimes \Oo(SL_2(\C)$.

This operation also removes the fermionic closed-string generator $\alpha$, leaving only the generator $\gamma$.

\subsection{The global symmetry algebra with matter for the large $N$ chiral algebra}
For the chiral algebra with $\mf{gl}(k \mid k)$ flavour symmetry, we also have a global symmetry algebra $\mf{a}_{(k \mid k)}^{CFT}$.  As we have seen, the chiral algebra is generated by bosonic closed-string operators $A^{(n)}$, $D^{(n)}$, fermionic closed-string operators, $B^{(n)}$, $C^{(n)}$, and the open-string operators 
\begin{equation}
E_\mathfrak{t}^{(n)} = \Tr \mathfrak{t}\, I  Z^{(i_1} Z^{i_2} \cdots Z^{i_n)} J
\end{equation}
(Here $\mf{t}$ is an element of  $\mf{gl}(k\mid k)$). This operator is of spin $n/2+1$, and so the modes 
\begin{equation} 
	\oint \d z z^k E_{\mathfrak{t}}^{(n)} 	 
\end{equation}
 defines elements of the global symmetry algebra for $k \le n$.  These modes transform in the representation of spin $n/2$ for the conformal $SL_2$ and spin $n/2$ for the $R$-symmetry $SU(2)$.  Summing over $n$, and identifying the two copies of $SL_2$ with the left and right actions of $SL_2(\C)$ on itself, we find that the open-string global symmetries transform in
 \begin{equation} 
 \mf{gl}(k \mid k) \otimes \left( \oplus V_j^L \otimes V_j^R\right) =  \mf{gl}(k \mid k) \otimes \Oo(SL_2(\C)). 
  \end{equation}
 
 The global symmetries corresponding to the generators of the chiral algebra will be denoted by $a,b,c,d,e_{\mf{t}}$, where $\mf{t} \in \mf{gl}(k \mid k)$.  Under the left and right $SL_2(\C)$ action, $a$ and $d$ transform as an element of $\op{Vect}_0(SL_2(\C))$, where as $c,d,e_{\mf{t}}$ transform as elements of $\Oo(SL_2(\C))$. We will often use the notation $c(F)$, $d(F)$, $e_{\mf{t}}(F)$ for $F \in \Oo(SL_2(\C))$ to refer to a particular global symmetry.  
 
As we saw in section \ref{subsection_matter}, at leading order in $1/N$ there is a term in the BRST operator which cancels the symmetries built from $B^{(n)}$ with the identity component of the modes of $E^{(n)}$.  This reduces the open-string part of the symmetry algebra to $\mf{pgl}(k \mid k)$ tensored with a representation of $SL_2(\C) \times SL_2(\C)$ which is easily seen to be $\Oo(SL_2(\C)$. 

This leads to a very satisfying match with the holographic global symmetry algebra.  On both sides, to leading order in $1/N$ and in the string coupling constant $\lambda$, the open-string sector of the algebra is $\mf{pgl}(k \mid k)\otimes \Oo(SL_2(\C))$ as a representation of $GL(k \mid k) \times SL_2(\C)^L \times SL_2(\C)^R$.  In the closed string sector, one of the two fermionic towers survives: the surviving tower is $c$ on the CFT side and $\gamma$ on the holographic side. 

To subleading order in $1/N$ or $\lambda$, there is a term in the BRST differential on both sides which cancels the remaining fermionic closed-string generators $c$ and $\gamma$ with open-string generators, and reduces to the open-string sector to $\mf{psl}(k \mid k)$.

\subsection{Building the deformed conifold from the global symmetry algebra}
In this section we will show that the upper-triangular part of the open-string global symmetry algebras $\mf{a}^{CFT}_{(k \mid k)}$ and $\mf{a}^{hol}_{(k \mid k)}$ are isomorphic. As a corollary, we will see how to build the deformed conifold from the global symmetry algebra $\mf{a}^{CFT}_{(k \mid k)}$. 

Choose a generic element of the real Cartan of $\mf{pgl}(k \mid k)$.  This is the same as the Cartan of its bosonic real part $\mf{su}(k) \oplus \mf{su}(k) \oplus \mf{u}(1)$. The extra copy of $\mf{u}(1)$ corresponds to the matrix which is $1$ on bosonic elements of $\C^{k \mid k}$ and $-1$ on fermionic elements. 

We let $\mf{a}_{+,(k\mid k)}^{CFT}$ be the sub-algebra of elements of positive weight under the chosen generic  element of the Cartan.  This only includes open-string operators, because closed-string operators are of weight $0$.  The open-string operators are those of the form
\begin{equation} 
\oint	\d z z^l E_{ij}^{(n)} 	  
\end{equation}
where $i < j$ in an ordering of a basis of $\C^{k \mid k}$ coming from the choice of element of the Cartan of $\mf{pgl}(k \mid k)$.  We can arrange so that if the index $i$ is bosonic and $j$ fermionic, $i < j$.

For each $i$ and $j$, these open-string global symmetries transform in the representation $\oplus V_s^L \otimes V_s^R$ of $SL_2(\C) \times SL_2(\C)$.  We will change notation slightly, and write an element of the global symmetry algebra of this form as $e_{ij}(F)$, where $F$ denotes an element of $\oplus V_j^L \otimes V_j^R = \Oo(SL_2(\C))$.

We will show how to recover the  product in the algebra $\Oo(SL_2(\C))$ from the Lie bracket on $\mf{a}^{CFT}_{+,(k \mid k)}$. 

To see this, note that if $i < j < k$, we must have
\begin{equation} 
	[e_{ij}(F), e_{jk}(G)] = e_{ik}(\mu(F,G) ) 
\end{equation}
for some bilinear $SL_2(\C) \times SL_2(\C)$-invariant map
\begin{equation} 
	\mu : \left ( \oplus V_j^L \otimes V_j^R\right)^{\otimes 2} \to \oplus V_j^L \otimes V_j^R.  
\end{equation}
(The fact that the commutator in $\mf{a}_{+, (k \mid k)}^{CFT}$ must be of this form is dicated by the $PSL(k \mid k)$ symmetries).  

The Jacobi identity tells us that 
\begin{equation} 
	[[e_{ij}(F),[  e_{jk}(G), e_{kl}(H) ]] = [[e_{ij}(F), e_{jk}(G)], e_{kl}(H)]. 	
\end{equation}
This implies that 
\begin{equation} 
	e_{il}( \mu(F, \mu(G,H) ) = e_{il}(\mu(\mu(F,G), H)). 
\end{equation}
In other words, the operator $\mu$ is an associative product on $\oplus V_j^L \otimes V_j^R$. 

We would like to show that this associative product is given by multiplication of polynomial functions on the deformed conifold.  This will show how to recover the deformed conifold from the global symmetry algebra $\mf{a}^{CFT}_{+,(k \mid k)}$. 

To do this, we should write down the open-string global symmetries which correspond to the functions $u_i, w_i$ which generate the ring of functions on $SL_2(\C)$.  These are the elements which transform in $V_{1/2}^L \otimes V_{1/2}^{R}$.  The global symmetries are given by the expressions
\begin{align} 
e_{ij}(w_1) &= \oint_z I_i Z_1 J_j z \d z\\
e_{ij}(w_2) &= \oint_z I_i Z_2 J_j z \d z\\
e_{ij}(u_1) &= \oint_z I_i Z_1 J_j   \d z\\
e_{ij}(u_2) &= \oint_z I_i Z_2 J_j   \d z\\
\end{align}

To prove that the product we have defined on $\oplus V_j^L \otimes V_j^R$ is indeed that given by multiplication of functions on the conifold, we need to prove two things.
\begin{enumerate} 
	\item The elements $u_i$, $w_i$ generate the associative algebra built from the global symmetry algebra $\mf{a}^{CFT}_{(k \mid k)}$ 
	\item The elements $u_i$, $w_i$ commute with each other in this algebra. 
	\item In the Lie algebra $\mf{a}^{CFT}_{(k \mid k)}$, we have
		\begin{equation} 
			[e_{ij}(u_1), e_{jk}(w_2)] - [e_{ij}(u_2), e_{jk}(w_1)] = e_{ik} N. \label{eqn_conifold} 
		\end{equation}
\end{enumerate}
The first statement follows immediately from equation \ref{eqn_lvl1_open_currents} in the appendix. Indeed, this equation shows that the action of $e_{ij}(u_a)$, $e_{ij}(w_b)$ on $e_{jk}(F)$ for $F \in V_n^L \otimes V_n^R$ is a sum of two terms of the form $e_{ik}(G)$, where in one term $G \in V_{n+1/2}^L \otimes V_{n+1/2}^R$ and in the other $G \in V_{n-1/2}^L \otimes V_{n-1/2}^R$.  Since the $SL_2(\C) \times SL_2(\C)$-invariant map
\begin{equation} 
V_{1/2}^L \otimes V_{1/2}^R \otimes V_n^L \otimes V_n^R \to V_{n+1/2}^L \otimes V_{n+1/2}^R 
 \end{equation}
is surjective, we conclude that the elements in $V_{1/2}^L \otimes V_{1/2}^R$ generate.

Next, we need to verify that the elements $u_i, w_i$ commute with each other.  To see this, note that the commutator of these elements must be an $SL_2(\C) \times SL_2(\C)$-equivariant map
\begin{equation} 
	\wedge^2  \left( V^L_{1/2} \otimes V^R_{1/2}\right) \to \oplus V^L_j \otimes V^R_j. 
\end{equation}
There are no such maps.

The final thing we need to verify is equation \eqref{eqn_conifold}.  Since this statement is so central to our derivation of the conifold algebra from the global symmetry algebra, let us verify it carefully.   

We need to show that the following commutation relation holds in the mode algebra:
\begin{equation}  
	\left[ \oint_z I_i Z_1 J_j   \d z, \oint_z I_j Z_2 J_k z \d z\right] -\left[ \oint_z I_i Z_2 J_j   \d z,   \oint_z I_j Z_1 J_k z \d z\right] = N \oint I_i J_k  \d z. 
\end{equation}
Factors of $\pi$ can be absorbed into an $N$-independent rescaling of the variables, so we will ignore them.  

The commutator in the mode algebra is given by the modes of certain OPE coefficients. For any two operators $\mc{O}_1,\mc{O}_2$ we have
\begin{equation} 
	\left[ \oint \mc{O}_1(z) z^k \d z, \oint \mc{O}_2(z) z^l \d z \right] = \oint_{\abs{z} = 1} \oint_{\abs{u} = \eps}  \mc{O}_1(z) \mc{O}_2(z + u) z^k (z + u)^l \d u \d z. 
\end{equation}
The product $\mc{O}_1(z) \mc{O}_2(z+u)$ can be expanded in a series in $u^{-1}$ using the OPE. 

From this, we find that
\begin{equation} 
	\left[ \oint_z I_i Z_1 J_j   \d z, \oint_z I_j Z_2 J_k z \d z\right] = \oint_{\abs{z} = 1} \oint_{\abs{u} = \eps} I_i(z) Z_1(z) J_j (z)  I_j(z+u) Z_2(z+u) J_k(z+u)  (z+u) \d u \d z. 
\end{equation}
(Note that there are no ordering ambiguities in defining either of the composite operators we are considering). 

In the planar limit, there are two contractions between $I_i Z_1 I_j$ and $I_j Z_2 I_k$ that contribute. We must always contract $J_j$ in the first operator with $I_j$ in the second, yielding $u^{-1}$.  We may or may not contract $Z_1$ in the first operator with $Z_2$ in the second operator. This gives
\begin{equation} 
	I_i(0) Z_1(0) J_j(0) I_j(u) Z_2(u) J_k(u) = u^{-1} I_i(0) : Z_1 Z_2 : (0) + N u^{-2} I_i(0) J_k(0). 
\end{equation}
Inserting this into the formula for the commutator of operators, we find
\begin{align} 
	\left[ \oint_z I_i Z_1 J_j   \d z, \oint_z I_j Z_2 J_k z \d z\right] =& \oint_{\abs{z} = 1} \oint_{\abs{u} = \eps} I_i(z) : Z_1 Z_2: (z) J_k(z) u^{-1}  (z+u) \d u \d z \\
	+ N \oint_{\abs{z} = 1} \oint_{\abs{u} = \eps}  I_i(z) J_k(z) u^{-2} (z+u) \d u \d z \\
	=& \oint_{\abs{z} = 1} I_i(z) : Z_1 Z_2:(z) J_k(z) z \d z + N \oint I_i(z) J_k(z) \d z. 
\end{align}
By a similar calculation, we have
\begin{equation} 
 \left[ \oint_z I_i Z_2 J_j  z \d z,   \oint_z I_j Z_1 J_k \d z\right] =  \oint_{\abs{z} = 1} I_i(z) : Z_2 Z_1:(z) J_k(z) z \d z - N \oint I_i(z) J_k(z) \d z. 
 \end{equation}
 The minus sign on the right hand side occurs because we are contracting $Z_1$ and $Z_2$ in the other order.  

Since $:Z_2 Z_1: = :Z_1 Z_2:$,   comparing the two equations gives us exactly the relation in the algebra of functions on the conifold.

In sum, we have proved that the positive part of the open-string symmetries in $\mf{a}^{CFT}_{(k \mid k)}$ is isomorphic to $\mf{n}_+ \otimes \Oo(SL_2(\C))$, where $\mf{n}_+$ is the positive part of $\mf{pgl}(k \mid k)$. Further, we can recover the multiplication on $\Oo(SL_2(\C))$ from the Lie bracket.

\subsection{Matching all open-string symmetries}
The calculation given above shows that concatenation of open-string symmetries, as in figure \ref{figure_concatenate}, is given by the product in the algebra of functions on the conifold.  The result of this concatenation is clearly independent of the open-string indices $a,b,c$. 

The commutator in the algebra of open-string global symmetries is always given by such concatenations, where one pair of matter fields $I,J$ are contracted.  We conclude that for open-string global symmetries $e_{\mf{t}}(F)$, for $\mf{t} \in \mf{pgl}(k \mid k)$, we have  
\begin{equation} 
	[e_{\mf{t}} (F), e_{\mf{t}'}(G)] = e_{[\mf{t}, \mf{t}']}(F G).\label{formula_commutator} 
\end{equation}
This matches with the commutator of open-string global symmetries in $\mf{a}^{hol}_{(k \mid k)}$.  

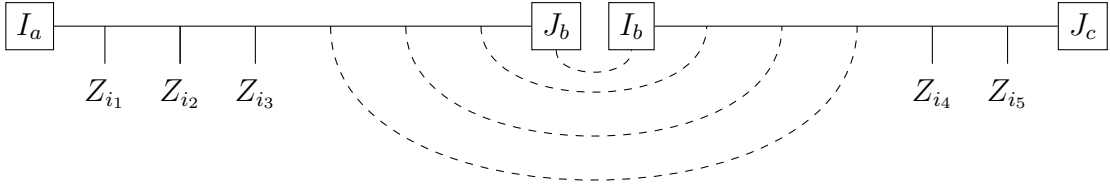
\begin{figure}
\begin{tikzpicture}
\node[draw, rectangle] (N1) at (-7,0) {$I_a$};
\node[draw, rectangle] (N2) at (0,0) {$J_b$};
\node[draw, rectangle] (N3) at (1,0) {$I_b$};
\node[draw, rectangle] (N4) at (7,0) {$J_c$};
\draw (N1) to (N2);
\draw[dashed, out=-90,in=-90](N2) to (N3);
\node(Z1) at (-6,-0.9) {$Z_{i_1}$};
\node(Z1) at (-5,-0.9) {$Z_{i_2}$};
\node(Z1) at (-4,-0.9) {$Z_{i_3}$};
\node(Z1) at (5,-0.9) {$Z_{i_4}$};
\node(Z1) at (6,-0.9) {$Z_{i_5}$};
\draw (-5,0) to (-5,-0.5); 
\draw (N3) to (N4);
\draw [dashed, out=-90,in=-90] (-1,0) to (2,0);
\draw [dashed, out=-90,in=-90] (-2,0) to (3,0);
\draw [dashed, out=-90,in=-90] (-3,0) to (4,0);
\draw (-4,0) to (-4,-0.5);
\draw (-5,0) to (-5,-0.5);
\draw (5,0) to (5,-0.5);
\draw (-6,0) to (-6,-0.5);
\draw (6,0) to (6,-0.5); 
\end{tikzpicture}
	\caption{Concatenating open string global symmetries. The dotted lines represent  Wick contractions that are allowed in the planar limit.  Note we contract one $I$ field with one $J$ field and a varying number of fields $Z_i$.  \label{figure_concatenate} }  
\end{figure}

\subsection{Matching closed-string global symmetries}

 It is straightforward to verify that there is an isomorphism of $GL(k \mid k) \times SL_2(\C) \times SL_2(\C)$ representations between $\mf{a}^{hol}_{(k \mid k)}$ and $\mf{a}^{CFT}_{(k \mid k)}$, and we have seen that the Lie brackets match on the open-string sector. In this section we will check that the Lie brackets between bosonic elements of the closed-string sector match.  

For $\mf{a}^{hol}_{(k \mid k)}$ the bosonic closed string symmetries are given by $\op{Vect}_0(SL_2(\C))$.  For $\mf{a}^{CFT}_{(k \mid k)}$ it is the subalgebra given by the modes of $A^{(n)}$ and $D^{(n)}$.  

One can check readily that the subalgebra of closed string operators is, in each case, independent of $k$. For $ \mf{a}^{hol}_{(k \mid k)}$ this is obvious. For $ \mf{a}^{CFT}_{(k \mid k)}$, this follows from the fact that the operators $A^{(n)}$ and $D^{(n)}$ do not contain the matter fields.  

We need to show that the bosonic closed string operators in $\mf{a}^{CFT}_{(k \mid k)}$ are isomorphic to $\op{Vect}_0(SL_2(\C))$ as a Lie algebra.  We only need to check this for $k$ sufficiently large.
 
We have seen that the algebra of open-string global symmetries is, on both side, $\mf{pgl}(k \mid k) \otimes \Oo(SL_2(\C))$.  The algebra of closed-string global symmetries acts on this by Lie algebra derivations, which commute with the action of $PGL(k \mid k)$.  The only such derivations are derivations of the commutative algebra $\Oo(SL_2(\C))$.

Derivations of $\Oo(SL_2(\C))$ are holomorphic vector fields.  Therefore we find a homomorphism of Lie algebras
\begin{equation} 
 \mf{a}^{CFT}_{closed} \to \op{Vect}(SL_2(\C)).
 \end{equation}
The sub-algebra $\op{Vect}_0(SL_2(\C))$ is given by all those $SL_2(\C) \times SL_2(\C)$ representations in $\op{Vect}(SL_2(\C))$ which are not of the form $V_j^L \otimes V_j^R$.  Since these representations do not occur in $\mf{a}^{CFT}_{closed}$, we find that we must have a homomorphism of Lie algebras
\begin{equation} 
 \mf{a}^{CFT}_{closed} \to \op{Vect}_0(SL_2(\C)).
 \end{equation}
It remains to show that this map is an isomorphism, which is equivalent to showing that it is surjective (because the $SL_2(\C) \times SL_2(\C)$ representations on each side are isomorphic). 

To see this, we note that $\op{Vect}_0(SL_2(\C))$ is generated by the infinitesimal left and right $\mf{sl}_2(\C)$ actions, together with the vector fields which transform in $V^L_{3/2} \otimes V^R_{1/2}$ and $V^L_{1/2} \otimes V^{R}_{3/2}$.  To show surjectivity, we need only find two global symmetries which act in a non-zero way on the open string algebra and which transform under $SL_2(\C) \times SL_2(\C)$ as highest weight vectors of weights $(3/2,1/2)$ and $(1/2,3/2)$.  The symmetries $J^{(3)}_{1/2}$ and $L^{(3)}_{3/2}$ transform in the correct representations, and one can check that both act in a non-zero way on the open-string algebra.

\subsection{Matching fermionic closed-string symmetries when $k > 0$}
When $k > 0$, there are fermionic closed string symmetries which on the CFT side are given by the modes of $C$, and on the holographic side are given by $\gamma \in \Oo(SL_2(\C))$.  On the holographic side, $\gamma$ transforms under the action of $\op{Vect}_0(SL_2(\C))$ as a holomorphic function on $SL_2(\C)$.  Further, the commutator of $\gamma$ with an open-string symmetry vanishes.

Let us check that the same holds on the CFT side.  The commutator of $c(F)$ with open-string global symmetries must produce a fermionic symmetry of the open-string Lie algebra which commutes with the action of $PGL(k \mid k)$.  All such symmetries are trivial, so that $c(F)$ commutes with the open-string symmetries.  

To see that $c(F)$ transforms in the correct way under the bosonic closed-string symmetries, we note that there is a subleading term in the $1/N$ expansion of the BRST differential which sends $c(F)$ to $e_{\op{Id}}(F)$. Since $e_{\op{Id}}(F)$ transforms under the action of $\op{Vect}_0(SL_2(\C))$ as an element of $\Oo(SL_2(\C))$, and the subleading term in the BRST action must commute with the action of $\op{Vect}_0(SL_2(\C))$, we conclude that $c(F)$ also transforms as an element of $\Oo(SL_2(\C))$ under the $\op{Vect}_0(SL_2(\C))$ action. 

\subsection{Matching fermionic closed-string commutators at $k = 0$}

The argument so far provides a complete proof that the global symmetry algebras match on both sides as long as $k > 0$.  At $k = 0$, the argument above tells us that the $a,c,d$ global symmetries commute with each other in the correct way.  The $a,d$ towers give us a copy of $\op{Vect}_0(SL_2(\C))$, and the $c$ tower transforms under the $a,d$ towers via the action of $\op{Vect}_0(SL_2(\C))$ on $\Oo(SL_2(\C))$. 

It remains to check that the $b$ tower transforms in the same way under the action of $\op{Vect}_0(SL_2(\C))$, and that there is a certain non-trivial commutation relation between the $b$ and $c$ towers.  

On the holographic side, the $b,c$ towers are represented by $\alpha,\gamma \in \Oo(SL_2(\C))$, and the commutation relation is that
\begin{equation} 
[\alpha,\gamma] = \partial \alpha \partial \gamma  
 \end{equation}
where the right hand side is viewed as a divergence-free vector field on $SL_2(\C)$ using the isomorphism between divergenge free holomorphic vector fields and closed $(2,0)$ forms.

To check that $b$ transforms in the correct way under $\op{Vect}_0(SL_2(\C))$, we introduce matter fields again. Then, there is a term in the BRST operator (to leading order in $1/N$) whereby
\begin{equation} 
Q_{BRST} e_{\mf{t}}(F) = (\op{Tr} \mf{t}) b(F) 
 \end{equation}
 for $F \in \Oo(SL_2(\C))$. Because $e_{\mf{t}}(F)$ transforms as an element of $\Oo(SL_2(\C))$ under the $\op{Vect}_0(SL_2(\C))$ action, we conclude that $b(F)$ must as well. 

It remains to calculate the commutator between $b(F)$ and $c(G)$. We denote the commutator by $V(F,G) \in \op{Vect}_0(SL_2(\C))$.  To match with the holographic computation, we need to show that for $H \in \Oo(SL_2(\C))$
\begin{equation} 
V(F,G) H = \Omega^{-1} \partial F \wedge \partial G \wedge \partial H \label{bc_commutator} 
 \end{equation}
 On the left hand side, we have the action of the vector field $V(F,G)$ on $H$, and on the right hand side, we are using the holomorphic volume form to identify closed $(3,0)$-forms on $SL_2(\C)$ with holomorphic functions.  

To prove the identity \eqref{bc_commutator}, we will use the Jacobi identity.  Since we know that $\op{Vect}_0(SL_2(\C))$ acts on the symmetries $c(H)$ according to the action on $H \in \Oo(SL_2(\C))$, we have
\begin{equation} 
[[b(F), c(G)], c(H)] = c( V(F,G) H ). 
 \end{equation}
The Jacobi identity tells us that
\begin{equation} 
 [[b(F), c(G)], c(H)] = [b(F), [c(G), c(H)] ] - [ [b(F), c(H)] , c(G) ]. 
 \end{equation}
Since $[c(G), c(H)] = 0$, we find that
\begin{equation} 
c (V(F,G) H ) = - c( V(F,H), G ). 
 \end{equation}
Thus, $V(F,G) H$ is anti-symmetric under the permutation of $G$ and $H$. A similar argument using the Jacobi identity for the commutator with two $b$'s and one $c$ tells us that $V(F,G) H$ is anti-symmetric under permutation of $F$ and $H$, and hence totally anti-symmetric.  We will change notation slightly and right
\begin{equation} 
\Lambda(F,G,H) = V(F,G) H. 
 \end{equation}
 The operator $\Lambda$ is a totally anti-symmetric map 
\begin{equation} 
\Lambda :  \Oo(SL_2(\C))^{\otimes 3} \to \Oo(SL_2(\C)). 
 \end{equation}
Since $V(F,G)$ acts as a derivation on $\Oo(SL_2(\C))$, we have
\begin{equation} 
\Lambda(F,G,H H') =  \Lambda(F,G,H)  H' +     \Lambda(F,G,H')  H. 
 \end{equation}
 By anti-symmetry, $\Lambda$ is a derivation in each factor. Therefore $\Lambda$ involves only a single derivative in each variable, so that $\Lambda (F,G,H)$ is a function of the $3$-form $\d F \wedge \d G \wedge \d H$.  Because $\Lambda$ is invariant under the $SL_2(\C)$ action, we must have
\begin{equation} 
\Lambda (F,G,H) = C \Omega^{-1} \d F \wedge \d G \wedge \d H 
 \end{equation}
for some constant $C$.  

This determines $\Lambda$ up to a constant.  The constant can be set to $1$ by rescaling the symmetry $c(F)$.

This concludes the proof that the holographic global symmetry algebra $\mf{a}^{hol}_{(k \mid k)}$ is isomorphic to the CFT global symmetry algebra $\mf{a}^{CFT}_{(k\mid k)}$, for all values of $k$ including $k = 0$.

\subsection{An alternative interpretation of the large $N$ global symmetry algebra}
Let us sketch here an alternative way to understand the large $N$ global symmetry algebra, which can be applied more generally. 

For any field theory, we can consider the cohomology of the space of local Lagrangians of varying ghost number, under the BRST operator.  If we shift the ghost number by $1$ we find that this cohomology is a graded Lie algebra, with bracket given by the BV anti-bracket.   The first cohomology of this Lie algebra consists of infinitesimal deformations of the theory, and $H^0$ consists of infinitesimal symmetries.
These symmetries and deformations need not preserve any symmetry of the theory, including the Poincar\'e or conformal symmetry.  

At the classical level, the definition of this Lie algebra is straightforward. At the quantum level, things are quite a bit more difficult, because some of the structures in this Lie algebra are subject to UV divergence.  The renormalization formalism of \cite{Cos11} allows one to define the Lie algebra\footnote{In the language of \cite{Cos11}, if $L_i$ are two BSRT closed Lagrangians of arbitrary ghost number, and $\eps_i$ are parameters satisfying $\eps_i^2 = 0$ of opposite ghost number, then the Lie bracket $[L_1,L_2]$ is the obstruction to $\eps_1 L_1 + \eps_2 L_2$ satisfying the quantum master equation modulo $\eps_i^2$. }  

In the language of deformation theory, this is the Lie algebra \emph{controlling} deformations of the quantum field theory. This means that deformations of the QFT are given by Maurer-Cartan solutions in the Lie algebra, and the Lie algebra of symmetries of the QFT is given by $H^0$ of the Lie algebra.   

For a large $N$ gauge theory in the planar limit, there is a version of this Lie algebra consisting of only single-trace symmetries and deformations. 

For a large $N$ chiral CFT,  it turns out that the Lie algebra given by single-trace Lagrangians is isomorphic to the Lie algebra of global symmetries.  Let us explain how to match global symmetries with Lagrangians.

Suppose $\mc{O}$ is a single-trace primary of spin $r$ in the large $N$ gauge theory, then we get a single-trace Lagrangian by the expression 
\begin{equation} 
	\int_{\CP^1} z^k \d z \til{\mc{O}} 
\end{equation}
where $\til{\mc{O}}$ is the descendent of $\mc{O}$: it satisfies $Q_{BRST} \til{\mc{O}} = \dbar \mc{O}$. Since $\mc{O}$ is of spin $r$, the operator $\til{\mc{O}}$ takes values in $\Omega^{0,1}(\CP^1,\Oo(-2 r))$. It follows that we must multiply it by a section of $\Oo(2r-2)$ to get a $(1,1)$-form, so that $0 \le k \le 2 r -2$. 

From this we see that the collection of single-trace Lagrangians, up to total derivative, is in bijection with the collection of global symmetries. This isomorphism involves a shift by $1$ in the ghost number, because we used descent.

It is not \emph{a priori} obvious that the Lie bracket on Lagrangians matches the commutator of global symmetries. However, this was proved by Si Li \cite{Li16}, who showed in that quite generally the Lie bracket on Lagrangians which are obtained from descent of local operators matches the commutator of the modes of the same local operators.   

From this point of view, the isomorphism between holographic and large $N$ global symmetry algebras has a nice interpretation.  The large $N$ global symmetry algebra is the Lie algebra controlling symmetries and deformations of the large $N$ theory, as a field theory on $\CP^1$.  The holographic global symmetry algebra is the Lie algebra controlling symmetries and deformations of the $B$-model background on $SL_2(\C)$.  The statement that they are isomorphic has a very clear holographic meaning.

\section{Boundary conditions from the holomorphic point of view}

In order to define boundary conditions in an holomorphic setup it is useful to compactify the deformed conifold. 
The deformed conifold $X$ is the affine algebraic variety defined by the equations $\epsilon^{ij} u_i w_j= N$. 
We can compactify this to the projective algebraic variety with homogeneous coordinates $(U_1:U_2:W_1:W_2:V)$ satisfying the equation 
\begin{equation} 
\br{X}_N = \{(U_1:U_2:W_1:W_2:V)\mid  \epsilon^{ij} U_i W_j= N V^2\}. \label{equation_quadric} 
\end{equation}
This compactified deformed conifold is a quadric in $\CP^4$. All quadrics in $\CP^4$ associated to non-degenerate quadratic forms in $5$ variables are equivalent.  

Because $\br{X}$ is associated to a quadratic form in $5$ variables, it has a transitive action of the group $\op{Spin}(5) = \op{Sp}(4)$.  We can identify $\br{X}$ with the set of null complex lines in $\C^5$, the vector representation of $\op{Spin}(5)$.

The boundary of the deformed conifold is the subvariety $V= 0$ inside the quadric \eqref{equation_quadric}.   The boundary is thus given by the quadric
\begin{equation} 
	\epsilon^{ij} U_i W_j= 0 
\end{equation}
inside the $\CP^3$ with homogeneous coordinates $(U_1:U_2:W_1:W_2)$.   A quadric in $\CP^3$ is a copy of $\CP^1 \times \CP^1$. Explicitly, a point in the boundary is given by a $2 \times 2$ matrix 
\begin{equation} 
	\begin{pmatrix}
	W_1 & W_2 \\ U_1 & U_2 
	\end{pmatrix}
\end{equation}
of rank $1$, taken up to scale.  To such a matrix we can assign its kernel and image, which are both lines in $\CP^1$.  In this way the set of rank $1$ matrices up to scale is identified with $\CP^1 \times \CP^1$.  

If we let $\Oo(1)$ denote the line bundle on $\br{X}$ restricted from the line bundle of the same name on $\CP^4$, then the canonical line bundle of $\br{X}$ is 
\begin{equation} 
K_{\br{X}} = \Oo(-3). 
 \end{equation}
This follows from the fact that $K_{\CP^4} = \Oo(-5)$ and that $\br{X}$ is the zero locus of a section of $\Oo(2)$.  Since the boundary divisor $D \subset \br{X}$ is the zero locus of a section of $\Oo(1)$, we conclude that the holomorphic volume form on $X$ must extend to a volume form on $\br{X}$ which has a cubic pole along $D$.

\subsection{Boundary conditions for holomorphic Chern-Simons theory}

Let $A^{0,\ast} \in \Omega^{0,\ast}(X) \otimes \mf{gl}(K \mid K)[1]$ be the field for holomorphic Chern-Simons theory, in the BV formalism.  Our boundary condition is that  $A^{0,\ast}$ extends to an element
\begin{equation} 
	A^{0,\ast} \in \Omega^{0,\ast}(\br{X}, \Oo(- D) ) \otimes \mf{gl}(K \mid K)[1]. 
\end{equation}
That is, $A$ has a first-order zero along the boundary.

Since $\Omega_X$ has a pole of order $3$ on the boundary divisor, the cubic interaction
\begin{equation} 
\int \Omega_X (A^{0,\ast} \wedge A^{0,\ast} \wedge A^{0,\ast}) 
 \end{equation}
 is well-defined.  

To check that this is a consistent boundary condition, we need to know that we can write down a propagator compatible with this constraint.   We will verify this in the appendix \ref{app:bc}.  

As we verify in the appendix \ref{appendix_coherent_cohomology}, we have 
\begin{equation} 
H^\ast_{\dbar}(\br{X}, \Oo(-D)) = 0. 
 \end{equation}
 This implies that there are no zero-modes of the theory, and that any modification of the boundary condition on $D$ will source a unique bulk field. 
\subsection{Boundary conditions for Kodaira-Spencer theory}

As we have seen \ref{sec:KS_alternative}, we can view Kodaira-Spencer theory as $(1,0)$ Kodaira-Spencer theory coupled to holomorphic Chern-Simons theory for the Abelian Lie algebra $\Pi \C^2$. The field of $(1,0)$ Kodaira-Spencer theory is $(2,1)$ form $\alpha^{2,1} \in \Omega^{2,1}(X)$, which satisfies $\partial \alpha = 0$. 

We will prescribe the same boundary conditions for holomorphic Chern-Simons theory for $\Pi \C^2$ as we did for holomorphic Chern-Simons for $\mf{gl}(k \mid k)$.

For the field $\alpha^{2,1}$, we will ask that $\alpha^{2,1}$ extends to a $(2,1)$ form on $\br{X}$ with a logarithmic pole at the boundary divisor. Let us recall what this means. 

Following the terminology standard in algebraic geometry, we can define the sheaf on $X$ of holomorphic differential forms with a log pole on $D$ as follows. Away from the divisor $D$, such differential forms are ordinary holomorphic differnential forms.  Near the divisor, in a system of coordinates $n,w,z$ where the divisor is at $n = 0$, the logarithmic de Rham complex is generated over the ring of holomorphic functions by the anti-commuting elements $n^{-1} \d n$, $\d w$, $\d z$.  The de Rham differential makes this into a subsheaf of the sheaf of differential forms with arbitrary poles on $D$. 

In a similar way, we can define $\Omega^{\ast,\ast}(\br{X}, \log D)$, where the only poles we allow appear as $n^{-1} \d n$. In this complex $\br{n}^{-1} \d \br{n}$ does not appear.  We remind the reader of the theorem of Deligne \cite{Deligne1971}: there is an isomorphism
\begin{equation} 
H^\ast(\Omega^{\ast,\ast}(\br{X}, \log D), \partial + \dbar) \iso H^\ast_{dR} (X).  
 \end{equation}

The boundary condition for the field $\alpha^{2,1}$ we impose is that we ask that $\alpha^{2,1}$ extends to a $(2,1)$ form on $\br{X}$ with logarithmic poles on $D$: 
\begin{equation} 
\alpha \in \Omega^{2,1}(\br{X}, \log D) \label{eqn_bc} 
 \end{equation}
We will also ask that the gauge transformations for $\alpha$ have a logarithmic pole, so that the whole field $\alpha^{2,\ast}$ has a logarithmic pole. Finally we require that the constraint $\partial \alpha^{2,\ast} = 0$ holds in $\Omega^{3,\ast}(X, \log D)$. 

It is convenient, as in \cite{CosLi15}, to impose the constraint $\partial \alpha^{2,\ast} = 0$ cohomologically. This involves introducing a new tower of fields $\alpha^{3,\ast} \in \Omega^{3,\ast}(\br{X}, \log D))$, so that $\alpha^{2,\ast}$ and $\alpha^{3,\ast}$ together can be viewed as living in a piece of the de Rham complex:
\begin{equation} 
\alpha^{\ge 2,\ast} = \alpha^{2,\ast} + \alpha^{3,\ast} \in \Omega^{\ge 2,\ast} (\br{X}, \log D) [1]. 
 \end{equation}
The BRST operator is then $\partial + \dbar$.  For the fields of ghost number $0$, this imposes the constraint $\partial \alpha^{2,1}= \dbar \alpha^{3,0}$, $\dbar \alpha^{2,1} = 0$.  

By using the holomorphic volume form on $X$, we can identify $\alpha^{2,1}$ with an element $v^{1,1} \in \PV^{1,1}(X)$.  The volume form on $\br{X}$ has an order $3$ pole on the boundary, so that in local coordinates it takes the form $n^{-3} \d n \d w \d z$.   Asking that $\alpha^{2,\alpha}$ has logarithmic poles on the boundary means that the $\partial_z,\partial_w$ components of $v^{1,\ast}$ have zeroes of order $2$, and the $\partial_n$ component has a zero of order $3$.  
The cubic interaction for Kodaira-Spencer theory is quadratic in $\Omega_X$, and therefore has an order $6$ pole along the boundary divisor $D$.  The order $6$ pole in the cubic interaction cancels with the order $7$ zero in $(v^{1,\ast})^3$, so the cubic interaction is well-defined. 

\subsection{Zero modes in Kodaira-Spencer theory}
Unfortunately, there is a single zero-mode for Kodaira-Spencer theory when we use this boundary condition.  One can show, by explicit computation or using the degeneration of the logarithmic Hodge to de Rham complex, that the natural map
\begin{equation} 
H^\ast (\Omega^{\ge 2,\ast}(\br{X}, \log D) ) \to H^\ast(X)  
 \end{equation}
is in isomorphism in total degrees $2$ and higher.  This implies that
\begin{equation} 
H^1(\Omega^{\ge 2,\ast}(\br{X}, \log D)) = \C 
 \end{equation}
 and other cohomology groups vanish. The non-zero cohomology class is represented by some logarithmic $(2,1)$-form $\alpha^{2,1}$ such that 
\begin{equation} 
\int_{n= 0 , z = 0} \op{Res}_{n = 0} \alpha^{2,1} = 1. 
 \end{equation}
The reside of $\alpha^{2,1}$ is a $(1,1)$ form on the boundary, which we integrate over the $\CP^1$ given by $z = 0$.

Such a $(2,1)$-form is cohomologous to the holomorphic volume form on $SL_2(\C)$.  This zero mode therefore represents the string coupling constant, which we would like to be a parameter of the theory and not a mode to be integrated over.  

To fix this problem, we will modify the boundary condition by asking that our fields not only have logarithmic poles, but that also the integral of $\op{Res} \alpha^{2,1}$ over $n = 0, z = 0$ vanishes.  At first sight, it may seem that we have broken one of the $SL_2(\C)$ symmetries to impose this condition. However, at the level of the equations of motion, the choice of $\CP^1$ over which we integrate does not matter.

In the appendix \ref{app:bc} we will verify that the boundary condition is consistent, i.e. that there is well-defined propagator.

\section{Boundary operators as modifications of the boundary condition}

We can identify $AdS_3$ with the quotient of $X = SL_2(\C)$ by the left action of $SU(2)$.  This action extends to an action on $\br{X}$, because $\br{X}$ has a manifest $SO(5)$ symmetry which includes $SU(2)$.   The $SU(2)$ orbits, as they approach the boundary, degenerate from being $S^3$'s to $\CP^1$'s.   The boundary divisor is $D = \CP^1 \times \CP^1$, and $SU(2)$ rotates one of the $\CP^1$'s.  The other $\CP^1$ remains in the quotient.

The quotient $\br{X}/ SU(2)$ is therefore a compactification $\br{AdS}_3$, where we have added a $\CP^1$ at infinity.  The boundary divisor $D$ of $\br{X}$ is $\CP^1 \times \CP^1$, and is the trivial $\CP^1$ fibration over the boundary of $\br{AdS}_3$. 

The choice of boundary conditions for the fields on $X$ gives rise to a choice of boundary conditions for the effective theory on $AdS_3$.   We therefore get an effective theory on $\br{AdS}_3$ whose fields are the modes of Kodaira-Spencer theory on $\br{X}$.  

Single-trace operators of the holographic dual theory are given by changes of the boundary condition which are localized to a point on the boundary of $AdS_3$.  Because a point on the boundary of $AdS_3$ corresponds to a $\CP^1$ in the boundary of the deformed conifold $X$, we see that changes of the boundary condition for Kodaira-Spencer theory along a $\CP^1$ should match the operators of the large $N$ chiral algebra.  In this section we will verify this statement.

The formal statement is the following.  Consider curve $\CP^1 \times z$ in the boundary of the deformed conifold.  Let $U$ be a small neighourhood of this curve in $\br{X}$, and consider the complement $U\setminus (\CP^1 \times z)$.  A modification of the boundary condition localized at $\CP^1 \times z$ will induce a solution to the equations of motion on $U \setminus (\CP^1 \times z)$. Away from $\CP^1 \times z$, this solution will satisfy the boundary conditions described above. 

What we will show in the next several sections is the following.
\begin{proposition}
There is an isomorphism of $SU(2)_R \times SO(2)$ representations between:
\begin{enumerate} 
\item Single-trace operators of the large $N$ chiral algebra.
\item Solutions to the linearized equations of motion for Kodaira-Spencer theory on $U \setminus (\CP^1 \times z)$, modulo those which extend to  solutions on all $U$ (satisfying the boundary conditions).  
\end{enumerate}
\end{proposition}

\subsection{Modifying the boundary conditions for holomorphic Chern-Simons theory}

Fix a surface $\CP^1 \times 0$ in the boundary divisor $D$. We will choose coordinates $w$ on this curve, $z$ along the other $\CP^1$ in the boundary, and $n$ in the normal direction.  Since the normal bundle to the boundary is $\Oo(1,1)$, the function $n$ is a section of the bundle $\Oo(-1,-1)$ and so has a pole at $w = \infty$ and $z = \infty$.  

We will work on a neighbourhood $\CP^1 \times \C$ of the algebraic curve $\CP^1 \times 0$ in the boundary divisor. A neighbourhood of this region in the whole manifold $\br{X}$ is  not isomorphic to the total space of the normal bundle $\Oo(1)$ on $\CP^1 \times \C$. Instead, this bundle is deformed by the Beltrami differential   
\begin{equation} 
N n^2  \d \wbar \frac{1}{(1 + \abs{w}^2)^2}   \partial_z . \label{eqn_beltrami} 
 \end{equation}
(We are dropping some factors of $\pi$, which can be absorbed into scale of the holomorphic volume form on the deformed conifold). 

Here, we note that because $n^2$ has a pole of order $2$ at $w =\infty$, this expression makes sense globally on the total space of $\Oo(1) \to (\CP^1 \times \C)$.    

The complement of $n = 0$ in this complex manifold is the deformed conifold. There are holomorphic functions with poles at $n = 0$ given by $w_1 = 1/n, w_2 = w/n$ and
\begin{align} 
u_1 &= z/n -  N n    \frac{\wbar}{(1 + \abs{w}^2)}\\ 
u_2 &= wz/n + N n    \frac{1}{(1 + \abs{w}^2)}\label{eqn_boundary_coords} 
 \end{align}
We have $u_2 w_1 - u_1w_2 = N$. 

The group $SU(2)_R$ rotates the $\CP^1$ with coordinate $w$, and the Lorentz group $SO(2)$ rotates the $z$-plane. 

In these coordinates, the holomorphic volume form is, up to a factor independent of $N$, 
  \begin{align}
\Omega &= \frac{1}{w_1} \d u_1 \d w_1 \d w_2 \\
          &= - n^{-3} \d n \d w \d z  + N n^{-1} \d n \d w \d \wbar \frac{1}{(1 + \abs{w}^2)^2}. 
  \end{align}
  One can see quite explicitly in these coordinates that the integral over the three-cycle in $SL_2(\C)$ of $\Omega$ must be proportional to $N$.  This three-cycle is a Hopf fibration over the $\CP^1$ with coordinate $w$ in the boundary, so the integral of $\Omega$ over the three-cycle is proportional to the integral of the reside of $\Omega$ over the $w$-plane, which is evidently proportional to $N$. 

The fundamental field of holomorphic Chern-Simons theory is an element $A\in \Omega^{0,1}(\br{X}, \Oo(-D))\otimes \mf{psl}(K \mid K)[1]$. The ``vacuum'' boundary condition requires $A$ to vanish at the boundary, so that in the expansion in powers of $n$ the leading term is of order $n$.   We can try to modify the boundary condition in a very simple way, by asking that away from $z = 0$, the leading order term is still of order $n$, but near $z = 0 $ the leading order term is
\begin{equation} 
n^{-k}w^l \delta^{(r)} _{z = 0} t_a \label{configuration_initial} 
 \end{equation}
 where $a$ is an index for a basis of $\mf{psl}(K \mid K)$.  In this expression, $k \ge 0$, and $l \le k$ to ensure that there are no poles at $w = \infty$. Also $\delta^{(r)}_{z = 0}$ indicates the $r$th $z$-derivatives of the delta function. The group $SU(2)_R$ acts by change of coordinates on the $ w $ plane. This expression transforms in the $ SU (2)_R$ representation of dimension $ k +1 $. 

This is not quite sufficient to describe a boundary condition, however.  To fully describe a boundary condition, we need to write down a field configuration localized at $z = 0$ which solves the equations of motion modulo terms of order $n$ and higher. The configuration of \eqref{configuration_initial} does not satisfy the equations of motion, because of the Beltrami differential \eqref{eqn_beltrami} deforming the complex structure.   We can, however, easily add correction terms to \eqref{configuration_initial}   to ensure that it satisfies the equations of motion, while still ensuring that the singular terms remain localized near $z = 0$.  We will write the complete solution for the boundary condition which transforms as a highest weight vector under $SU(2)_R$.  The general case is obtained by applying the lowering operator
$$
\partial_{\wbar} - w^2 \partial_{w}
$$
of the $R$-symmetry algebra.

The complete expression for the solution to the equations of motion for the highest weight boundary condition,  stripping off the colour factor, is 
\begin{multline} 
	\mc{E}_{k,r} =   n^{-k} \delta^{(r)}_{z = 0} + \sum_{s = 1}^{k} \frac{1}{s!} n^{2s - k} (-1)^s N^s \frac{\wbar^s  }{ (1 + \abs{w}^2)^s   }  \delta_{z = 0}^{(r + s)} \\
	+ n^{k+2} \frac{(r+k+1)!}{k!} (-1)^{r}  \frac{N^{k+1}}{ 2 \pi \i} \frac{\wbar^k \d \wbar  }{ (1 + \abs{w}^2)^{k+2}   } \frac{1}{z^{r+k+2}} \label{eqn_boundary_soln_hcs} 
 \end{multline}
The field $\mc{E}_{k,r} t_a $ solves the equations of motion exactly, not just up to terms which satisfy the vacuum boundary condition.  To see this, we note that the Beltrami differential \eqref{eqn_beltrami} of one term in the expansion is minus the $\dbar$ operator in the $w$ plane applied to the next term; and the Beltrami differential applied to the last term vanishes, because it contains a $\d \wbar$.

The expression $\mc{E}_{k,r}$ is of weight $k/2$ under the Cartan of $SU(2)_R$ and of spin $k/2+r+1$ under $SO(2)$.  Because the boundary conditions require that our field is divisible by $n$, then $\mc{E}_{k,r}$ violates the boundary conditions at $z = 0$ even if $k = 0$.  

This expression is localized at $z = 0$ up to order $n^k$. At this point, we are forced to introduce a non-local expression. If we continued the pattern of the previous terms, we would find $n^{k+2} \wbar^{k+1} (1 + \abs{w}^2)^{-k-1}$ has a pole at $w = \infty$. To determine the prefactors in the last term we use the equation
\begin{equation} 
	\frac{1}{2 \pi i} (-1)^{r+k+1} (r+k+1)! \dbar z^{-r-k-2} = \delta_{z = 0}^{r+k+1}. 
\end{equation}

\subsection{An alternative formula for a field satisfying modified boundary conditions.}
Above we wrote an explicit formula for a field satisfying the linearized equations of motion of holomorphic Chern-Simons theory and the boundary condition except at $z = 0$. Here we will write another formula for such a field which has better global behaviour. These two field configurations are gauge equivalent.

We use homogeneous coordinates $U,W,V$ on $\br{X}$ so that $(U,W) =  N V^2$. Here the $(\cdot,\cdot)$ is the contraction of $SU(2)_R$ indices by $\epsilon^{ij}$. We let $\sigma$ be the Pauli matrix
\begin{equation}
\sigma = \begin{pmatrix} 0 & 1 \\
         -1 & 0 
\end{pmatrix}
\end{equation}

Then, we can take
\begin{align}
\mc{F}_{0,0}  &= \dbar \frac{(W, \sigma \br{U} )}{(U, \sigma \br{U})} \\
        &= \dbar \frac{W_1 \br{U}_1 + W_2 \br{U}_2  }{\abs{U_1}^2 + \abs{U}_2^2  }.
 \end{align}

The field $\mc{F}_{0,0} t_a$ satisfies the linearized equations of motion for holomorphic Chern-Simons theory. 

The boundary of $\br{X}$ is the locus $V = 0$, and so is the quadric with homogeneous coordinates $U,W$ satisfying $(U,W) = 0$.  This is a $\CP^1 \times \CP^1$ with coordinates $z,w$ so that $W = (1,w)$, $U = zW$.  As above we take $z$ to be the chiral algebra plane. The field $\mc{F}_{0,0}$ is non-singular, and divisible by $V$, except at the locus $U = 0$.  In the boundary, this is the curve where $z =0$.  Therefore $\mc{F}_{0,0}$ defines a modification of the boundary condition at $z = 0$.

We can analyze the behaviour of $\mc{F}_{0,0}$ near the boundary in the coordinates $w,z$: 
\begin{equation}
\mc{F}_{0,0} = \dbar \frac{\br{z} ( 1 + \abs{w}^2) }{\abs{z}^2(1 + \abs{w}^2)  }  = \dbar z^{-1}.  
\end{equation}
Thus, $\mc{F}_{0,0}$ is defined by taking some function which on the boundary is $1/z$ and applying $\dbar$.  Since $\mc{E}_{0,0}$ is also defined in the same way, they are gauge equivalent. 

Away from the boundary and from the locus $U = 0$ we have
\begin{equation}
\mc{F}_{0,0}= -N V^2 \frac{(\sigma \br{U},\d \sigma \br{U})}{(U, \sigma \br{U})^2} . 
\end{equation}
It turns out that for all $n$
\begin{equation}
\dbar \left( \bar U_{i_1} \cdots \bar U_{i_n} \frac{(\bar U,\d \bar U)}{(U, \sigma\br{U})^{n+2}} \right) =0 
\end{equation}
and thus we can define 
\begin{equation}
\mc{F}_{0,k} =  V^{2+k} \bar U_{i_1} \cdots \bar U_{i_k} \frac{(\bar U,\d \bar U)}{(U,  \sigma \br{U})^{k+2}}
\end{equation}
This satisfies the equations of motion and the boundary condition except at $U = 0$.  

The field $\mc{E}_{0,k}$ has a similar, but slightly more singular, global expression. The term which is not a $\delta$-function at $z = 0$ in $\mc{E}_{0,k}$ takes the form
\begin{equation}
 N^{k+1}  \bar{W}_2^k V^{2+k} \frac{ (\bar{W}, \d \bar W) } {(U,\sigma \br{W} )^{k+2}  } 
\end{equation}
where we have used used \eqref{eqn_boundary_coords} to translate between the coordinates $z,w,n$ and the coordinates $U_i,W_i$.  (Here we write the field which is a highest weight vector for $SU(2)_R$).   

The two field configurations are gauge equivalent up to a function of $N$. We can write down a family of field configurations
\begin{equation}
\mc{G}_{0,k}(s) = (s \bar{U}_2 + (1-s) \bar{W}_2)^k V^{2+k} \frac{ ((s \bar{U} + (1-s) \br{W}, s \d \bar U + (1 -s) \d \bar W  } {(U,(s \sigma \br{U} + (1-s)\sigma \br{W}) ^{k+2}  }. 
\end{equation}
We have
\begin{equation}
\partial_s \mc{G}_{0,k}(s) = \dbar \left(\iota_{(- \br{W}_i + \br{U}_i)(\partial_{\br{W}_i} + \partial_{\br{U}_i})} \mc{G}_{0,k}(s)   \right)  .
\end{equation}
Here $\iota_{-}$ indicates contraction with respect to the vector field.  Thus, the family of fields $\mc{G}_{0,k}(s)$ are gauge equivalent, and in particular $N^{-k-1}\mc{E}_{0,k}$ and $\mf{F}_{0,k}$ are.  

We will tend to use the field $\mc{E}_{0,k}$ to calculate holographic OPE coefficients because of its nice expansion near the boundary.

\subsection{Classifying all modifications of the boundary condition}
It turns out that all possible modifications of the boundary condition are of this form, up to gauge equivalence. Consider the most general consistent modification of the boundary condition localized at the $\CP^1$ given by $z = 0$, $n=0$.
We will assume that our fields will always extend across the boundary as distributions, so that 
\begin{equation} 
A \in \br{\Omega}^{0,1}(\br{X}, \Oo(-D)) 
 \end{equation}
(where $\br{\Omega}^{0,\ast}$ is the distributional de Rham complex). 

Then, the most general modification we can make to the equations of motion is asking that
\begin{equation} 
\dbar A = \eta 
 \end{equation}
 for some $\eta \in \br{\Omega}^{0,2}_{\CP^1}(\br{X})$. That is, $\eta$ is a distributional $(0,2)$ form supported on the $\CP^1$ where $n = z = 0$.  For example, we could take $\eta$ to be some $n$ and $z$ derivative of the $\delta$-function at $n = z = 0$, leading to a solution $A = n^{-k} \delta_{z =0}^{(r)} + \dots$.  

 In order for this equation to be consistent, we need $\dbar \eta = 0$.  Further, if $\eta = \dbar \chi$, for some $\chi \in \br{Omega}^{0,1}_{\CP^1}(\br{X})$, then we can solve this equation by setting $A = \chi$. Since $\chi$ is supported on the $\CP^1$ in the boundary, $A$ will be zero away from the boundary. This modified boundary condition does not source any bulk field and should be treated as being trivial. 

We see that possible modifications of the boundary condition are given by the Dolbeault cohomology groups 
\begin{equation} 
H^2_{\CP^1}(\br{X},\Oo(-D)) 
 \end{equation}
 with support on the chosen $\CP^1 \subset D$.
\begin{lemma} 
The cohomology groups $H^i_{\CP^1}(\br{X}, \Oo(-D))$ vanish unless $i = 2$. A basis for $H^2_{\CP^1}(\br{X}, \Oo(-D))$ is provided by $\dbar$ of the fields 
\begin{equation} 
	(\partial_{\wbar} - w^2 \partial_{w})^l \mc{E}_{k,r} t_a 
\end{equation}
for $k,r \ge 0$ as $0 \le l \le k$. These boundary modifications transform in the representation of spin $k/2$ of $SU(2)_R$, and are of spin $k/2+r+1$ under the $SO(2)$ Lorentz group.   \end{lemma} 
\begin{proof}
Let $U$ be a neighourhood of the $\CP^1$ given by $n = 0$, $z = 0$.  As usual, let $D \subset U$ be the boundary divisor $n = 0$. 

	We can calculate the relevant Dolbeault cohomology groups by first dropping the effect of the Beltrami differential \ref{eqn_beltrami}.  In this case, $U$ is a neighourhood of a $\CP^1$ inside the total space of $\Oo(1,1) \to \CP^1 \times \CP^1$. We can therefore identify $U$ with a neighbourhood of $\CP^1$ in $\Oo(1) \to \CP^1 \times \C$.

Let us work locally on $\CP^1$ for the moment. The Dolbeault cohomology of $\C^3$ with support on $\C \subset \C^3$ is entirely in degree $2$, and is spanned by the $\delta$-function on $\C$ and its normal derivatives.  

Thus, locally on $\CP^1$, the cohomology is given by expressions like
	\begin{equation} 
		w^l \partial_n^k \partial_z^r \delta_{n = z = 0}. 
	\end{equation}
This transforms, at first sight, as a section of the bundle $\Oo(k+1)$ on $\CP^1$, however we have an overall twist by $\Oo(-D) = \Oo(-1)$ so we find that this expression transforms in $\Oo(k)$. Under the $SO(2)$ action which rotates the plane of the CFT, this expression has spin $r+1+k/2$.   

Globally on $\CP^1$, we find that the relevant cohomology group is given by the cohomology of $\CP^1$ with coefficients in $\Oo(k)$ for each expression like $\partial_n^k \partial_z^r$. Because $k \ge 0$, we only find cohomology on $\CP^1$ in degree $0$. The total cohomological degree is $2$.  This cohomology transforms as the spin $k/2$ representation of $SL_2(\C)_R$, and spin $k/2+r+1$ under $SO(2)$, as desired.

\end{proof}
Although we have given an explicit formula for the bulk field sourced by these boundary modification, we can prove abstractly that there is a unique bulk field (up to gauge equivalence) satisfying any of these modified boundary conditions.

To see this, we note that there is a long exact sequence in Dolbeault cohomology
\begin{equation} 
\dots \to H^i (\br{X},\Oo(-D)) \to H^i(\br{X} \setminus \CP^1, \Oo(-D)) \to H^{i+1}_{\CP^1}(\br{X}, \Oo(-D)) \to \dots 
 \end{equation}
where the subscript $\CP^1$ on the last cohomology group indicates Dolbeault cohomology supprted at the $\CP^1$ in the boundary.  Because $H^i(\br{X}, \Oo(-D))$ vanishes for all $i$, we find that 
\begin{equation} 
H^1(\br{X} \setminus \CP^1, \Oo(-D)) \iso H^2_{\CP^1}(\br{X}, \Oo(-D)).  
 \end{equation}
 The space $H^1(\br{X} \setminus \CP^1,\Oo(-D))$ represents gauge equivalence classes of solutions to the linearized equations of motion for holomorphic Chern-Simons theory, which satisfy the boundary conditions except at the $\CP^1$ in the boundary.  This space is then isomorphic to $H^2_{\CP^1}(\br{X}, \Oo(-D))$, which is the space of modifications of the boundary conditions.  

 This abstract argument shows that given any modification of the boundary conditions, there is a unique up to gauge equivalence solutions of the linearized equations of motion for holomorphic Chern-Simons theory satisfying this modified boundary condition. 

It is clear from this calculation that boundary operators for holomorphic Chern-Simons transform in the same representation for $SL(2) \times SO(2)$ as the open-string operators for the large $N$ chiral algebra with matter.

\subsection{Boundary modifications for Kodaira-Spencer theory}
The fundamental field of $(1,0)$ Kodaira-Spencer theory is an element $\alpha^{2,1} \in \Omega^{2,1}(X)$, which satisfies the equation $\partial \alpha^{2,1}$ and the equation of motion $\dbar \alpha^{2,1} = 0$.  The boundary conditions state that $\alpha$ has a pole of order $1$ at $n = 0$.  

In this section we will calculate all possible boundary modifications at $z = 0$. The calculation is quite a bit more involved than the calculation we did for holomorphic Chern-Simons theory. 

As in the case of holomorphic Chern-Simons theory, the possible boundary modifications are controlled by distributional Dolbeault cocycles 
\begin{equation} 
 \mu^{2,2}  \in \br{\Omega}^{2,2}(\br{X}, \log D ).
 \end{equation}
We require that $\mu^{2,2}$ is supported on the $\CP^1$ given by $z = n = 0$. The modified boundary condition is given by modifying the equations of motion so that 
\begin{equation} 
\dbar \alpha^{2,1} = \mu^{2,2}. 
 \end{equation}
Since $\mu$ is given by a $\delta$-function and its derivatives, this equation is equivalent to specifying the poles of $\alpha^{2,1}$ along $z = n = 0$.  

For this equation to be consistent, we need $\mu^{2,2}$ to be $\dbar$-closed.  The field $\alpha^{2,1}$ also has the constraint that $\partial \alpha^{2,1} = 0$. This constraint can also be modified along a $\CP^1$ in the boundary, to give a constraint of the form
\begin{equation} 
\partial \alpha^{2,1} = \mu^{3,1} 
 \end{equation}
where 
\begin{equation} 
\mu^{3,1} \in \br{\Omega}^{3,1} (\br{X}, \log D)  
 \end{equation}
has support along the $\CP^1 \subset \br{X}$.  

For this modified boundary condition to be consistent, we need 
\begin{equation} 
\dbar \mu^{3,1} + \partial \mu^{2,2} = 0. 
 \end{equation}
As in the case of holomorphic Chern-Simons theory, $\mu = \mu^{3,1} + \mu^{2,2}$ should be taken up to gauge equivalence.  We thus find that the possible boundary modifications are given by cohomology classes of degree $1$ in the complex
\begin{equation} 
\Omega^{\ge 2,\ast}_{\CP^1}(\br{X}, \log D)  
 \end{equation}
with the differential $\dbar + \partial$.

We can calculate the cohomology group by first computing the cohomology of the $\dbar$-operator, dropping the term involving the Beltrami differential. We first work locally on $\CP^1$, and take Dolbeault cohomology in the normal direction. The resulting classes will transform as sections of some line bundle on $\CP^1$, and also will have a weight under $SO(2)$. We will write down the line bundles on $\CP^1$: 
\begin{align} 
	\d \log n \d z \partial_n^{(k)} \partial_z^{(l)} \delta_{z = n = 0} & \in \Oo(k+1)   \\ 
	\d \log n \d w \partial_n^{(k)} \partial_z^{(l)} \delta_{z = n = 0} & \in \Oo(k-1) \\
	 \d w \d z \partial_n^{(k)} \partial_z^{(l)} \delta_{z = n = 0} & \in \Oo(k-1)   \\ 
	\d \log n \d w \d z \partial_n^{(k)} \partial_z^{(l)} \delta_{z = n = 0} &\in  \Oo(k-1). 
 \end{align}

Note that there are no line bundles of degree $< -1$, so that all Dolbeault cohomology with coefficients in these bundles is in degree $0$.  (Incorporating the degree shifts from the $\delta$-functions, the total degree will always be $2$) .The spin under $SU(2)_R$ of  a class transforming in $\Oo(k)$ is $k/2$.  

Next let us incorporate the $\partial$ differential,
\begin{equation} 
\partial = ( \d \log n) n \partial_n + \d w \partial_w + \d z \partial_z. 
 \end{equation}
 This kills $\d \log n \d z \partial_n^{(k)} \partial_z^{(l)} \delta_{z = n = 0}$, and sends  
\begin{align} 
	\partial \left(\d \log n \d w \partial_n^{(k)} \partial_z^{(l)} \delta_{z = n = 0}\right)  & = \d n \d w \d z \partial_n^{(k)} \partial_z^{(l+1)} \delta_{z = n = 0}.   \\
	\partial \left( \d w \d z \partial_n^{(k)} \partial_z^{(l)} \delta_{z = n = 0} \right) &= (k+1) \d \log n \d w \d z \partial_n^{(k)} \partial_z^{(l)} \delta_{z = n = 0}.   
 \end{align}
From this we see that every $3$-form class is exact. 

Therefore a basis of the cohomology is provided by expressions like
\begin{equation} 
 \d \log n \d z \partial_n^{(k)} \partial_z^{(l)} \delta_{z = n = 0} 
 \end{equation}
which is of spin $(k+1)/2$ under $SU(2)_R$ and of spin $l+(k+1)/2$ under $SO(2)$, together with 
\begin{equation} 
  \d \log n \d w \partial_n^{(k)} \partial_z^{(l)} \delta_{z = n = 0} - \frac{1}{k+1}  \d w \d z \partial_n^{(k)} \partial_z^{(l+1)} \delta_{z = n = 0}  
 \end{equation}
This has spin $(k-1)/2$ under $SU(2)_R$ and spin $l + 2 + (k-1)/2$ under $SO(2)$.  

This gives us a perfect match with the operators of the large $N$ chiral algebra.  The first collection of operators corresponds to the $A$-tower of operators in the large $N$ chiral algebra, and the second collection to the $D$-tower.  The spins and $R$-charges match up perfectly. 

To find the boundary modifications corresponding to the $C$ and $D$ towers, we incorporate holomorphic Chern-Simons theory for $\Pi \C^2$.  As we have seen above, this gives us two towers of operators which transform in the spin $k/2$ representation of $SL_2(\C)_R$, and are of spin $k/2+r+1$ under $SO(2)$, for $k \ge 0$, $r \ge 0$.  This matches precisely the $B$ and $C$ towers of the large $N$ chiral algebra.

\subsection{The field sourced by a boundary modification in Kodaira-Spencer theory}
We can give a cohomological construction of the field sourced by a boundary modification, in analogy with what we did for holomorphic Chern-Simons theory. To do this, we note that there is an exact sequence of Dolbeault cohomology
\begin{equation} 
\dots \to H^i(\br{X}, \Omega^{\ge 2,\ast}_{\br{X}}(\log D) ) \to H^i(\br{X} \setminus \CP^1, \Omega^{\ge 2,\ast}_{\br{X}}(\log D) ) \to H^{i+1}_{\CP^1}(\br{X}, \Omega^{\ge 2,\ast} _{\br{X}} (\log D)) \to \dots  
 \end{equation}
 Here $H^i_{\CP^1}(-)$ indicates Dolbeault cohomology with support on a $\CP^1$ in the boundary.

As we have seen, $H^i(\br{X}, \Omega^{\ge 2,\ast}_{\br{X}}(\log D)) = \C$ if $i = 1$, and is zero for $i \neq 1$. Further $H^i_{\CP^1}(\Omega^{\ge 2,\ast}_{\br{X}} (\log D))$ is concentrated in degree $2$.  and the other cohomology groups vanish. We thus find an exact sequence
\begin{equation} 
0 \to \C \to  H^1(\br{X} \setminus \CP^1, \Omega^{\ge 2,\ast}_{\br{X}}(\log D) )  \to H^{2}_{\CP^1}(\br{X}, \Omega^{\ge 2,\ast}_{\br{X}} (\log D)) \to 0. 
 \end{equation}
This tells us that every boundary modification -- defined by an element of $H^{2}_{\CP^1}(\br{X}, \Omega^{\ge 2,\ast}_{\br{X}} (\log D))$ -- lifts to a solution to the equations of motion on $\br{X} \setminus \CP^1$, which satisfies the boundary conditions except along the $\CP^1$. 

This lift is not unique, because of the first term in the exact sequence. We are free to add on a closed $(2,1)$-form with logarithmic poles along $D$ whose residue has non-trivial period along a curve in the boundary. However, if we impose the additional constraint that $\int_{n = 0, z = 0 } \op{Res} \alpha^{2,1} = 0$, uniqueness is restored. 

We have shown that the space of boundary modifications of Kodaira-Spencer theory matches exactly with the single-trace elements of the large $N$ chiral algebra. Further, every boundary modification sources a unique field which satisfies the boundary conditions everywhere except at the chosen $\CP^1$ in the boundary.

The same statements hold when we couple $\mf{gl}(k \mid k)$ holomorphic Chern-Simons theory. In the $\lambda \to 0$ limit, the identity component of $\mf{gl}(k \mid k)$ cancels with one of the two towers of fermionic closed-string fields, leaving us with holomorphic Chern-Simons theory for $\mf{pgl}(k \mid k) \oplus \Pi \C$. The $\Pi \C$ factor gives the other tower of fermionic closed-string fields. This matches exactly with what happens in the planar limit of the large $N$ chiral algebra.

\section{Hamiltonian approach to boundary operators}
Above we described boundary operators in terms of modifications of the boundary condition. 
We will find it useful to discuss the other side of the state-operator map: the space of states of the bulk 
theory on a manifold analogous to the Wick rotation of global Lorentzian AdS$_3$, whose boundary is a 
$\C^*$. We will construct the space of states by geometrcic quantization of the phase space, and we will see that it matches the space of local modifications of the boundary conditions.

\subsection{The space of states}
The Wick rotation of global Lorentzian AdS$_3$ differs from global AdS$_3$ by the removal of 
the points $z=0$ and $z=\infty$ at the boundary, so that the boundary is a cylinder and the compactified 
geometry is a solid cylinder. Time-translations are given by the 
$L_0+ \bar L_0$ generator of the conformal group and one is interested in the space of states 
for a constant time slice $y^2 + |z|^2 = \mathrm{const}$. 

In six dimensions, we can remove the $\CP^1$'s at $z=0$ and $z=\infty$ at the boundary, i.e. 
the loci $U_1 = U_2 = 0$ and $W_1 = W_2 = 0$ in $\br{X}$ (recalling that $\br{X}$ is the projective variety defined by the homogeneous equation $\eps_{ij} U_i W_j = N V^2$).   We can denote the resulting manifold as 
$\widehat{X}$. We have an obvious projection $\widehat{X} \to \CP^1 \times \CP^1$
where $U_i$ are homogeneous coordinates on the first $\CP^1$ factor, $W_i$ on the second. 
The fibers are acted upon transitively by $L_0$. The image of the boundary is the diagonal in $D^1 \in \CP^1 \times \CP^1$.

We can now consider the space of linearized solutions of the KS or hCS equations of motion 
on this geometry.   We will find that this space has a symplectic form.  By quantizing, using the polarization given by the decomposition into positive and negative eigenspaces of $L_0$,  we will find the Hilbert space of the theory.

This is guarateed to give us the same space as we found when we studied boundary modifications. The variety $\what{X}$ is obtained from $\br{X}$ by removing the two $\CP^1$'s, $\CP^1_{z = 0}$ and $\CP^1_{z = \infty}$.  We get a cover of $\br{X}$ by saying $Y_0$ is the locus where we have removed $\CP^1_{0}$, and $Y_{\infty}$ is the locus where we have removed $\CP^1_{\infty}$. The intersection $Y_0 \cap Y_{\infty}$ is $\what{X}$. 

We can compute the space of linearized solutions to the equation of motion for holomorphic Chern-Simons theory (or KS theory) on $\br{X}$ by using a Mayer-Vietoris exact sequence in Dolbeault cohomology:
\begin{equation}
\dots \to H^i_{\dbar}(\br{X}, \Oo(-D)) \to  H^i_{\dbar}(Y_0,\Oo(-D)) \oplus H^i_{\dbar}(Y_{\infty}, \Oo(-D) ) \to  H^i_{\dbar}(\what{X}, \Oo(-D)) \to \dots. 
\end{equation} 
Since $H^i_{\dbar}(\br{X}, \Oo(-D)) = 0$, because of the absence of zero-modes on $\br{X}$, we find that $H^i_{\dbar}(\what{X}, \Oo(-D))$ is a direct sum of $H^i_{\dbar}(Y_0,\Oo(-D))$ and $H^i_{\dbar}(Y_{\infty}, \Oo(-D))$.  That is, the phase space is the direct sum of boundary modifications at $0$ and at $\infty$.

The positive eigenvalues of $L_0$ are at $0$, and the negative are at $\infty$. Quantizing using the polarization given by the eigenspaces of $L_0$ gives a Fock space which is given by products of boundary modifications at $0$.

The same argument applies to KS theory, as all we have used is the fact that there are no zero modes. We know this is the case for KS theory, as long as we impose a constraint about the vanishing of the integral of the residue of the $(2,1)$ form along a $\CP^1$ in the boundary.

For holomorphic Chern-Simons theory, we can write down a basis of solutions to the equations of motion explicitly. The boundary modifications at $z = 0$ which are in the kernel of $L_1$ are given by 
\begin{equation}
\mc{F}_{0,k}^+ =  V^{2+k} \bar U_{i_1} \cdots \bar U_{i_k} \frac{(\bar U,\d \bar U)}{(U,  \sigma \br{U})^{k+2}}
\end{equation}
(times some Lie algebra element $\mf{t}_a$).   Those at $z = \infty$ which are in the kernel of $L_{-1}$ are given by
\begin{equation}
\mc{F}_{0,k}^{-} =  V^{2+k} \bar W_{i_1} \cdots \bar W_{i_k} \frac{(\bar W,\d \bar W)}{(W,  \sigma \br{W})^{k+2}.}
\end{equation}
Since the global conformal $SL_2(\C)$ is represented by the vector fields
\begin{align}
L_1 &= -U_i \partial_{W_i} - \br{U}_i \partial_{\br{W}_i} \\ 
L_0 &= \tfrac{1}{2} \left( - U_i \partial_{U_i} - \br{U}_i \partial_{\br{U}_i} + W_i \partial_{W_i} + \br{W}_i \partial_{\br{W}_i}\right) \\ 
L_{-1} &= W_i \partial_{U_i} + \br{W}_i \partial_{\br{U}_i}
\end{align}
we find that a complete basis of solutions to the linearized equations of motion is given by
\begin{equation}
L_{-1}^n \mc{F}_{0,k}^+ \ \ L_{1}^n \mc{F}_{0,k}^{-}.
\end{equation}
The Fock space given by the positive eigenvalues of $L_0$ is spanned by polynomials in $L_{-1}^n \mc{F}_{0,k}^+$. This, of course, matches what we found studying boundary modifications.

\subsection{The symplectic form}
The symplectic form on the phase space for hCS theory is easy to define.  An element in the phase space is represented by $(0,1)$ form $A$ on $\what{X}$, with coefficients in $\mf{psl}(n \mid n)$.  The symplectic form is given by 
\begin{equation}
\ip{A,A'}  =\int \Omega_{\br{X}} \op{Tr} A A'. 
\end{equation}
where the integral is performed over the $5$-cycle which is the unit circle bundle in the $\C^\times$ fibration $\what{X} \to \CP^1 \times \CP^1$.

For KS theory, the symplectic form is given by the same formula, except we must use $\partial^{-1}$:
 \begin{equation}
\ip{\alpha,\alpha'}  =\int \alpha \partial^{-1} \alpha'
\end{equation}
Here $\alpha$, $\alpha'$ are $(2,1)$ forms and the integral is performed over the same $5$-cycle.

In the case of holomorphic Chern-Simons theory, we can perform the integral over the base of the $S^1$ fibration: 
\begin{equation}
\int_{\CP^1 \times \CP^1} \left(i_{L_0} \Omega\right) \wedge  A \wedge  A'
\end{equation}
The $(2,0)$ form $i_{L_0} \Omega$ can be written as 
\begin{equation}
i_{L_0} \Omega = \frac{4 \pi^2}{N} V^{-4} (U, \d U)(W,\d W) = 4 \pi^2 N \frac{(U, dU)(W,dW)}{(U,W)^2}
\end{equation}

We can compute 
\begin{equation}
\int_{\CP^1 \times \CP^1} \left(i_{L_0} \Omega\right) \wedge \mf{F}_{0,0}^+ \wedge \mc{F}_{0,0}^{-} =  
\frac{4 \pi^2}{N}  \int_{\CP^1} \frac{(W,\d W)(\bar W,\d \bar W)}{(W, \bar W)^2}\int_{\CP^1} \frac{(U,\d U)(\bar U,\d \bar U)}{(U, \bar U)^2}
\end{equation}
or more generally
\begin{multline}
\int_{\CP^1 \times \CP^1} \left(i_{L_0} \Omega\right) \wedge \mc{F}_{0,n}^+ \wedge \mc{F}_{0,n}^-   \\
        = \frac{4 \pi^2}{N^{1+n}}  \int_{\CP^1 \times \CP^1} \frac{\bar W_{i_1} \cdots \bar W_{i_n}\bar U_{j_1} \cdots \bar U_{j_n}(U,W)^{n}}{(U, \bar U)^n(W, \bar W)^n}\frac{(U,\d U)(\bar U,\d \bar U)}{(U, \bar U)^2} \frac{(W,\d W)(\bar W,\d \bar W)}{(W, \bar W)^2} 
\end{multline}

As we quantize the system, we will obtain a Fock space generated by the creation operators. 
Each of these correspond to one of the expected closed string operators in the chiral algebra.  The two point functions of the closed-string operators are determined by the commutator between a creation operator and the corresponding annihilation operator.  Thus, the two-point functions of the state $\mc{F}_{0,n}^+$ is of order $N^{-1-n}$. Since $\mc{E}_{0,n}$ is cohomologous to $N^{1+n} \mc{F}_{0,n}^+$, we find that it's two point function is of order $N^{1+n}$, a result which we will reproduce by Witten diagrams shortly.

\section{Witten diagrams}
In this section we will explain how to calculate two and three-point functions in the planar limit by using Witten diagrams.  We will not actually perform many such calculations, as they quickly become difficult. We will find that a more efficient computational method is to use the global symmetry algebra, which we have already shown matches on the CFT and the holographic side.  Our goal is rather to show that one can perform Witten diagram computations in principle; and to demonstrate that what we asserted is the holographic global symmetry algebra really is the global symmetries of the chiral algebra one builds using Witten diagrams.

\subsection{Two point functions for holomorphic Chern-Simons theory}
We will first compute Witten diagrams for holomorphic Chern-Simons theory. We will strip off the colour factors in this analysis. 

Suppose we have some different boundary modifications
\begin{equation} 
	\mu_i^{0,2} \in \br{\Omega}^{0,2}_{\CP^1 \times z_i} (\br{X},\Oo(-D)) 
 \end{equation}
 where $\CP^1 \times z_i$ are curves in the boundary $D = \CP^1_w \times \CP^1_z$.
We solve the equations of motion with the modified boundary conditions:
\begin{align} 
	A^{0,1}_i& \in \br{\Omega}^{0,1}(\br{X}, \Oo(-D))\\
	\dbar A^{0,1}_i &= \mu_i^{0,2}.
\end{align}
Explicit formulae for gauge representatives of $A^{0,1}_i$ are given in equation \eqref{eqn_boundary_soln_hcs}.

Then, the two-point function is given by 
\begin{align} 
	\ip{\mu_1(z_1) \mu_2(z_2)} &= \int_{X} A^{0,1}_1 \dbar A^{0,1}_2 \Omega_X \\
        &= \int_{X} A^{0,1}_1 \mu^{0,2}_2 \Omega_X. 
\end{align}
Note that, by construction, $\mu^{0,2}_2$ is a $\delta$-function supported near a point $z_i$ in the boundary, so that this integral reduces to an evaluation of $A^{0,1}_1$ at $z_2$. 

The three-point function is
\begin{equation} 
	\ip{\mu_1(z_1)\mu_2(z_2) \mu_3(z_3)} = \int_{X} A^{0,1}_1 \wedge A^{0,1}_2 \wedge A^{0,1}_3 \Omega_X. 
\end{equation}

How do we compute these? Let us do the example where both boundary modifications have no $z$-derivatives and live in the trivial representation of $SU(2)_R$. In this case, we can take \begin{equation} A^{0,1}_i =  
	\mc{E}_{0,0} =    \delta_{z = z_i} + n^{2}  \frac{N}{ 2 \pi \i} \frac{ \d \wbar  }{ (1 + \abs{w}^2)^{2}   } \frac{1}{(z-z_i)^{2}} 
\end{equation}
Recall that the boundary conditions require the fields to be divisible by $n$. The term $\delta_{z = z_i}$ fails to satisfy the boundary conditions, and has a first-order pole when viewed as a section of the line bundle $\Oo(-D)$. Therefore have 
\begin{equation} 
	\mu_i = \dbar A^{0,1}_i = \delta_{n = 0} \delta_{z = z_i}. 
\end{equation}
Here we must take a little care: we view $\delta_{n = 0}$ as a distributional section of $\Oo(-D)$.

The holomorphic volume form is, up to a constant,
 \begin{equation}
\Omega =  n^{-3} \d n \d w \d z  - N n^{-1} \d n \d w \d \wbar \frac{1}{(1 + \abs{w}^2)^2}. 
 \end{equation}
  So the two-point function is
\begin{equation} 
	\int_{n,w,z} \Omega  A^{0,1}_{1} \dbar A^{0,1}_2 .
\end{equation}
Note that the terms in $A^{0,1}_1$ and $A^{0,1}_2$ which are not localized at $z = z_i$ both involve $\d \wbar$.  Therefore only the $n^{-3} \d n\d w \d z$ term in the holomorphic volume form contributes to the integral.

Since $\dbar A^{0,1}_2$ is supported near $n = 0$, $z = z_2$, we should only integrate $n,z$ in a neighbourood of these values.  This means that we should only keep the coefficient of $n^2$ in $A^{0,1}_1$, as the coefficient of $n^0$ is supported near $z = z_1$.  The integral becomes
\begin{equation} 
  \int_{\abs{n} \le \eps, \abs{z - z_2} \le \eps,w}  n^{-3} \d n \d w \d z     n^2 \frac{N}{2 \pi \i}  \frac{ \d \wbar  }{ (1 + \abs{w}^2)^{2}   } \frac{1}{(z-z_1)^{2}} \dbar \delta_{z = z_2} 
\end{equation}
Note that here the $\dbar$ operator can only apply in the $n$ direction, as we already have a $\d \wbar$ and a $\d \zbar$ (present in $\delta_{z = z_2}$).  For the purposes of integrating over the $z$ and $w$ planes, the $\dbar$ plays no role. 

Because of this, we can write $\dbar (\delta_{z = z_2})$ as $\dbar(1) \delta_{z = z_2}$, recalling that the $\dbar$ operator is that acting on sections of $\Oo(-D)$. Integrating over $z,w$ yields 
\begin{equation} 
	\op{Vol}(\CP^1_w)\frac{ N}{2 \pi \i} \frac{1}{(z_1 - z_2)^2}  \int_{{\abs{n} \le \eps}} n^{-1} \d n \dbar 1 
\end{equation}
Since $1$ is viewed as a section of the line bundle $\Oo(-D)$ of functions vanishing at $n = 0$, it has a pole at $n = 0$. Sections of $\Oo(-D)$ pair naturally with $1$-forms such as $n^{-1} \d n$ with a pole at $n = 0$, as in this integral.

By integrating by parts and picking up the boundary term the integral becomes
\begin{equation} 
\op{Vol}(\CP^1_w)  \frac{ N}{2 \pi \i} \frac{1}{(z_1 - z_2)^2}  \oint_{\abs{n} = 1} n^{-1} \d n.
 \end{equation}

We have found that the two-point function is
\begin{equation} 
 	\ip{\mu_1(z_1) \mu_2(z_2)} \simeq \frac{N}{(z_1 - z_2)^2} 
 \end{equation}
where we have dropped factors of $\pi$, etc. 

One can, in the same way, compute the $2$-point function for the other boundary modifications of holomorphic Chern-Simons theory.  We can take
\begin{equation} 
A^{0,1}_1 = \mc{E}_{k,0} 
 \end{equation}
where as before
\begin{multline} 
	\mc{E}_{k,0}(z_1) =   n^{-k} \delta_{z = z_1} + \sum_{s = 1}^{k} \frac{1}{s!} n^{2s - k} (-1)^s  N^s \frac{\wbar^s  }{ (1 + \abs{w}^2)^s   }  \delta_{z = z_1}^{ (s)} \\
	+ n^{k+2} \frac{(k+1)!}{k!}  \frac{N^{k+1}}{ 2 \pi \i} \frac{\wbar^k \d \wbar  }{ (1 + \abs{w}^2)^{k+2}   } \frac{1}{(z-z_1)^{k+2}} 
 \end{multline}
This field is the highest weight vector in a representation of $SU(2)_R$ of spin $k/2$. To get a non-zero two-point function, we should pair it with the lowest weight vector which is
\begin{equation} 
 A^{0,1}_2 = \mc{L}_{\partial_{\wbar} - w^2 \partial_{w}}^k \mc{E}_{k,0}(z_2)
 \end{equation}
 where we bear in mind that we are applying the Lie derivative and that $n$ transforms as a half-density $(\d w)^{k/2}$. 

By following the analysis for the case $k = 0$, we find that the integral
\begin{equation} 
\int n^{-3} \d n \d w \d z A^{0,1}_1 \dbar A^{0,1}_2 
 \end{equation}
reduces to an integral over the region where $n$ is near $0$, and $z$ is near $z_2$.  Thus only the coefficient of $n^{k+2}$ in $A^{0,1}_1$ plays a role. Similarly, only the coefficient of $n^{-k}$ in $A^{0,1}_2$ plays a role, because we need to get an overall $n^{-1}$ to find a non-zero answer. We find we are computing (dropping various factors of $\pi$) 
\begin{equation} 
N^{k+1} \frac{1}{(z_1 - z_2)^{k+2} }  \int_{n,w} n^{k-1} \d n \d w  \frac{\wbar^k \d \wbar  }{ (1 + \abs{w}^2)^{k+2}} \dbar \mc{L}_{w^2 \partial_w}^k n^{-k}. 
 \end{equation}
 Recalling that $n$ transforms as $(\d w)^{1/2}$, $\mc{L}_{w^2 \partial_w} n^{-1} =- w n^{-1}$, so that the integral becomes (up to an overall constant)
 \begin{equation} 
 N^{k+1} \frac{1}{(z_1 - z_2)^{k+2} }  \int_{n,w} n^{k-1} \d n \d w  \frac{\wbar^k \d \wbar  }{ (1 + \abs{w}^2)^{k+2}} w^k  \dbar n^{-k}. 
  \end{equation}
  The integral over the $w$-plane is the convergent expression
  \begin{equation} 
  \int \d w \d \wbar \frac{\norm{w}^{2k} }{ (1 + \norm{w}^2)^{k+2}}. 
   \end{equation}
   The remaining integral over the disc $\abs{n} \le 1$ can be evaluated by integration by parts as before, to show that the $2$-point function is proportional to 
   \begin{equation} 
   N^{k+1} (z_1 - z_2)^{-k-2} . 
    \end{equation}
This matches the two-point function of the corresponding operators in the large $N$ chiral algebra.

Note that the calculations given above apply also to the holographic dual of the $C$ and $D$ towers of the large $N$ chiral algebra.

Unfortunately, we did not see an easy way to compute the two-point functions of Kodaira-Spencer theory using Witten diagrams (except for the fermionic operators which match  the $C$ and $D$ towers). The reason is that the kinetic term of Kodaira-Spencer theory is $\int \alpha \dbar \partial^{-1} \alpha$, and the presence of $\partial^{-1}$ complicates the analysis.  However, this is not a serious problem, because these two-point functions are determined from the OPE coefficients and the two-point functions of the $C$ and $D$ towers. 

\subsection{OPE coefficients and three point functions}
In principle, we could calculate the three-point function by analyzing the integral
\begin{equation} 
\int_{SL_2(\C)} A_1^{0,1} A_2^{0,1} A_3^{0,1} \Omega 
 \end{equation}
 where $A_i^{0,1}$ are forms sourced by modifications of the boundary condition.  

 It turns out to be simpler to compute the OPE coefficients by a slightly different technique.  The $3$-point functions are determined from the OPE coefficients and the $2$-point functions. 

 Let us explain, in general, how to compute the OPE coefficients.  We will describe the computation in a way that is independent of gauge. If we choose a gauge, we will find an answer equivalent to that given by Witten diagrams.

 Suppose that $A_{1,a_1}(z_1)$, $A_{2,a_2}(z_2)$ are two fields of holomorphic Chern-Simons theory which satisfy the linearized equations of motion $\dbar A = 0$, and satisfy the boundary condition except along the curve $z = z_i$ in the boundary.  The three-point function with any field $A_{3,a_3}(z_3)$ is given by 
 \begin{equation} 
\int f^{a_1 a_2 a_3} A_{1,a_1}(z_1)   A_{2,a_2}(z_2)   A_{3,a_3}(z_3)\Omega . 
  \end{equation}
The OPE coefficients will be defined by constructing some field  $A_{12,c}(z_1,z_2)$ whose two-point function with $A_{3,a_3}$ will be given by the three-point function above. That is, we need
\begin{equation} 
 \int g^{c,a_3}  \dbar A_{12,c} (z_1,z_2)  A_{3,a_3}\Omega  =   \int f^{a_1,a_2,a_3} A_{1,a_1}(z_1)   A_{2,a_2}(z_2)   A_{3,a_3}(z_3)\Omega . 
 \end{equation}
 (Here $g$ is the inner product on the Lie algebra and $f$ are the structure constants).  For this to hold, we must have
 \begin{equation} 
 \dbar A_{12,c}(z_1,z_2) = f^{a_1a_2}_c A_{1,a_1}(z_1) A_{2,a_2}(z_2). \label{eqn_holographic_ope1} 
  \end{equation}

To find the OPE, we need to take any solution to this equation and then expand in series in $(z_1 - z_2)^{-1}$.  Note that equation \eqref{eqn_holographic_ope1} is equivalent to the statement that 
\begin{equation} 
\eps_1 A_{1} (z_1) + \eps_2 A_{2}(z_2) - \eps_1 \eps_2 A_{12}(z_1,z_2)\label{eqn_holographic_ope2} 
 \end{equation}
solves the equations of motion for holomorphic Chern-Simons theory modulo terms involving $\eps_1^2$ or $\eps_2^2$.  

The terms in equation \eqref{eqn_holographic_ope2} linear in $\eps_i$ have no pole at $z_1 = z_2$. It follows that the polar part of $A_{12}(z_1,z_2)$ must satisfy the linearized equations of motion $\dbar A_{12}(z_1,z_2) = 0$, which is what we need to define a modification of the boundary condition.

This is a general prescription for calculating OPE coefficients holographically. A boundary operator is given by a solution of the linearized field equations which fails to satisfy the boundary conditions at a point on the boundary.  The OPE of two such boundary operators is obtained by finding solutions to the field equations as in \eqref{eqn_ope}, whose terms linear in $\eps_i$ are the two solutions of the linearized field equations whose OPE we are computing.  The singular part of the coefficient of $\eps_1 \eps_2$ gives the OPE coefficients.

In the example of holomorphic Chern-Simons theory, we can compute the OPE coefficients explicitly.  Let us do this for the example when 
\begin{equation} 
	A_i =  
	\mf{t}_{a_i} 	\mc{E}_{0,0} (z_i) = \mf{t}_{a_i}   \delta_{z = z_i} + \mf{t}_{a_i} n^{2}  \frac{N}{ 2 \pi \i} \frac{ \d \wbar  }{ (1 + \abs{w}^2)^{2}   } \frac{1}{(z-z_i)^{2}} 
\end{equation}
where $\mf{t}_{a}$ is a basis of the Lie algebra $\mf{gl}(n \mid n)$.  Note that 
\begin{equation} 
\mc{E}_{0,0}(z_i) =  \frac{1}{2 \pi \i} \dbar \frac{1}{z - z_i}. 
 \end{equation}
Here we use the full $\dbar$ operator of the deformed conifold, including the term from the Beltrami differential.  

Then, we see that 
\begin{equation} 
\eps_1 \mc{E}_{0,0}(z_1) \mf{t}_{a_1} + \eps_2 \mc{E}_{0,0}(z_2) \mf{t}_{a_2} + \eps_1 \eps_2 [ \mf{t}_{a_1}, \mf{t}_{a_2}] \mc{E}_{0,0}(z_1) \frac{1}{2 \pi \i} \frac{1}{z - z_2}   
 \end{equation}
satisfies the equations of motion modulo $\eps_1^2$, $\eps_2^2$.  We conclude that the OPE coefficients are given by the polar part of 
\begin{multline} 
- [ \mf{t}_{a_1}, \mf{t}_{a_2}] \mc{E}_{0,0}(z_1) \frac{1}{2 \pi \i} \frac{1}{z - z_2} 
\\= -[ \mf{t}_{a_1}, \mf{t}_{a_2}]\left( \frac{1}{2 \pi \i} \frac{1}{z_1 - z_2} \delta_{z = z_1} + n^2  \frac{N}{2 \pi \i} \frac{ \d \wbar  }{ (1 + \abs{w}^2)^{2}   } \frac{1}{(z-z_1)^{2}(z - z_2)}  \right) \label{eqn_ope} 
 \end{multline}
 The part of this expression that violates the boundary condition is $ -[ \mf{t}_{a_1}, \mf{t}_{a_2}]\left( \frac{1}{2 \pi \i} \frac{1}{z_1 - z_2} \delta_{z = z_1}\right) $.  We can conclude formally that if we take the coefficient of $n^2$, and take its polar part as a function of $z_1 - z_2$, that this must be 
 \begin{equation} 
   -[ \mf{t}_{a_1}, \mf{t}_{a_2}]n^2 \frac{1}{2 \pi \i (z_1 - z_2)}  \frac{N}{ 2 \pi \i} \frac{ \d \wbar  }{ (1 + \abs{w}^2)^{2}   } \frac{1}{(z-z_1)^{2}}  
  \end{equation}
The point is that  the coefficient of $n^2$ is characterized by the fact that together with the coefficient of $n^0$, the linearized equations of motion are satisfied.

We can also see this quite explicitly.  Let us change coordinates, and set $z_1 = -s$, and write $z_2=s$.  We then expand in series in $1/s$.  We are computing the polar part in $s$ of $(\dbar (z+s)^{-1} ) (z-s)^{-1}$. Up to a gauge transformation by $(z+s)^{-1} (z-s)^{-1}$, this is anti-symmetric under $s \to -s$.   

To extract the polar part of the coefficient of $n^2$ we should compute the integral
\begin{equation} 
\oint_s s^k \d s \frac{1}{(z+s)^2 (z - s)}. 
 \end{equation}
We compute this integral around a very large circle, and we only care about $k$ even because of anti-symmetry under $s \mapsto -s$. We find that the leading order term in the expansion in $1/s$ is $s^{-1} z^{-2}$, which is what we want.  The sub-leading terms are of the form $s^{-3} + s^{-5} z^2 + \dots$.  These sub-leading terms contribute field configurations which are regular everywhere except at $z = \infty$, and so do not correspond to modifications of the boundary condition at $z = 0$.   

To sum up, we find that the polar part in the expansion in powers of $z_1 - z_2$ of \eqref{eqn_ope} is 
 \begin{equation} 
   -[ \mf{t}_{a_1}, \mf{t}_{a_2}] \frac{1}{2 \pi \i} \frac{1}{z_1 - z_2} \mc{E}_{0,0}(z_i). 
 \end{equation}
These operators correspond to the $\mf{gl}(k \mid k)$ currents. We find that they have the standard OPE for the $\mf{gl}(k \mid k)$ current algebra.

\subsection{OPE coefficients for Kodaira-Spencer theory}
One can perform a similar analysis to obtain the OPE coefficients for Kodaira-Spencer theory.  These contain the stress-energy tensor and their superconformal cousins, among other operators. Let us explain how to check that $T T \sim z^{-1} \partial T + z^{-2} T + c z^{-4} $ (the central charge $c$ is given by the two-point function and we will not calculate it directly).

Any chiral CFT can be coupled to a background Beltrami differential, governing deformations of the complex structure of the two-dimensional space-time.  The stress-energy tensor is the operator given by the background Beltrami differential $\delta_{z= 0} \partial_z$.  

Field configurations for Kodaira-Spencer include divergence free Beltrami differentials on $SL_2(\C)$.  The stress-energy tensor should be given by such a Beltrami differential which, near the boundary, looks like
\begin{equation} 
\delta_{z = 0} \partial_z  + \delta^{(1)}_{z = 0} \tfrac{1}{2} n \partial_n + O(n^2). 
 \end{equation}
The term involving $n \partial_n$ is included to ensure that the Beltrami differential is divergence free.   The higher-order terms in $n$ are needed to make sure that this Beltrami differential is $\dbar$-closed. 

We can write the required Beltrami differential with a pole at $z_0$, including all higher-order terms, in a succinct way as
\begin{align} 
\mc{T} (z_0) =& \dbar \left( (z-z_0)^{-1} \partial_{z=z_0} - \tfrac{1}{2} (z-z_0)^{-2} n \partial_n \right) \\
=& \delta_{z=z_0} \partial_z + \tfrac{1}{2}\delta^{(1)}_{z = z_0} n \partial_n + N \frac{n^2}{(z-z_0)^2}  \frac{ \d \wbar  }{ (1 + \abs{w}^2)^{2}} \\
& + 2 N^2 \frac{n^3}{(z-z_0)^3} \frac{ \d \wbar  }{ (1 + \abs{w}^2)^{2}  }. 
 \end{align}
 In the second and third lines we have dropped factors of $2$ and $\pi$ to keep the formulae to a reasonable length.

To compute the OPE coefficient of $\mc{T}(z_1)$ with $\mc{T}(z_2)$, we should try to find a solution to the equations of motion of the form
\begin{equation} 
\eps_1 \mc{T}(z_1) + \eps_2 \mc{T}(z_2) + \eps_1 \eps_2 \mc{B}(z_1,z_2)  
 \end{equation}
 As before, we work modulo $\eps_i^2$.  We then expand $\mc{B}(z_1,z_2)$ in series in $(z_1 - z_2)^{-1}$ to find the OPE coefficients.

The equations of motion state that
\begin{equation} 
\dbar \mc{B}(z_1,z_2) + [\mc{T}(z_1), \mc{T}(z_2) ] = 0. 
 \end{equation}
Here $[-,-]$ is given by combining the commutator of $(1,0)$ vector fields on a complex manifold with the wedge product of $(0,1)$ forms.

We have
\begin{equation} 
[\mc{T}(z_1), \mc{T}(z_2) ] =   \dbar \left[ (z- z_1)^{-1} \partial_z - \tfrac{1}{2} (z-z_1)^{-2} n \partial_n, \mc{T}(z_2)   \right]. 
 \end{equation}
 We can thus take 
 \begin{align} 
 \mc{B}(z_1,z_2) =& \left[ (z- z_1)^{-1} \partial_z - \tfrac{1}{2} (z-z_1)^{-2} n \partial_n, \mc{T}(z_2)   \right] \\
 =&  \left[ (z - z_1)^{-1} \partial_z - \tfrac{1}{2} (z-z_1)^{-2} n \partial_n , \delta_{z=z_2} \partial_z + \tfrac{1}{2}\delta^{(1)}_{z = z_2} n \partial_n \right] \\
 & + O(n^2). 
  \end{align}
  The terms of order $n^2$ and higher are non-singular at the boundary and their polar parts in $(z_1 - z_2)$ are fixed by the requirement that the whole polar part  satisfies the equations of motion.

Computing the commutator of the vector fields appearing in the expression for $\mc{B}(z_1,z_2)$,  we find 
\begin{multline} 
\mc{B}(z_1,z_2) = (z_2 - z_1)^{-2} \delta_{z = z_2} \partial_z + (z - z_1)^{-1} \delta_{z = z_2}^{(1)} \partial_z \\  + \tfrac{1}{2} (z - z_1)^{-1} \delta_{z = z_2}^{(2)} n \partial_n- (z_2 - z_1)^{-3}\delta_{z = z_2}  n \partial_n + O(n^2). 
 \end{multline}
By noting that 
\begin{equation} 
 f(z) \delta^{(l)}_{z = z_2} =\sum (-1)^{m} {{l}\choose{m}} f^{(m)}(z_z) \delta^{(l-m)}_{z = z_2} 
 \end{equation}
 we can simplify the expression for $\mc{B}(z_1,z_2)$ to
 \begin{align} 
 \mc{B}(z_1,z_2)  =& 2 (z_2 - z_1)^{-2} \delta_{z = z_2} \partial_z  + (z_2 - z_1)^{-1} \delta^{(1)}_{z = z_2} \partial_z \\ 
  &+  (z_2 - z_1)^{-2} \delta^{(1)}_{z = z_2} n \partial_n    + \tfrac{1}{2}(z_2 - z_1)^{-1} \delta^{(2)}_{z = z_2} n \partial_n  + O(n^2) \\
=& 2 (z_2 - z_1)^{-2} \mc{T}(z_2) + (z_2 - z_1)^{-1} \partial \mc{T}(z_2)  + O(n^2). 
  \end{align}
which is exactly the usual OPE of the stress tensor, except for the central extension term determined by the two-point function.

\section{The mode algebra}
\label{sec:modes}
The mode algebra of the holographic chiral algebra is defined in a very similar way to the algebra of local operators. In any ordinary vertex algebra, the mode algebra is the associative algebra generated by integrals of modes of local operators around a circle.   In the holographic chiral algebra, similarly, the mode algebra is generated by modifications of the boundary condition localized at $\abs{z} = 1$ in the boundary.  These boundary modifications can be realized as integrals of point-like modifications around the circle $\abs{z} = 1$.

The mode algebra of the holographic chiral algebra has a special feature: it is the universal enveloping algebra of the central extension of a Lie algebra.  As such, it can be described by specifying three pieces of data.
\begin{enumerate} 
	\item The vector space underlying the Lie algebra. These are the boundary modifications localized on $\abs{z} = 1$.
	\item The Lie bracket.  These are encoded by OPE coefficients, and can be calculated in a similar way.
	\item The central extension. This is given by $2$-point functions.
\end{enumerate}

Since the analysis here is so similar to that given for local operators, we will be brief.  If we take a field of holomorphic Chern-Simons theory  such as 
\begin{equation} 
\mf{t}_a	\mc{E}_{0,0} (s) =\mf{t}_a  \delta_{z = s} +\mf{t}_a  n^{2}  \frac{N}{ 2 \pi \i} \frac{ \d \wbar  }{ (1 + \abs{w}^2)^{2}   } \frac{1}{(z-s)^{2}} 
 \end{equation}
 which satisfies the equation of motion and violates the boundary condition at $z = s$, we can average it over the circle $\abs{s} = 1$ to get a field configuration
 \begin{equation} 
\mf{t}_a \oint_{\abs{s} = 1} s^k \d s \mc{E}_{0,0}(s) =\mf{t}_a z^k \delta_{\abs{z} = 1} + \mf{t}_a n^{2}  \frac{N}{ 2 \pi \i} \frac{ \d \wbar  }{ (1 + \abs{w}^2)^{2}   }\oint_{\abs{s} = 1} \d s s^k  \frac{1}{(z-s)^{2}} 
  \end{equation}
  This violates the boundary condition only at $\abs{z} = 1$. 

A field configuration of this form (or defined more generally using $\mc{E}_{m,0}(s)$) defines a generator of the mode algebra for holomorphic Chern-Simons theory.

We will explain how to calculate the Lie bracket and central extension of such elements of the mode algebra, focusing on the simplest case $\mc{E}_{0,0}(s)$.  Similar calculations apply to Kodaira-Spencer theory.

Let us compute the commutator between 
\begin{equation} 
M_1(C) =  \oint_{\abs{s} = C} s^k \mf{t}_a \d s \mc{E}_{0,0}(s) 
 \end{equation}
and
\begin{equation} 
M_2(C') =  \oint_{\abs{s} = C'} s^l \mf{t}_b \d s \mc{E}_{0,0}(s) 
 \end{equation}
Note that these modes violate the boundary condition only along the circle $\abs{z} = C$ or $\abs{z} = C'$. 

The multiplication of two elements of the mode algebra is given by taking the composite operator where $M_2$ is placed to the right of $M_1$: this is $M_1(1) M_2(1 + \eps)$.  The multiplication in the other order is $M_1(1) M_2 (1 - \eps)$.  We are interested in the commutator
\begin{equation} 
M_1(1) M_2(1 + \eps) - M_1(1) M_2 ( 1 - \eps). 
 \end{equation}
There is a gauge transformation $\chi$ such that 
\begin{equation} 
\dbar \chi = M_2(1 + \eps) - M_2 ( 1 - \eps). 
 \end{equation}

We can apply the gauge transformation $\chi$ in the presence of the field $M_1(1)$: 
\begin{equation} 
\dbar_{M_1(1)} \chi =  M_1(1) M_2(1 + \eps) - M_1(1) M_2 ( 1 - \eps) + [M_1(1), \chi].  
 \end{equation}
 Since this is a gauge-trivial variation of the field $M_1(1)$, we conclude that
 \begin{equation} 
   M_1(1) M_2(1 + \eps) - M_1(1) M_2 ( 1 - \eps) = - [M_1(1), \chi].  
  \end{equation}
This allows us to compute the commutator in the mode algebra using the gauge transformation $\chi$.

The gauge transformation $\chi$ moving an element from $\abs{z} = 1 + \eps$ to $\abs{z} = 1-\eps$ is determined by the formula
\begin{equation} 
 \oint_{\abs{s} = 1+ \eps}  s^l \mf{t}_b \d s \mc{E}_{0,0}(s) -  \oint_{\abs{s} = 1+ \eps}  s^l \mf{t}_b \d s \mc{E}_{0,0}(s) = \dbar \int_{1 - \eps \le \abs{s} \le 1 + \eps} s^l \mf{t}_b \d s \frac{1}{2 \pi \i} \frac{1}{s - z}  
  \end{equation}
We can compute the contour integral on the right hand side by observing that there is a single simple pole at $s = z$ with residue $z^l$. So we find 
\begin{equation} 
  \oint_{\abs{s} = 1+ \eps}  s^l \mf{t}_b \d s \mc{E}_{0,0}(s) -  \oint_{\abs{s} = 1+ \eps}  s^l \mf{t}_b \d s \mc{E}_{0,0}(s) =  \dbar  \delta_{1 - \eps \le \abs{z}  \le 1 + \eps } z^l \mf{t}_b. \label{eqn_mode_gauge} 
 \end{equation}
Thus, the commutator of the two elements $M_1,M_2$ of the mode algebra is given by
\begin{equation} 
[M_1,M_2] = - [\delta_{1 - \eps \le \abs{z}  \le 1 + \eps } z^l \mf{t}_b,  \oint_{\abs{s} = 1} s^k \mf{t}_a \d s \mc{E}_{0,0}(s)]. 
 \end{equation}
This is 
\begin{equation} 
- \oint_{\abs{s} = 1} s^{k+l} f_{ab}^c \mf{t}_c \delta_{z = s} + O(n^2), 
 \end{equation}
where the subleading term is fixed by the equations of motion. 

We find that the leading term in the expansion of powers of $n$ is indeed the $k+l$ mode of the operator corresponding to $\mc{E}_{0,0}$, as desired.

To compute the central extension, we can use a similar trick to our computation of the two-point function.  If $M_1(C)$, $M_2(C')$ are two field configurations which violate the boundary condition along $\abs{z} = C,C'$  the central extension of the Lie algebra is defined by the integral
\begin{equation} 
\int \Omega M_1(C) \dbar M_2(C')  - \int \Omega M_1(C') \dbar M_2(C). 
 \end{equation}
We will take, as before, 
\begin{align} 
M_1(C) &=  \oint_{\abs{s} = C} s^k \mf{t}_a \d s \mc{E}_{0,0}(s) \\
M_2(C') &=  \oint_{\abs{s} = C'} s^l \mf{t}_b \d s \mc{E}_{0,0}(s) 
 \end{align}
We assume $k \ge 0$.

To compute the central extension term, we will first simplify the expression
\begin{equation} 
M_1(1) =\mf{t}_a s^k  \delta_{\abs{z} = 1} +\mf{t}_a  n^{2}  \frac{N}{2 \pi \i} \frac{ \d \wbar  }{ (1 + \abs{w}^2)^{2}   }\oint_{\abs{s} = 1}s^k \d s   \frac{1}{(z-s)^{2}}. 
 \end{equation} 
 Using the residue theorem, we find that if $k \ge 0$ this becomes
\begin{equation} 
  z^k \delta_{\abs{z} = 1} - \mf{t}_a n^{2}  \frac{N}{ 2 \pi \i} \frac{ \d \wbar  }{ (1 + \abs{w}^2)^{2}   }  k \delta_{\abs{z} \le 1} z^{k-1}.  
 \end{equation}

As in the analysis of the two-point function, we can compute the central extension as the integral
\begin{equation} 
\int_{z,n,w} n^{-3} \d n \d z \d w \left( M_1(1) \dbar M_2(1 + \eps) - M_1(1 + \eps) \dbar M_2(1) \right).  
 \end{equation}
We assume $k \ge 0$, so that $M_1(1)$ is localized at $\abs{z} \le 1$ and $\dbar M_2(1+\eps)$ is localized at $\abs{z} = 1+\eps$, $n=0$, and only the second term contributes.  The integrand is zero outside  a neighbourhood of $\abs{z} = 1$, $n = 0$. The central extension is, dropping various powers of $\pi$, 
\begin{equation} 
N \op{Vol}(\CP^1_w)  \int_{\abs{z} = 1,\abs{n} \le \eps}  n^{-1} \d n \d z k z^{k-1}   \dbar z^l. 
 \end{equation}
 The $\dbar$ operator only appplies in the $n$ direction, so that by the same argument we used when we compute the two-point function, the integral becomes
 \begin{equation} 
 N \op{Vol}(\CP^1_w) \int_{\abs{z} = 1, \abs{n} = \eps} n^{-1} \d n  k z^{k+l-1} \d z   
  \end{equation}
This is proportional to $N k \delta_{k+l}$, which is exactly what we wanted.  The factor of $k$ is the standard one that appears in Kac-Moody central extensions.

\subsection{The holographic global symmetry algebra in the mode algebra}
We have explained how to define and compute the mode algebra, which is the universal enveloping algebra of the central extension of a Lie algebra.  Next we will show that there is a natural  subalgebra of the mode algebra which preserves the vacuum at zero and infinity, and on which the central extension vanishes. This is the global symmetry algebra we discussed before.

Let us define it for holomorphic Chern-Simons theory. Let $F \in \mf{gl}(k \mid k) \otimes \Oo(SL_2(\C))$ be an infinitesimal gauge transformation which preserves the vacuum field configuration $A = 0$.  Let us extend the function $z$, which is a coordinate on the chiral algebra plane on the boundary, to a global continuous function on $SL_2(\C)$. This function is not, of course, holomorphic.  Then, $\delta_{\abs{z} = 1}$ defines a $1$-form on $SL_2(\C)$; we view it as a $(0,1)$ form by the projection map. 

Then, we can define a field configuration
\begin{equation} 
\delta_{\abs{z} = 1} F \in \Omega^{0,1}(SL_2(\C), \mf{gl}(k \mid k) ). 
 \end{equation}
 This satisfies the linearized equations of motion, because $F$ is holomorphic and\footnote{It is important to note that $z$ does not have to be a holomorphic function for $\dbar \delta_{\abs{z} = 1}$ to vanish.  This holds for any continuous complex-valued function, as a consequence of the stronger fact that, when viewed as a $1$-form and not a $(0,1)$-form, $\d \delta_{\abs{z} = 1} = 0$.  } $\dbar \delta_{\abs{z} = 1} = 0$. 

Near the boundary, this field configuration vanishes except on the circle $\abs{z} = 1$. In particular, it is compatible with the boundary condition except along this circle, and so defines an element of the mode algebra. An element of the mode algebra of this form is a global symmetry. 

Next, let us show that the central extension vanishes on global symmetries. If $M_i(C_i)$ are field configurations which violate the boundary condition at $\abs{z} = C_i$, the central extension involves the value of $M_1(C_1)$ along the circle $\abs{z} = C_2$ in the boundary (and conversely).  If $M_1, M_2$ are both global symmetries, then $M_i(C_i)$ are exactly zero away from $\abs{z} = C_i$, so there is no central extension.

To compute the commutator, note that 
\begin{equation} 
\dbar \delta_{1 - \eps \le \abs{z} \le 1 + \eps } F = \delta_{\abs{z} = 1 + \eps} F - \delta_{\abs{z} = 1 - \eps} F.  
 \end{equation}
We have seen that we can compute the commutator between two elements of the mode algebra by using a gauge transformation to move one past the other.  Applying this to global symmetries, we find that the commutator between $\delta_{\abs{z} =1} F$ and $\delta_{\abs{z} = 1} G$ is given by the commutator $[F,G]$ in the Lie algebra of infinitesimal gauge transformations.

We have shown that $\Oo(SL_2(\C)) \otimes \mf{gl}(k \mid k)$ lies inside the mode algebra as a sub-Lie algebra on which the central extension vanishes.  To check that these are indeed global symmetries, we need to verify that these elements preserve the vacuum at $0$ and $\infty$.  

Applying the element of the mode algebra corresponding to $F \in \Oo(SL_2(\C)) \otimes \mf{gl}(k \mid k)$ to the vacuum at $z = 0$ will yield a field configuration localized near $z = 0$.  This field configuration is $\delta_{\abs{z} = \eps} F$.  To show that this field configuration corresponds to the trivial operator, we need to show that we can remove it by a gauge transformation which only violates the boundary conditions near $z = 0$.  The required gauge transformation is $\delta_{\abs{s} \le \eps} F$. 

In a similar way, these elements of the mode algebra preserve the vacuum at $\infty$, and so deserve to be called global symmetries. 

We have seen that the Lie algebra of global symmetries for the holographic chiral algebra build from holomorphic Chern-Simons theory is\footnote{Strictly speaking, we have not verified that there can not be any other global symmetries.  This is not too difficult to check and we will not elaborate on this point.}  $\Oo(SL_2(\C)) \otimes \mf{gl}(k \mid k)$.  In a similar way, the global symmetries coming from the bosonic fields of Kodaira-Spencer theory contribute $\op{Vect}_0(SL_2(\C))$, the Lie algebra of infinitesimal symmetries of the complex manifold $SL_2(\C)$ preserving the holomorphic volume form. Those from the fermionic fields contribute $\Pi \C^2 \otimes \Oo(SL_2(\C))$.  The Lie bracket, and the differential relating fermionic global symmetries of Kodaira-Spencer theory to the bosonic global symmetries of holomorphic Chern-Simons theory, are exactly as in section \ref{sec:global_symmetry_isomorphism}.    

\section{Characterizing the chiral algebra by its global symmetries}
The large $N$ chiral algebra is generated by single-trace operators, and its holographic dual by boundary modifications.  In each case, the OPE of two generators only contains the generators, together with the identity operator.  The OPE of two single-trace operators never yields a multi-trace operator, much like in a Kac-Moody algebra.  

We will show that knowledge of the global symmetry algebra is enough to determine directly \emph{all} OPEs, except for the coefficient of the identity operator.  The coefficient of the identity operator encodes the two-point functions. A second argument will show that all two-point functions are determined by the two-point functions of $T$ and of the operators of dimension $\le 1$.

The idea is very simple, and can be illustrated with the example of the stress-energy tensor $T$.  The modes $\oint z^i T(z) \d z$ for $i \le 2$ are global symmetries of any chiral CFT,  and they commute according to the global conformal algebra $\mf{sl}_2(\C)$.  This implies non-trivial constraints on the OPE of $T$ with itself: we must have $T(0) T(z) \simeq z^{-1} \partial T (0) + 2 z^{-2} T(0) + z^{-4} c$. The central charge $c$ is not constrained by the $\mf{sl}_2$ commutation relations, but all other OPE coefficients are.  

In our chiral algebra, we have an infinite number of global symmetries given by certain modes of the primaries.  We know how they commute with each other.  This immediately imposes an infinite number of constraints on the OPE coefficients of the primaries, which we will now derive carefully.

\subsection{OPE coefficients from global symmetry algebra structure constants}
We will first recall some basic general facts about OPEs, bearing in mind a special feature of our situation: the space of single-trace operators is \emph{not} a representation of the Virasoro algebra, but only of its positive part spanned by $L_{-1}, L_0,L_1,\dots$. 

Recall that an operator is a quasi-primary if it is annihilated by $L_1$.  As a representation of the $\mf{sl}_2(\C)$ of global conformal symmetries, the single trace operators are a sum of lowest-weight Verma modules whose lowest weight is positive.  The lowest-weight vectors are, of course, the quasi-primaries, and the weight is the conformal dimension.  Because these lowest-weight modules do not have any vectors of weight $0$, they can not have any non-trivial sub-modules. This implies that single-trace operator can be written uniquely as a sum of operators obtained by applying $L_{-1}^k$ to single-trace quasi-primaries (and the same holds in the holographic context).

In this setting, quasi-primaries are the same as primaries, however we will use the terminology quasi-primaries because we don't make use of the higher Virasoro modes $L_i$, $i > 1$. 

 Suppose we have two quasi-primaries $\mc{O}_i$ of dimension $r_i$.  Then the modes $\oint z^{k_i} \mc{O}_i(z) \d z$, where $k_i \le 2r_i-2$, are in the global symmetry algebra.  We will show that the commutator in the global symmetry algebra can be written in terms of OPE coefficients.

Suppose that 
\begin{equation} 
	\mc{O}_1(0) \mc{O}_2(z) \simeq \sum z^{-n}  \mc{O}_{(n)} (0) 
\end{equation}
is the OPE.  Let us further expand
\begin{equation} 
	\mc{O}_{(n)} = \sum \frac{1}{m!} L_{-1}^m \mc{O}_{(n,m)} \label{eqn_qp_expansion} 
\end{equation}
where each $\mc{O}_{(n,m)}$ are quasi-primaries.

The fact that the operators $\mc{O}_i$ are quasi-primaries implies we can recover the full OPE from only thse OPE coefficients $\mc{O}_{(n,0)}$ that are quasi-primaries.  This is a standard fact in CFT whose derivation we recall. We have the identity
\begin{align} 
	L_1 \mc{O}_{(n+1)}(0)&= \oint_{\abs{z} = 1} \oint_{\abs{w} = 1+\eps} \mc{O}_1(0) \mc{O}_2(z) T(w) w^2 z^n \d z \d w \\
	&= \oint_{\abs{z} = 1} \oint_{\abs{u} = \eps} \mc{O}_1(0) z^n (u^2 + 2 u z) \mc{O}_2(z) T(z+u)  \d u \d z \\
	&= 2 \oint_z \mc{O}_1(0) z^{n+1} L_0 \mc{O}_2(z) + \oint_z \mc{O}_1(0) z^{n+2} L_{-1} \mc{O}_2(z) \\
	&= (2 r_2  - (n+2) ) \mc{O}_{(n+2)}(0) 
\end{align}
where $r_2$ is the spin of $\mc{O}_2$.

Applying this to the the expansion \eqref{eqn_qp_expansion} of $\mc{O}_{(n)}$ we find 
\begin{equation} 	
	  (L_0 + \tfrac{1}{2}(m-1)  ) \mc{O}_{(n,m)} = (2 r_2 - (n+1) ) \mc{O}_{(n+1,m-1)}. \label{eqn_L-1_identity}	
\end{equation}
Thus, all the $\mc{O}_{(n,m)}$ are determined from the $\mc{O}_{(n+m,0)}$ in a simple way. Note that the coefficient on both sides is non-zero, because $n+1 \le 2 r_2 - 2$.

Our goal is to use the commutation relations in the global symmetry algebra to calculate as many of the quasi-primary operators $\mc{O}_{(n,0)}$ appearing in the OPE as possible. The commutator in the global symmetry algebra takes the form 
\begin{align} 
	\left[ \oint  \mc{O}_1(z_1) \d z_1 , \oint z_2^{k} \mc{O}_2(z_2) \d z_2 \right] &= \oint_{\abs{z_1} = 1}  \oint_{\abs{u} = \eps}  (z_1+u)^{k} \mc{O}_1 ( z_1 ) \mc{O}_2(z_1 + u) \d z_1 \d u\\
	&=   \oint_{\abs{z_1} = 1}  \oint_{\abs{u} = \eps} (z_1+u)^{k} \sum_{n} u^{-n}    \mc{O}_{(n)} ( z_1 ) \d z_1 \d u \\
	&= \sum \oint { k \choose r} z_1^{k-r} \mc{O}_{(r+1)}(z_1) \d z_1\\
	&= \sum (-1)^m \oint { k \choose r} z_1^{k-r-m} \frac{(k-r)!}{(k-r-m)!m! } \mc{O}_{(r+1,m)}(z_1) \d z_1.
\end{align}
In the last line we have used the expansion \eqref{eqn_qp_expansion}.   

By projecting onto the highest weight vectors for the global $SL_2(\C)$ conformal symmetry, we only keep those terms with $k-r-m=0$ in the last line: 
\begin{equation} 
 	\Pi_{HW}  \left[ \oint  \mc{O}_1(z_1) \d z_1 , \oint z_2^{k} \mc{O}_2(z_2) \d z_2 \right]  = \sum (-1)^{k-r} \oint { k \choose r} \mc{O}_{(r+1,k-r)}(z) \d z. 
 \end{equation}
where $\Pi_{HW}$ is the projection onto the highest weight.  Using the identity \eqref{eqn_L-1_identity}, we can write this as
\begin{equation} 
 \Pi_{HW}  \left[ \oint  \mc{O}_1(z_1) \d z_1 , \oint z_2^{k} \mc{O}_2(z_2) \d z_2 \right]  = F(r_1,r_2,k) \oint \mc{O}_{(k+1,0)}(z) \d z \label{eqn_ope_final} 
 \end{equation}
 where $F(r_1,r_2,k)$ is a non-zero combinatorial constant only depending on the spins $r_i$ of the operators $\mc{O}_i$ and on $k$.  	

For any quasi-primary $\mc{O}$, the operator $\mc{O}$ is determined by the mode $\oint \d z \mc{O}(z)$.  Equation \eqref{eqn_ope_final} thus gives us an explicit formula for the quasi-primary component of the coefficient of $z^{-k-1}$ in the OPE between $\mc{O}_1$ and $\mc{O}_2$, for any $k \le 2 r_2 - 2$.

As one can see from the explicit form of the since-trace operators in the large $N$ gauge theory, the OPE of quasi-primaries of spins $r_1,r_2$ with $r_1 \le r_2$ has as its most singular term $z^{-2 r_1 }$ times a quasi-primary of spin $r_2 - r_1$.   

The structure constants of the global symmetry algebra $\mf{a}_{\infty}$ determines all OPE coefficients up to order $z^{1 - 2 r_2}$.  If $r _1 < r_2$, this means that all OPE coefficients have been determined. If $r_1 = r_2$,  the only OPE coefficient which is not determined is the two-point function, which is given as the coefficient of $z^{-2 r_1}$. If $r_1 > r_2$ we can swap the roles of the two operators and again all OPE coefficients are determined.

Strictly speaking, this method does not determine the OPEs involving the two quasi-primaries $\op{Tr} Z_i$ of spin $1/2$, because these do not correspond to any global symmetries. However, these operators are not very interesting. 

\subsection{Matching holographic and large $N$ OPE coefficients} 
We know that the large $N$ global symmetry algebra $\mf{a}_{\infty}$ is isomorphic, as a Lie algebra, to the holographic global symmetry algebra $\mf{a}_{\infty}^{hol}$.  In this section we will explain how to explicitly constrain the structure of the large $N$ chiral algebra from this isomorphism. 

Suppose we choose a basis $M_i$ of highest weight vectors under the $SL_2(\C)$ conformal symmetry of the holographic global symmetry algebra $\mf{a}_{\infty}^{hol}$.  We let $M_{i,k}$, $k \ge 0$ be the descendents.  Suppose we have commutation relations of the form
\begin{equation} 
[M_i,M_{j,k}] = \sum c_{ij,k}^l M_{l} + \text{ descendents}. 
 \end{equation}
Since we have an explicit description of $\mf{a}_{\infty}^{hol}$ one can, in principle, determine the $c_{ij,k}^l$. 

We know that there is an isomorphism $\mf{a}_{\infty}^{hol} \iso \mf{a}_{\infty}$ between the holographic global symmetry algebra and the large $N$ global symmetry algebra.   The analysis above tells us that there exists \emph{some} basis $\mc{O}_i$ of quasi-primaries of the large $N$ chiral algebra, whose OPE coefficients are
\begin{equation} 
\mc{O}_i(0) \mc{O}_j(z) \sim \sum z^{-k-1} \til{c}_{ij,k}^l \mc{O}_l(0) + \text{ descendents} + z^{-2 r_2} \ip{\mc{O}_i, \mc{O}_j} \delta_{r_1,r_2},  
 \end{equation}
and where $\til{c}_{ij,k}^l$ and $c_{ij,k}^l$ are related by a universal factor $F(r_1,r_2,k)$ depending only on $k$ and the spins $r_i$ of the operators involved. 

We can also choose a basis of quasi-primaries of the holographic chiral algebra with the same quantum numbers, and with identical OPE coefficients, except for possibly the two-point functions.  Therefore the algebras are isomorphic except for possibly the two-point functions. 

\subsection{Matching two-point functions}
So far, we have seen that the large $N$ and holographic chiral algebras have the same OPE coefficients, but not necessarily the same two-point functions.  Here we will verify that in each case all two-point functions are determined by the global symmetry algebra, the two-point functions of $T$, and those of operators of spin $\le 1$. 

The global symmetry algebra acts on the space of local operators. If $\mc{O}_i$ are quasi-primaries, the action of a global symmetry $\oint z^k \mc{O}_2 \d z$ on $\mc{O}_1$ is given by the formula
\begin{equation} 
	\oint \mc{O}_1(0) \mc{O}_2(z) z^k \d z 
\end{equation}
for $k \le 2 r_2 - 2$, $r_i$ the spins of the operators.

The OPE between $T$ and any quasi-primary $\mc{O}$ immediately tells us that
\begin{equation} 
	\mc{O} (0) =(2r-2) \oint T(0) \mc{O}(z) z \d z 
\end{equation}
if $\mc{O}$ has spin $r>1$. Thus, all quasi-primaries of spin $>1$ are obtained by applying a global symmetry to $T$.

Suppose $\mc{O}_1,\mc{O}_2$ are two operators of spin $r > 1$.  Let $M_i = \frac{1}{2r-2}\oint z \mc{O}_i (z) \d z$ be corresponding global symmetries, so that $\mc{O}_i = M_i \cdot T$. Because global symmetries preserve the vacuum at $0$ and $\infty$ we have
\begin{align} 
	\ip{\mc{O}_1(0) , \mc{O}_2(\infty)} &= \ip{M_1 \cdot T(0), \mc{O}_2(\infty)} \\
	&= \ip{T,M_1 \cdot \mc{O}_2(\infty)}.
\end{align}
The action of $M_1$ on $\mc{O}_2(\infty)$ is 
\begin{equation} 
	\oint \mc{O}_1(z) z \d z \mc{O}_2(\infty). 
\end{equation}
Because $\mc{O}_2$ is placed at $\infty$, this is determined by the coefficient of $z^{2-2r}$ in the OPE between $\mc{O}_1(z)$ and $\mc{O}_2$.  This is in the range of OPE coefficients determined by the global symmetry algebra.   

Suppose that the OPE between $\mc{O}_1$ and  $\mc{O}_2$ contains $z^{2 - 2r} T C$ for some constant $C$, which can be determined from the structure constants of the global symmetry algebra. Then, this argument tells us that two point function is 
\begin{equation} 
	\ip{\mc{O}_1(0), \mc{O}_2(\infty)} =\frac{1}{2r-2} C \ip{T(0), T(\infty)}  
\end{equation}
Thus, all two-point functions between operators of spin $>1$ have been reduced to the central charge and the structure constants of the global symmetry algebra.

\section{Comparing with type IIB supergravity}
In \cite{Costello:2016mgj}, the first-named author and Si Li gave a conjectural description of a twisted version of type IIB supergravity on $AdS_5 \times S^5$.  In this section we will see that this proposal is fully compatible with the results of this paper.

Recall that in the set-up of \cite{Beem:2013sza}, one localizes a superconformal theory with respect to a super-conformal super-charge of the form
\begin{equation} 
\mathbf{Q} = Q + S 
 \end{equation}
where $Q$ is an ordinary super-symmetry and $S$ is a super-confomral super-charge.  We can instead consider the family of super-charges 
\begin{equation} 
\mathbf{Q}_{\eps} = Q + \eps S 
 \end{equation}
for a parameter $\eps$. For non-zero $\eps$, these are related by conjugating with an element of the superconformal algebra, so that they all gives the same cohomology. We can thus compute the $Q + \eps S$ cohomology by first taking the $Q$-cohomology, and then the $S$-cohomology.

The $Q$-cohomology brings us to the holomorphic twist of $\mc{N}=4$ gauge theory. This has a $\mf{psl}(3 \mid 3)$ holomorphic superconformal symmetry algebra, which is the $Q$-cohomology of $\mf{psl}(4 \mid 4)$ superconformal symmetry of the physical theory.   The superconformal supercharge $S$ commutes with $Q$, and so can be viewed as an element of $\mf{psl}(3 \mid 3)$.  The localization studied in \cite{Beem:2013sza} can be obtained by localizing the holomorphic twist of $\mc{N}=4$ using the superconformal supercharge $S \in \mf{psl}(3 \mid 3)$. 

In \cite{Costello:2016mgj}, a proposal for the holographic dual of the holomorphic twist of $\mc{N}=4$ gauge theory was given. This is a $5$-dimensional topological $B$-model background with an action of $\mf{psl}(3 \mid 3)$. 

It is natural to expect that localizing this $5$-dimensional topological $B$-model with respect to the element $S \in \mf{psl}(3 \mid 3)$ should bring us to the $3$-dimensional $B$-model studied in this paper.  Here we will verify this, showing compatibility with the proposal of \cite{Costello:2016mgj}.
\subsection{The twist of $AdS_5 \times S^5$}
In \cite{Costello:2016mgj}, the holographic dual of the holomorphic twist of $\mc{N}=4$ gauge theory was proposed to be the $B$-model on $\C^5 \setminus \C^2$, in the presence of a certain closed-string field. The closed-string fields are given by the kernel of the operaotr $\partial$ inside the polyvector fields:
\begin{equation} 
\op{Ker} \partial \subset \PV^{\ast,\ast}(\C^5 \setminus \C^2) = \Omega^{0,\ast}(\C^5 \setminus \C^2, \wedge^\ast T (\C^5 \setminus \C^2) ). 
 \end{equation}

Let us choose coordinates $z_i$ on the $\C^2$ wrapped by the brane, and $w_j$ on the three normal directions.  We are removing the locus where $w_i = 0$.  The field which deforms us to the holographic dual of the holomorphic twist of the $N=4$ gauge theory is given by 
\begin{equation} 
F = N \eps_{ijk} \frac{ \wbar_i \d \wbar_j \d \wbar_k }{\norm{w}^6} \partial_{z_1} \partial_{z_2} \in \PV^{2,2}(\C^5 \setminus \C^2).   
 \end{equation}
Note that unlike in analogous setting in the $3$-dimensional topological string which is the main focus of this paper, this field is not a Beltrami differential and so does not deform the geometry. 

We want to analyze the affect of further localizing with respect to the $S$-supercharge in $\mf{psl}(3 \mid 3)$. The action of $\mf{psl}(3 \mid 3)$ on the $5$-dimensional topological string is given by an embedding of $\mf{psl}(3 \mid 3)$ into the closed-string fields of the theory.  The field corresponding to an element $S \in \mf{psl}(3 \mid 3)$ should be thought of as arising from the corresponding super-ghost in type IIB supergravity. To localize with respect to $S$, we should move to the topological string background given by this closed-string field. 
	 
The formula for the embedding of $\mf{psl}(3 \mid 3)$ into the closed-string fields is as follows, up to order $N$ corrections:
\begin{align} 
\partial_{z_i}  & \text{ translations }\\
z_j \partial_{z_i} & \text{ rotations }\\
z_i \partial_{z_i} - \tfrac{2}{3} ( w_j \partial_{w_j}) & \text{ scaling } \\
z_i \left( \sum_j z_j \partial_{z_j} - \sum_{k } w_k \partial_{w_k} \right) & \text { special conformal transformations  } \\
w_i \partial_{w_j} &  SU(3)_R \text{ generators}\\
w_i  & \text{ ordinary supercharges } \\
\partial_{z_i} \partial_{w_j}  & \text{ ordinary supercharges} \\
z_i w_j  & \text{ superconformal supercharges}  \\
\partial_{w_i} \left(\sum_l z_l \partial_{z_l} -  \sum_k w_k \partial_{w_k}\right)  & \text{ superconformal supercharges.} 
 \end{align}
These expressions need to be corrected at order $N$ and possibly higher. These corrections were not explicitly compute in \cite{Costello:2016mgj}: they were shown to exist by a cohomological argument. Shortly we will write down an explicit formula for the order $N$ corrections we care about.

One can check that, under the Schouten-Nijenhuis bracket on polyvector fields, these fields have the commutation relations of $\mf{psl}(3 \mid 3)$.  For example, the fact that a superconformal supercharge and an ordinary supercharge commute to a combination of  rotation, scaling, and an $SU(3)_R$ symmetry is given by the formula
\begin{equation} 
\{z_i w_j,  \partial_{z_k} \partial_{w_l} \} = \delta_{ki} w_j \partial_{w_l} - \delta_{jl} z_i \partial_{z_k}. 
 \end{equation}

The superconformal supercharge $S$ used in \cite{Beem:2013sza} can be taken to be any odd, rank $1$ element of $\mf{psl}(3 \mid 3)$. We will take it to be 
\begin{equation} 
S = z_2 w_1 + O(N). 
 \end{equation}
 The reason we need an order $N$ term is that the bracket of $z_2 w_1$ with the flux sourced by the $D3$ brane is non-zero:
\begin{equation} 
\{z_2 w_1, F\} = N w_1 \frac{\eps_{ijk} \wbar_i \d \wbar_j \d \wbar_k} { \norm{w}^6} \partial_{z_1}.  
 \end{equation}
To correct for this, we need to add on to $z_2 w_1$ at order $N$ a Beltrami differential $\beta$ so that
\begin{equation} 
\dbar \beta =    N w_1 \frac{\eps_{ijk} \wbar_i \d \wbar_j \d \wbar_k} { \norm{w}^6} \partial_{z_1}.  
 \end{equation}
One can check that
\begin{equation} 
\beta =  \frac{1}{4} N \frac{ \wbar_2 \d \wbar_3 - \wbar_3 \d \wbar_2}    { \norm{w}^4} \partial_{z_1}   
 \end{equation}
 works, and is indeed the unique solution to the equation in an appropriate gauge. 

What we have found is that the complete expression for the closed-string field given by $N$ $D3$ branes together with the superghost for the superconformal symmetry $S$ is given by the expression
\begin{equation} 
\Phi = z_2 w_1 + \frac{1}{4} N \frac{ \wbar_2 \d \wbar_3 - \wbar_3 \d \wbar_2}    { \norm{w}^4} \partial_{z_1}  +  N \eps_{ijk} \frac{ \wbar_i \d \wbar_j \d \wbar_k }{\norm{w}^6} \partial_{z_1} \partial_{z_2}. 
 \end{equation}
 This satisfies the Maurer-Cartan equation
 \begin{equation} 
 \dbar \Phi + \tfrac{1}{2} \{\Phi,\Phi\} = 0 
  \end{equation}
  and so defines a consistent background for the topological $B$-model on $\C^5 \setminus \C^2$.   
  
We conjecture, following \cite{Costello:2016mgj}, that $5$-dimensional Kodaira-Spencer theory in this background is equivalent to type IIB supergravity on $AdS^5 \times S^5$, in the background given by the superghosts for the $\mbf{Q} = Q+S$ supercharge used in \cite{Beem:2013sza}.  

One of the terms in the closed-string field $\Phi$ is a Beltrami differential $\beta$, which does not satisfy $\dbar \beta = 0$ except on the locus $z_2 = w_1 = 0$.  Therefore $\beta$ defines a \emph{non-integrable} deformation of the complex structure of $\C^5 \setminus \C^2$, which is integrable when restricted to $\C^3 \setminus \C$.   (Even though $\beta$ is a non-integrable deformation of complex structure, the closed string background is still consistent because of the presence of the other two terms). 

If we restrict to the locus where $z_2 = w_1 = 0$, the closed-string field $\Phi$ becomes the Beltrami differential on $\C^3 \setminus \C$ which turns it into the deformed conifold.  Thus, our $5$-dimensional topological string background contains a copy of the deformed conifold $SL_2(\C)$. This is consistent with the idea that the $3$-dimensional topological string we study in the bulk of the paper lives on $AdS_3 \times S^3$ inside $AdS_5 \times S^5$.

The term $z_2 w_1$ in the closed-string field $\Phi$ is a superpotential term. In the topological $B$-model, a quadratic superpotential localizes the model to the $B$-model on the critical locus, which is the locus $z_2 = w_1 = 0$.  On this locus, the Beltrami differential is integrable and deforms the complex structure to $SL_2(\C)$. 

We have sketched that applying the $\mbf{Q} = Q + S$ localization to type IIB supergravity, using the proposal of \cite{Costello:2016mgj}, localizes us to the topological $B$-model on $SL_2(\C)$, as desired. 

\appendix

\section{Large $N$ matrix and vector indices} \label{app:index}
We can quickly review the large $N$ calculation of indices. 
We need to compute contour integrals for the projection on gauge-invariant operators:
\begin{equation}
I_N = \int \prod_a d\theta_a \prod_{a \neq b} (1-e^{i \theta_a-i \theta_b}) \prod_{a,b} \frac{\prod_u \left(1-f_u e^{i \theta_a-i \theta_b}\right)}{\prod_v \left(1-b_v e^{i \theta_a-i \theta_b}\right)}
\end{equation}

We can evaluate it in a saddle point approximation. The eigenvalues feel a potential 
\begin{equation}
-\log \left| \left(1-e^{i \theta_a-i \theta_b}\right)\frac{\prod_u \left(1-f_u e^{i \theta_a-i \theta_b}\right)}{\prod_v \left(1-b_v e^{i \theta_a-i \theta_b}\right)} \right|^2
\end{equation}
and the saddle point equations are 
\begin{equation}
\frac{e^{i \theta_a} +e^{i \theta_b}}{e^{i \theta_a} - e^{i \theta_b}}+ \sum_u f_u \frac{e^{2 i \theta_a} -e^{2 i \theta_b}}{(e^{i \theta_a} - f_u e^{i \theta_b})(f_u e^{i \theta_a} - e^{i \theta_b})}- \sum_v b_v \frac{e^{2 i \theta_a} -e^{2 i \theta_b}}{(e^{i \theta_a} - b_v e^{i \theta_b})(b_v e^{i \theta_a} - e^{i \theta_b})} =0
\end{equation}

We can also express the overall potential in terms of the eigenvalue density $\rho(\theta)$
\begin{equation}
S = -N^2 \int d \theta d \theta' \rho(\theta) \rho(\theta') \log\left(1-e^{i \theta-i \theta'}\right)\frac{\prod_u \left(1-f_u e^{i \theta-i \theta'}\right)}{\prod_v \left(1-b_v e^{i \theta-i \theta'}\right)} 
\end{equation}
Expanding in Fourier modes 
\begin{equation}
\rho_k = \frac{1}{N} \sum_a e^{i k \theta_a}
\end{equation} 
gives 
\begin{equation}
S = \frac{N^2}{2 \pi} \sum_k |\rho_k|^2 \frac{1}{k}\left(1 + \sum_u f_u^k - \sum_v b_v^k \right)
\end{equation}

As long as $1 + \sum_u f_u^n - \sum_v b_v^n$ is always positive, the obvious saddle is at $\rho_k=0$. 
In the neighbourhood of that saddle, the change of coordinates between $\theta_a$ fluctuations and $\rho_k$ 
fluctuations is a discrete Fourier transform and contributes a constant Jacobian. The semiclassical evaluation of the integral gives 
\begin{equation}
I_\infty = \frac{1}{\prod_{k=1}^\infty \left(1 + \sum_u f_u^k - \sum_v b_v^k \right)}
\end{equation}

Next, we can add vectorial contributions:
\begin{equation}
I'_N = \int \prod_a d\theta_a \prod_{a \neq b} (1-e^{i \theta_a-i \theta_b}) \prod_{a,b} \frac{\prod_u \left(1-f_u e^{i \theta_a-i \theta_b}\right)}{\prod_v \left(1-b_v e^{i \theta_a-i \theta_b}\right)} \prod_a \frac{\prod_s(1 - f'_s e^{i \theta_a})(1 - \bar f'_s e^{-i \theta_a})}{\prod_t(1 - b'_t e^{i \theta_a})(1 - \bar b'_t e^{-i \theta_a})}
\end{equation}
They will not change the saddle if the number of flavors does not scale with $N$. We can write their contribution to the semi-classical action as 
\begin{equation}
S' = -N \int d\theta \rho(\theta) \log \frac{\prod_s(1 - f'_s e^{i \theta})\prod_{\bar s}(1 - \bar f'_{\bar s} e^{-i \theta})}{\prod_t(1 - b'_t e^{i \theta})\prod_{\bar t}(1 - \bar b'_{\bar t} e^{-i \theta})}
\end{equation}
and then 
\begin{equation}
S' = N \sum_{k>0} \rho_k \frac{1}{k}\left(\sum_s (f'_s)^k-\sum_t  (b'_t)^k\right)+N \sum_{k>0} \rho_{-k} \frac{1}{k}\left(\sum_{\bar s} (\bar f'_{\bar s})^k-\sum_{\bar t}  (\bar b'_{\bar t})^k\right)
\end{equation}

Completing the squares, we get a contribution 
\begin{equation}
\exp \left[ \sum_{k>0} \frac{1}{k} \frac{\left(\sum_s (f'_s)^k-\sum_t  (b'_t)^k\right) \left(\sum_{\bar s}  (\bar f'_{\bar s} )^k-\sum_{\bar t}  (\bar b'_{\bar t})^k\right)}{1 + \sum_u f_u^k - \sum_v b_v^k} \right]
\end{equation}
If we expand 
\begin{equation}
\frac{\left(\sum_s (f'_s)^k-\sum_t  (b'_t)^k\right) \left(\sum_{\bar s} (\bar f'_{\bar s})^k-\sum_{\bar t}  (\bar b'_{\bar t})^k\right)}{1 + \sum_u f_u^k - \sum_v b_v^k}  = \sum_i x_i^k-\sum_j y_j^k
\end{equation}
then the index takes the form 
\begin{equation}
I'_\infty = \frac{1}{\prod_{k=1}^\infty \left(1 + \sum_u f_u^k - \sum_v b_v^k \right)} \frac{\prod_i (1-x_i) }{\prod_j (1-y_j)}
\end{equation}

\subsection{Example: one boson matrix}
Without vectors, we have 
\begin{equation}
I_\infty = \frac{1}{\prod_{k=1}^\infty \left(1 - b^k \right)}
\end{equation}
which counts words built from traces $\mathrm{Tr} B^k$.

Adding vectors, we have 
\begin{equation}
\frac{\left(\sum_s (f'_s)^k-\sum_t  (b'_t)^k\right) \left(\sum_{\bar s} (\bar f'_{\bar s})^k-\sum_{\bar t}  (\bar b'_{\bar t})^k\right)}{1 -b^k}  = \sum_i b^{i k}\left(\sum_s (f'_s)^k-\sum_t  (b'_t)^k\right) \left(\sum_{\bar s} (\bar f'_{\bar s})^k-\sum_{\bar t}  (\bar b'_{\bar t})^k\right)
\end{equation}
so that the index is 
\begin{equation}
I'_\infty = \frac{1}{\prod_{k=1}^\infty \left(1 - b^k \right)} \prod_{k=1}^\infty \frac{\prod_{s, \bar t}(1- f'_s \bar b'_{\bar t} b^k)\prod_{\bar s,t}(1- \bar f'_{\bar s} b'_{t} b^k)}{\prod_{s, \bar s}(1- f'_s \bar f'_{\bar s} b^k)\prod_{t, \bar t}(1- b'_t \bar b'_{\bar t} b^k)} 
\end{equation}
which counts words built from traces $\mathrm{Tr} B^k$ and open expressions $\bar B' B^k B'$, $\bar F' B^k F'$, $\bar B' B^k F'$, $\bar F' B^k B'$.

\section{Schur index}
Now we have a pair of adjoint symplectic bosons $Z_i$ and their derivatives, and ghosts $b$, $\partial c$ and their derivatives. 
Then 
\begin{equation}
1 + \sum_u f_u^k - \sum_v b_v^k  = 1 + \frac{2q^k }{1-q^k} - \frac{(y^k + y^{-k})q^{\frac{k}{2}}}{1-q^k} = \frac{(1-y^k q^{\frac{k}{2}})(1-y^{-k} q^{\frac{k}{2}})}{1-q^k}
\end{equation}

The index of single-trace operators is then
\begin{equation}
\sum_k \left(y^k q^{\frac{k}{2}} +y^{-k} q^{\frac{k}{2}} -q^k \right) = \frac{y q^{\frac{1}{2}}}{1-y q^{\frac{1}{2}} }+ \frac{y^{-1} q^{\frac{1}{2}}}{1-y^{-1} q^{\frac{1}{2}} }-\frac{q}{1-q}
\end{equation}
Removing total derivatives, we finally get 
\begin{equation}
\sum_k \left(y^k q^{\frac{k}{2}} +y^{-k} q^{\frac{k}{2}} -q^k \right) = \frac{y q^{\frac{1}{2}}}{1-y q^{\frac{1}{2}} }+ \frac{y^{-1} q^{\frac{1}{2}}}{1-y^{-1} q^{\frac{1}{2}} }-\frac{q}{1-q}
\end{equation}
Expanding out, we recognize the symmetrized traces $\mathrm{Tr} Z_{(i_1} \cdots Z_{i_n)}$ of $su(2)$ spin $n/2$ and
their small $N=4$ super-Virasoro descendants, a fermionic doublet with spin $(n-1)/2$ and a bosonic operator of spin $n/2-1$. 

Next, we can add $r$ fundamental and anti-fundamental symplectic bosons, and $r$ fundamental and anti-fundamental fermions. 
Then 
\begin{align}
\sum_s (f'_s)^k-\sum_t  (b'_t)^k &= \left( \sum_s z^k_s - w^k_s \right)\frac{q^{\frac{k}{2}}}{1-q^k}\cr
  \sum_{\bar s}  (\bar f'_{\bar s} )^k-\sum_{\bar t}  (\bar b'_{\bar t})^k &= \left( \sum_s z^{-k}_s - w^{-k}_s \right)\frac{q^{\frac{k}{2}}}{1-q^k}
\end{align}

Thus the extra contribution to the single trace index is
\begin{equation}
\left( \sum_s z^k_s - w^k_s \right)\left( \sum_s z^{-k}_s - w^{-k}_s \right) \sum_n \chi(y)_{\frac{n}{2}} \frac{q^{1+\frac{n}{2}}}{1-q} \end{equation}

These operators transform in the adjoint of $u(k|k)$. We can split them into an adjoint of 
$psu(k|k)$ and a contribution of the form 
\begin{equation}
2 \sum_n \chi(y)_{\frac{n}{2}} \frac{q^{1+\frac{n}{2}}}{1-q} \end{equation}
which cancels out all the fermionic superpartners of the $\mathrm{Tr} Z_{(i_1} \cdots Z_{i_n)}$ single traces. 

This makes sense: the theory with fundamentals has $N=2$ SUSY rather than $N=4$ and the $N=4$ super-Virasoro symmetry 
is broken to Virasoro and $su(2)$. At the level of cohomology, we thus expect the cancellation to actually happen.

\section{Some large $N$ correlation functions}
It is useful to employ some auxiliary variables $x_1,x_2$ which transform in the fundamental of $SU(2)_R$. We then write 
\begin{equation}
Z(x;z) \equiv Z^i(z) x_i
\end{equation}
to keep track of $SU(2)_R$ indices. For example, the OPE can be written as
\begin{equation}
Z(x;z) Z(x';z') \sim \frac{(x,x')}{z-z'}
\end{equation}

In this notation, we can define
\begin{equation}
A^{(n)}(x;z) = \Tr Z(x;z)^n 
\end{equation}
to package together all components of $A^{(n)}$. The right hand side of this expression is a polynomial of degree $n$ in the variables $x_i$.  The whole expression is $SU(2)_R$ invariant, so that the coefficients of $x_1^a x_2^{n-a}$ transform in the representation of spin $n/2$ of $SU(2)_R$.

In the large $N$ limit, a two-point function $\langle A^{(n)}(x;z) A^{(n)}(x';z') \rangle$ 
will receive $n$ planar contributions: once we choose which field in $A^{(n)}(x';z')$ is contracted with the first field in the trace in 
$A^{(n)}(x;z)$ all other contractions are determined. We get 
\begin{equation}
\langle A^{(n)}(x;z) A^{(n)}(x';z') \rangle = n N^n \frac{(x,x')^n}{(z-z')^n}
\end{equation}
For three-point functions $\langle A^{(n)}(x;z) A^{(n')}(x';z') A^{(n'')}(x'';z'') \rangle$, planar contractions involve a sequence of $\frac{n+n'-n''}{2}$ 
contractions of consecutive fields in the first two operators, etc. It is easy to count how many planar contractions are available: along the first operator there is a specific location where we begin contractions with fields of the second operator, etcetera. That gives $n n' n''$ choices: 
\begin{align}
\langle A^{(n)}(x;z) & A^{(n')}(x';z') A^{(n'')}(x'';z'') \rangle = n n' n'' N^{\frac{n+n'+n''}{2}-1} \cdot \cr &\cdot \frac{(x,x')^{\frac{n+n'-n''}{2}}(x',x'')^{\frac{n'+n''-n}{2}}(x'',x)^{\frac{n''+n-n'}{2}}}{(z-z')^{\frac{n+n'-n''}{2}}(z'-z'')^{\frac{n'+n''-n}{2}}(z''-z)^{\frac{n''+n-n'}{2}}}
\end{align}

\subsection{Some large $N$ OPE's}
We can also look at OPE's in the planar limit. Of course, the OPE's of a chiral algebra 
do not contain any further information than what can be gleaned from two- and three- point functions:
Virasoro symmetry determines the coefficients of all primaries and descendants in the OPE in terms of that 
basic data. On the other hand, as the Virasoro symmetry does {\it not} play very well with the planar approximation, 
a direct calculation of the large $N$ OPE can be useful.  

The computation of the OPE itself is completely straightforward, as the operators are polynomials in free fields. 
The hard part of the calculation is to manipulate the final expression until the OPE contains only standard BRST representatives for the 
single-trace operators. The $SU(2)_R$ global symmetry is extremely helpful, as a term of $SU(2)_R$ spin $\frac{n}{2}$ and dimension $\Delta$ dimension must be BRST equivalent to a combination of $\partial^{\Delta-\frac{n}{2}} A^{(n)}$,$\partial^{\Delta-\frac{n}{2}-1}B^{(n)}$,$\partial^{\Delta-\frac{n}{2}-1}C^{(n)}$ or $\partial^{\Delta-\frac{n}{2}-2}D^{(n)}$ of the correct Grassmann parity and ghost number. If we cannot identify the
canonical BRST representatives at glance, we can use two-point functions to help with the identification.
The procedure is somewhat tedious and we have not been able to write the final large $N$ OPE's in closed form for all generators. 
We leave that calculation for future work. 

The simplest OPE is 
\begin{equation}
\Tr Z(x;z) \Tr Z(y;w)^n \sim n \frac{(x,y)}{z-w} \Tr Z(y;w)^{n-1}
\end{equation}
i.e. 
\begin{equation}
A^{(1)}(x;z) A^{(n)}(y;w) \sim n \frac{(x,y)}{z-w} A^{(n-1)}(y;w)
\end{equation}
The next simplest is the action of $SU(2)_R$ currents: 
\begin{equation}
\Tr Z(x;z)^2 \Tr Z(y;w)^n \sim n N \frac{(x,y)^2}{(z-w)^2} \Tr Z(y;w)^{n-2} + 2 n \frac{(x,y)}{z-w} \Tr Z(x;w) Z(y;w)^{n-1} 
\end{equation}
which can be written as 
\begin{equation}
A^{(2)}(x;z) A^{(n)}(y;w) \sim n N \frac{(x,y)^2}{(z-w)^2} A^{(n-2)}(y;w) + 2 \frac{(x,y)(x \cdot \partial_y)}{z-w} A^{(n)}(y;w)
\end{equation}
In particular, the $A^{(n)}$ are not current primaries, but have the obvious action of the current zeromodes $J^{(2)}_0(x) = \oint \d z A^{(2)}(x;z)$ : 
\begin{equation}
[J^{(2)}_0(x),A^{(n)}(y;w)] =2 (x,y)(x \cdot \partial_y) A^{(n)}(y;w)
\end{equation}

Our final example is the OPE between $A^{(3)}(z)$ and $A^{(n)}(w)$. We can start from the $A^{(3)}(z)$ and $A^{(3)}(w)$ OPE. 
The most singular term in the OPE is constant. The next term goes as $(z-w)^{-2}$ 
and involves a field of dimension $1$, which can only be $A^{(2)}(w)$. The last term goes as $(z-w)^{-1}$ and involves a field of dimension $2$. 
The only $SU(2)_R$ irreps which can appear will be $A^{(4)}(w)$ at spin $2$, $\partial A^{(2)}(w)$ at spin $1$ and $D^{(0)}(w) \equiv T(w)$ 
at spin $0$. No other calculations are needed besides a decomposition into $SU(2)_R$ irreps. 

Concretely, 
\begin{align}
\Tr &Z(x;z)^3  \Tr Z(y;w)^3 \sim 3 N^3 \frac{(x,y)^3}{(z-w)^3} + 9 N \frac{(x,y)^2}{(z-w)^2} \Tr Z(x;w) Z(y;w) + \cr 
&+ 9 N\frac{(x,y)^2}{(z-w)} \Tr \partial_w Z(x;w) Z(y;w) + 9 \frac{(x,y)}{(z-w)} \Tr Z(x;w)^2 Z(y;w)^2
\end{align}
The decomposition into $SU(2)_R$ irreps is straightforward. 
With some work, we obtain 
\begin{align}
A^{(3)}&(x;z)  A^{(3)}(y;w) \sim 3 N^3 \frac{(x,y)^3}{(z-w)^3} + \frac{9 N}{2} \frac{(x,y)^2}{(z-w)^2} (x \cdot \partial_y) A^{(2)}(y;w) + \cr 
&+ \frac{3}{4} \frac{(x,y)}{(z-w)} (x \cdot \partial_y)^2 A^{(4)}(y;w) + \frac{9 N}{4}\frac{(x,y)^2}{(z-w)} (x \cdot \partial_y) \partial_w A^{(2)}(y;w)  + \cr
&+ \frac{3}{2} \frac{(x,y)^3}{(z-w)} \left[ 3 N\epsilon_{ij} \Tr \partial_w Z^i(w) Z^j(w) -\epsilon_{ij}\epsilon_{i'j'} \Tr Z^i(w) Z^j(w)Z^{i'}(w) Z^{j'}(w) \right]
\end{align}
where the last term must be some alternative cohomology representative for the stress tensor $T(w)$, with an overall coefficient of $9 N$. 

For $n > 3$, we have
\begin{align}
\Tr Z(x;z)^3 & \Tr Z(y;w)^n \sim 3 n N^2 \frac{(x,y)^3}{(z-w)^3} \Tr Z(y;w)^{n-3} + 3 n N \frac{(x,y)^2}{(z-w)^2} \Tr Z(x;w) Z(y;w)^{n-2} + \cr 
&+ 3 n N\frac{(x,y)^2}{(z-w)} \Tr \partial_w Z(x;w) Z(y;w)^{n-2} + 3 n \frac{(x,y)}{(z-w)} \Tr Z(x;w)^2 Z(y;w)^{n-1}
\end{align}
Standard manipulations lead to 
\begin{align}
A^{(3)}&(x;z)  A^{(n)}(y;w) \sim 3 n N^2 \frac{(x,y)^3}{(z-w)^3} A^{(n-3)}(y;w) + \frac{3 n}{n-1} N \frac{(x,y)^2(x \cdot \partial_y)}{(z-w)^2}A^{(n-1)}(y;w) + \cr 
&+ \frac{3 n}{(n-1)^2} N \frac{(x,y)^2(x \cdot \partial_y)}{(z-w)} \partial_wA^{(n-1)}(y;w) + \frac{3}{n+1} \frac{(x,y)}{(z-w)} (x \cdot \partial_y)^2 A^{(n+1)}(y;w) \cr
&+\frac{6 n N}{n-1} \frac{(x,y)^3}{(z-w)} D^{(n-3)}(y;w)
\end{align}
where the representative we encounter for $D^{(n-3)}(y;w)$ involves terms with one $\epsilon_{ij}$ tensor and one derivative and terms with two $\epsilon_{ij}$
tensors and no derivatives. 

We can use the OPE to get the action of some global symmetry generators, in particular $J^{(3)}_{\frac12} (x) = \oint z \d z \Tr Z(x;z)^3$
\begin{align}
[ J^{(3)}_{\frac12} (x), A^{(n)}(y;w)]&=  \frac{3 n}{n-1} N (x,y)^2(x \cdot \partial_y) A^{(n-1)}(y;w) +  \cr &+ w N \frac{3 n}{(n-1)^2} (x,y)^2(x \cdot \partial_y) \partial A^{(n-1)}(y;w)  \cr & + w \frac{3}{n+1} (x,y)(x \cdot \partial_y)^2 A^{(n+1)}(y;w) \cr
&+\frac{6 n N}{n-1}  (x,y)^3 w D^{(n-3)}(y;w)
\end{align}
and $J^{(3)}_{-\frac12} (x) = \oint \d z \Tr Z(x;z)^3$
\begin{align}
[ J^{(3)}_{-\frac12} (x), A^{(n)}(y;w)]&= \frac{3 n}{(n-1)^2} N (x,y)^2(x \cdot \partial_y) \partial A^{(n-1)}(y;w)  \cr & + \frac{3}{n+1} (x,y)(x \cdot \partial_y)^2 A^{(n+1)}(y;w) \cr
&+\frac{6 n N}{n-1} (x,y)^3 D^{(n-3)}(y;0)
\end{align}

\subsection{OPE computations for open-string operators}

Computations of large $N$ OPE's between open string operators, ignoring closed string operators, is rather straightforward. 
Define again 
\begin{equation}
E^{(n)}_{ab}(x;z) =I_a(z)  Z(x;z)^n J_b(z)
\end{equation}

The simplest OPE (for $n>0$)
\begin{align}
I_a(z)J_b(z)  & I_c(w)  Z(y;w)^n J_d(w) \sim \frac{(-1)^{f_a f_b + f_a f_c+ f_b f_c}\delta_{a d}}{z-w} I_c(w)  Z(y;w)^n J_b(w) \cr 
&- \frac{(-1)^{f_b}\delta_{b c}}{z-w} I_a(w)  Z(y;w)^n J_d(w)
\end{align}
encodes the $\mathfrak{pgl}(k|k)$ current algebra action. Here the $f_a$ are the Grassmann parity of the various components of $I$ and $J$ 
and the signs are the same as one would encounter in the $\mathfrak{pgl}(k|k)$ structure constants. 
We can thus write 
\begin{equation}
E_\mathfrak{t}^{(0)}(z)  E_\mathfrak{t'}^{(n)}(y;w) \sim \frac{1}{z-w} E_{[\mathfrak{t},\mathfrak{t'}]}^{(n)}(y;w)
\end{equation}
and the zeromode action
\begin{align}
[e^{(0)}_{\mathfrak{t},0},  E_\mathfrak{t'} ^{(n)}(y;w)] =   E_{[\mathfrak{t},\mathfrak{t'}]}^{(n)}(y;w)
\end{align}
and thus 
\begin{align}
[e^{(0)}_{\mathfrak{t},0},  e^{(n)}_{ \mathfrak{t'},k}(y)] = e^{(n)}_{[\mathfrak{t},\mathfrak{t'}],k}(y)
\end{align}

The next non-trivial OPE is
\begin{align}
E^{(1)}_\mathfrak{t} (x;z)  & E_\mathfrak{t'}^{(n)}(y;w) \sim N\frac{(x,y)}{(z-w)^2} E^{(n-1)}_{[\mathfrak{t},\mathfrak{t'}]}(y;w) \cr 
&+ N\frac{1}{n+1} \frac{(x,y)}{(z-w)}\partial \left[E_{[\mathfrak{t},\mathfrak{t'}]}^{(n-1)}(y;w) \right] \cr 
&+\frac{1}{n+1}\frac{1}{z-w} (x \cdot \partial_y) E_{[\mathfrak{t},\mathfrak{t'}]}^{(n+1)}(y;w)
\end{align}
It gives the action of the next set of global symmetry generators 
\begin{multline}
[e^{(1)}_{\mathfrak{t} ,-\frac12}(x),  E_\mathfrak{t'}^{(n)}(y;w) ] = N\frac{1}{n+1} (x,y)\partial \left[ E_{[\mathfrak{t},\mathfrak{t'}]}^{(n-1)}(y;w) \right] \\ 
+ \frac{1}{n+1} (x \cdot \partial_y)  E_{ [\mathfrak{t},\mathfrak{t'}]}^{(n+1)}(y;w)
\end{multline}
and 
\begin{align}
[e^{(1)}_{\mathfrak{t},\frac12}(x),  &E^{(n)}_\mathfrak{t'}(y;w)]= N(x,y) E^{(n-1)}_{[\mathfrak{t},\mathfrak{t'}]}(y;w) \cr 
&+N w \frac{1}{n+1} (x,y) \partial \left[  E^{(n-1)}_{ [\mathfrak{t},\mathfrak{t'}]}(y;w) \right] \cr 
&+\frac{1}{n+1} w (x \cdot \partial_y)E_{ [\mathfrak{t},\mathfrak{t'}]}^{(n+1)}(y;w)
\end{align}
and then
\begin{multline}
[e^{(1)}_{\mathfrak{t},-\frac12}(x),  e^{(n)}_{\mathfrak{t'},k}(y)] = - N\frac{\frac{n}{2}+k}{n+1} (x,y) e^{(n-1)}_{[\mathfrak{t},\mathfrak{t'}],k-\frac12}(y) 
+\frac{1}{n+1} (x \cdot \partial_y) e^{(n+1)}_{[\mathfrak{t},\mathfrak{t'}],k-\frac12}(y)
\label{eqn_lvl1_open_currents}
\end{multline}
and 
\begin{multline}
[ e^{(1)}_{\mathfrak{t} ,\frac12}(x),   e^{(n)}_{\mathfrak{t'} ,k}(y)]= N\frac{\frac{n}{2}-k}{n+1}(x,y)  e^{(n-1)}_{[\mathfrak{t},\mathfrak{t'}],k+\frac12}(y)  
+\frac{1}{n+1}  (x \cdot \partial_y)  e^{(n+1)}_{[\mathfrak{t},\mathfrak{t'}],k+\frac12}(y)
\label{eqn_lvl1_open_currents2}
\end{multline}

In important consequence of equations \eqref{eqn_lvl1_open_currents} and  \eqref{eqn_lvl1_open_currents2} is that
\begin{equation}  
[ e^{(1)}_{\mathfrak{t} ,\frac12}(x),   e^{(1)}_{\mathfrak{t'} ,-\frac12}(y)] -  [ e^{(1)}_{\mathfrak{t} ,-\frac12}(x),   e^{(1)}_{\mathfrak{t'} ,\frac12}(y)] = - \frac{N}{2} (x,y)e^{(0)}_{[t,t']}.  \label{eqn_conifold_commutator} 
 \end{equation}
This commutation relation is clearly similar to the relation defining the algebra of functions on the deformed conifold, and we will leverage this fact later.

It is also straightforward to find the action of the generators of the closed string global symmetry algebra onto the open string generators. For example, we have
\begin{align}
A^{(3)}(x;z) & E^{(n)}(y;w) \sim 3 (n-2) N^2 \frac{(x,y)^3}{(z-w)^3} E^{(n-3)}(y;w)+ \cr &+ 3 N \frac{(x,y)^2}{(z-w)^2} (x \cdot \partial_y) E^{(n-1)}(y;w)+ \cr 
&+ \frac{3 N}{n+1}\frac{(x,y)^2}{(z-w)} (x \cdot \partial_y) \partial_w E^{(n-1)}(y;w) + \cr &+ \frac{3}{n+1} \frac{(x,y)}{(z-w)} (x \cdot \partial_y)^2 E^{(n+1)}(y;w)
\end{align}
which gives 
\begin{align}
[ J^{(3)}_{\frac12} (x), & E^{(n)}(y;w)] \sim  3 N (x,y)^2 (x \cdot \partial_y)  E^{(n-1)}(y;w)+ \cr 
&+ w \frac{3 N}{n+1}(x,y)^2 (x \cdot \partial_y) \partial_w E^{(n-1)}(y;w) + \cr &+ w \frac{3}{n+1} (x,y) (x \cdot \partial_y)^2 E^{(n+1)}(y;w)
\end{align}
and
\begin{align}
[ J^{(3)}_{-\frac12} (x), & E^{(n)}(y;w)] \sim \frac{3 N}{n+1}(x,y)^2 (x \cdot \partial_y) \partial_w E^{(n-1)}(y;w) + \cr &+ \frac{3}{n+1} (x,y) (x \cdot \partial_y)^2 E^{(n+1)}(y;w)
\end{align}
and then
\begin{align}
[ J^{(3)}_{\frac12} (x), & e^{(n)}_k(y)] \sim  3 N \frac{\frac{n}{2}-k}{n+1}(x,y)^2 (x \cdot \partial_y)  e^{(n-1)}_{k+\frac12}(y)
+ \cr &+ \frac{3}{n+1} (x,y) (x \cdot \partial_y)^2 e^{(n+1)}_{k+\frac12}(y)
\end{align}
and
\begin{align}
[ J^{(3)}_{-\frac12} (x), & e^{(n)}_k(y)] \sim  3 N \frac{\frac{n}{2}+k}{n+1}(x,y)^2 (x \cdot \partial_y)  e^{(n-1)}_{k-\frac12}(y)
+ \cr &+ \frac{3}{n+1} (x,y) (x \cdot \partial_y)^2 e^{(n+1)}_{k-\frac12}(y)
\end{align}

\section{Computation of the BRST cohomology of the $\mscr{N}=4$ chiral algebra at large $N$} \label{app:coho}
The $\mscr{N}=4$ chiral algebra is the BRST reduction of a system of symplectic bosons valued in $\mf{gl}_N \oplus \mf{gl}_N $ by the adjoint action $\mf{gl}_N$.   As such, the chiral algebra is generated by operators:  
\begin{align} 
	X^i_j,Y^i_j &  \text{ of spin } 1/2 \\
	\b^i_j & \text{ of spin } 1 & \text{ and ghost number } -1 \\
	\c^i_j & \text{ of spin } 0 & \text{ and ghost number } 1.
\end{align}
The OPEs are
\begin{align} 
	X^i_j (0) Y^k_l (z) & \simeq \delta^i_l \delta^k_j \frac{1}{z} \\
	\b^i_j(0) \c^k_l(z) & \simeq \delta^i_l \delta^k_j \frac{1}{z}.
\end{align}
The BRST operator is, as usual, given by the current
\begin{equation} 
	J_{BRST} = \op{Tr} \c X  Y + \tfrac{1}{2} \op{Tr} \b \c^2. 
\end{equation}
We let $X_n = \oint z^{-n-1} X(z) \d z$, and similarly for $Y_n, \b_n, \c_n$.  Then $X_n$ is of spin $n+1/2$, etc.   

At a first pass, the vacuum module for the chiral algebra is then generated by the action of the modes $X_n$, $Y_n$,  $\b_n$, $\c_n$ for $n \ge 0$.   The BRST operator is defined by $\oint J_{BRST}(z) \d z$.     

This is not strictly correct. We should not include the zero modes of the $\c$-ghost in the vacuum module.  The purpose of the $\c$-ghost is to impose gauge invariance.  The zero mode of the $\c$-ghost imposes gauge invariance for constant gauge transformations.  Since constant gauge transformations form a compact group, we should impose gauge invariance for these transformations directly, by looking only at $U(N)$ (or $GL(N,\C$) invariant states.   Gauge invariance for non-constant gauge transformations is best achieved cohomologically, using the non-zero modes of the $\c$-ghost. 

Thus, the correct definition of the vacuum module is given by considering the $GL(N)$ invariants in the module generated by the action of $X_n$, $Y_n$, $\b_n$ for $n \ge 0$ and $\c_n$ for $n \ge 1$ on the vacuum vector. The BRST operator acts on this vacuum module by $Q_{BRST} = \oint J_{BRST}(z) \d z$, as before; except that we drop any terms containing the ghost zero mode $\c_0$ after applying $Q_{BRST}$. 

We will compute the BRST cohomology at large $N$ using spectral sequence methods.  This means that we will approximate the BRST operator by a simpler operator, whose cohomology we can compute; and then, if necessary, correct by taking cohomology with respect to a secondary differential.

The first approximation we will use is that we will pass to the semi-classical limit. This means the following. Introduce a parameter $\hbar$ so that the OPE between $X$ and $Y$ is $\hbar / z$, and that between $\b$ and $\c$ is $\hbar / z$.  Then, when we apply $Q_{BRST}$ to a state in the vacuum module, we can expand $\hbar^{-1} Q_{BRST} = Q_{BRST}^{sc} + \dots$.  The first approximation we will take is that we will compute the cohomology of $Q_{BRST}^{sc}$.

In the semi-classical approximation, the ordering of the operators forming a state doesn't matter.  We can write the vacuum module as the $GL(N)$ invariants in a polynomial algebra in even and odd variables: 
\begin{equation} 
	\op{Vac} = \C[X_n,Y_m, \b_r, \c_s]^{GL(N)} \text{ where } n,m,r \ge 0 \text{ and } s \ge 1. 
\end{equation}
The semi-classical BRST operator takes the form
\begin{equation} 
	Q_{BRST}^{sc} = \c_n X_m \partial_{X_{n+m}} + \c_n Y_m \partial_{Y_{n+m}} + \c_m \b_n \partial_{\b_{n+m}} + X_m Y_n \partial_{\b_{n+m}} + \tfrac{1}{2} \c_m \c_n \partial_{\c_{n+m}}.   
\end{equation}
Here we sum over $n,m$ in every term, and drop terms containing $\c_0$.

This BRST cochain complex can be rewritten in terms of an auxiliary graded algebra. Let
\begin{equation} 
	A = \C[z,\eps_1,\eps_2] 
\end{equation}
where $\eps_i$ are of ghost number $1$, and hence anti-commuting.  (The algebra $A$ is the cohomology of the space of open-string states on a single brane wrapping $\C$ inside $\C^3$). 

Then, we claim that the vacuum module, with its semi-classical BRST operator, is precisely the relative Lie algebra cochains of $A \otimes \mf{gl}_N$, relative to $\gl_N$. This is a general phenomenon: local operators in any open-string field theory are always relative Lie algebra cohomology for algebra of open-string states.

By definition, the relative Lie algebra cochains of $A \otimes \mf{gl}_N$ are the $GL(N)$ invariants in the symmetric algebra of $A^\vee/\C [-1] \otimes \gl_N$.  This space has a certain differential. 

To make the identification with the vacuum module of the chiral algebra, we will identify the dual of $\C[z]$ with expressions like $z^{-n-1} \d z$, using the residue pairing.  The dual to $\eps_i$ will be denoted $\eps_i^\ast$.  The identification is
\begin{align} 
	X_n &= \eps_1^\ast z^{-n-1} \d z \\
	Y_n &= \eps_2^\ast z^{-n-1} \d z\\
	\b_n &= \eps_1^\ast \eps_2^\ast z^{-n-1} \d z \\
	\c_n &= z^{-n-1} \d z.
\end{align}
Then, it is clear that the vacuum module is the same as the $GL(N)$ invariants in the symmetric algebra on $A^\vee/ \ C \otimes \mf{gl}_N$.  

One can also check, without much difficulty, that the semi-classical BRST operator matches the differential on Lie algebra cochains. The only term which is not completely trivial to match is the term $\sum X_n Y_m \partial_{\b_{n+m}}$.  There is such a term in the Lie algebra cochains also, coming from the commutator
\begin{equation} 
	[M_1 \eps_1 z^n, M_2 \eps_2 z^m] = [M_1,M_2] \eps_1 \eps_2 z^{n+m} 
\end{equation}
for $M_1,M_2 \in \mf{gl}_N$.  

We have rephrased the problem of computing the cohomology of the semi-classical BRST operator in terms of that of computing the relative Lie algebra cohomology of $A \otimes \mf{gl}_N$. There is an isomorphism
\begin{equation} 
	H^\ast(A \otimes \gl_N \mid \gl_N) \otimes H^\ast(\gl_N) \iso H^\ast(A \otimes \gl_N)  
\end{equation}
of Lie algebra cohomology groups, where on the right hand side we have the absolute Lie algebra cohomology and on the left we have the relative Lie algebra cohomology, tensored with the Lie algebra cohomology of $\gl_N$.  The Lie algebra cohomology of $\gl_N$ is an exterior algebra on variables $\eps_1,\eps_3,\dots,\eps_{2N-1}$ where $\eps_i$ is in degree $i$. (These correspond to the operators $\op{Tr} \c^i$ built from the ghost field).

From this we see that it suffices to compute the Lie algebra cohomology of $A \otimes \gl_N$, and then strip off the Lie algebra cohomology of $\gl_N$.

Tsygan \cite{Tsy83} and Loday-Quillen \cite{LodQui84} show that, when $N \to \infty$, the Lie algebra cohomology of $A \otimes \gl_N$ for any associative algebra $A$ is isomorphic to the symmetric algebra on the dual of the cyclic homology of $A$, with a shift by $1$:
\begin{equation} 
	H^\ast(A \otimes \gl_\infty) \simeq \Sym^\ast (HC_\ast(A)^\vee[-1]). 
\end{equation}
Elements of $HC_\ast(A)^\vee[-1]$ are single-trace operators in the gauge theory.  

To compute $HC_\ast(A)$, we use the result of Connes (see the book \cite{Loday2013}) that for any associative algebra there is a spectral sequence converging to $HC_\ast(A)$ whose first term is
\begin{equation} 
	HH_\ast(A) \otimes \C[t^{-1}]. 
\end{equation}
Here $HH_\ast(A)$ is the Hochschild homology of $A$, and $t$ is a parameter of degree $2$.  The differential on this first page of the spectral sequence is $t B$, where $B$ is the Connes $B$-operator on Hochschild homology.

The Hochschild homology of any graded commutative algebra can be computed by the Hochschild-Kostant-Rosenberg theorem.  In our case, we find
\begin{equation} 
	HH_\ast(A) = \C[z,\d z, \eps_i, \d \eps_i ] 
\end{equation}
where $\d z$ is of cohomology degree $-1$, and so fermionic; whereas $\d \eps_i$ are of cohomological degree $0$, and bosonic. The Connes $B$-operator is the de Rham operator
\begin{equation} 
	B = \d_{dR} = \d z \partial_z + \sum \d \eps_i \partial_{\eps_i}.  
\end{equation}
We can decompose this as a sum of three terms, $B^z = \d z \partial_z$ and $B^{\eps_i} = \d \eps_i \partial_{\eps_i}$.

Next, let us compute the cohomology of $HH_\ast(A)[t^{-1}]$ with the differential $t B_z$. This differential does not affect the $\eps_i$, $\d \eps_i$ variables, so we may as well compute the cohomology of $\C[z,\d z,t^{-1}]$ with the differential $B_z t$. This cohomology consists of $\C[z] \oplus t^{-1} \C[t^{-1}]$.  Re-introducing the $\eps_i, \d \eps_i$ variables, we find our cyclic homology groups at this stage of the spectral sequence are
\begin{equation} 
	z \C[z] \otimes \C[\eps_i, \d \eps_i] \oplus \C[\eps_i, \d \eps_i] [t^{-1}]. 
\end{equation}
At this stage, we can introduce the differential $B_{\eps_1} t + B_{\eps_2} t$.  We find that the cohomology of this differential consists of
\begin{equation} 
	z \C[z] \otimes \C[\eps_i, \d \eps_i] \oplus \eps_1 \eps_2 \C[\d \eps_i] \oplus \left( \eps_1 \C[\d \eps_i] \oplus \eps_2 \C[\d \eps_i]       \right)/ (\op{Im} B_{\eps_i} )      \oplus t^{-1} \C[t^{-1}] . 
\end{equation}
The last summand, $t^{-1} \C[t^{-1}]$, corresponds to those operators which disappear when we pass to relative Lie algebra cohomology. We will therefore drop them.  

The summand $ \left( \eps_1 \C[\d \eps_i] \oplus \eps_2 \C[\d \eps_i]       \right)/ (\op{Im} B_{\eps_i} ) $ denotes expressions which are linear in the $\eps_i$, polynomials in the $\d \eps_i$, and taken up to the image of $B_{\eps_i} = \sum \d \eps_i \partial_{\eps_i}$.

There are no further terms in the spectral sequence, so this computation gives a complete description of cyclic homology and hence of the single-trace operators in the gauge theory.   Let us use the notation $[\mbf{m},\mbf{n}]$ to denote the representation of $SU(2)_R \times SO(2)$ which is of spin $m$ under $SU(2)_R$ and of spin $n$ under $SO(2)$.  We will write down all the terms we have found in terms of these representations:
\begin{align*} 
	z \C[z,\d \eps_i] &= \oplus_{k \ge 1,m \ge 0} [\mbf{m/2}, \mbf{m/2 + k} ]  \\
	\eps_1 \eps_2 \C[z, \d \eps_i] &= \oplus_{k \ge 1, m \ge 0} [\mbf{m/2}, \mbf{m/2+k}] \\
	\eps_1 z\C[z, \d \eps_i] \oplus \eps_2 z \C[z], \d \eps_i] &= \oplus_{k \ge 1, m \ge 0} [\mbf{1/2}, \mbf{1/2}] \otimes [\mbf{m/2}, \mbf{m/2+k}]   \\
	&= \oplus_{k \ge 1, m \ge 0} \left(  [\mbf{(m+1)/2}, \mbf{(m+1)/2+k}] \oplus [\mbf{(m-1)/2}, \mbf{(m+1)/2+k}  \right)  \\
	\left( \eps_1 \C[\d \eps_i] \oplus \eps_2 \C[\d \eps_i]\right)/ \op{Im} B_{\eps_i} &= \oplus_{m \ge 1} [\mbf{m/2}, \mbf{m/2} ].
\end{align*}
The first two lines correspond to the fermionic single trace operators in the $B,C$ towers and the second two to bosonic single trace operators in the $A,D$ towers. 

\subsection{Absence of further terms in the spectral sequence}
We have seen that the semi-classical BRST cohomology of the single trace operators is exactly what we want: it is given by the bosonic $A,D$ towers and the fermionic $B,C$ towers. We need to check that there are no further differentials in the spectral sequence.

To see this, note that the BRST operator commutes with both $SU(2)_R$ and the $SL_2$ global conformal symmetry.  Therefore any further differentials in the spectral sequence will have the same feature, and will take Virasoro quasi-primaries to each other.  The $A$ tower contains quasi-primaries of spin $r$ living in the $SU(2)_R$ representation of spin $r$, the $D$ tower has quasi-primaries of spin $r+2$ in the representation of $SU(2)_R$ of spin $r$, and the $B$ and $C$ towers have quasi-primaries of spin $r+1$ in the $SU(2)_R$ representation of spin $r$.  Because of this, there are no possible fermionic operators commuting with the $SU(2)_R$ and $SL_2$ symmetries. 

\section{Chiral algebra on topological D-branes} \label{app:chiral}
In this section we will analyze the open-string field theory living on a brane in the topological $B$-model. We will focus on branes supported on an algebraic curve $C$, and re-derive the standard answer: the dimensional reduction of holomorphic Chern-Simons theory, i.e. a gauged $\beta \gamma$ system on $C$ valued in the adjoint of the $U(N)$ gauge group. This is the same as is the chiral algebra for $\mc{N}=4$ gauge theory. 

We will also modify the setup to include the effect of extra space-filling branes on the world-volume theory of the branes wrapping $C$. These will contribute extra $\beta \gamma$ systems valued in the fundamental and anti-fundamental representations of the 
$U(N)$ gauge group.

\subsection{Open strings on B-branes}

It is known (see the survey \cite{Aspinwall:2004jr}) that a coherent sheaf $E$ on $X$ gives rise to a topological $B$-brane.  The cohomology of the algebra of open-string states is the algebra $\op{Ext}^\ast(E,E)$: the self-Ext's of $E$ in the derived category of coherent sheaves on $X$.

We are interested in topological $B$-branes which are supported on a curve $C \subset X$.  The most general such brane can be quite complicated, but for the purposes of understanding the chiral algebra it suffices to consider branes which are of the form $\Oo_C^{\oplus N}$, where $\Oo_C$ of $X$.    Following the general presription \cite{Wit92} one should engineer a field theory living on $C$ whose fields are the ghost number $1$ states of an open string stretched between the brane $\Oo_C^{\oplus N}$ and itself. 

The Ext-groups $\op{Ext}^\ast(\Oo_C^{\oplus N}, \Oo_C^{\oplus N})$ are the zero-modes of the space of open string states, and hence do not capture the local fluctuations of the field.  To understand the field theory, we need to find a suitable cochain complex whose cohomology groups are these Ext-groups.

If we were studying the space-filling brane corresponding to a vector bundle $E$ on $X$, we can use the Dolbeault complex $\Omega^{0,\ast}(X, E^\vee \otimes E)$ as our cochain model for the Ext-groups.  With this model, the ghost number $1$ open-string states consist of $\Omega^{0,1}(X,E^\vee \otimes E)$, and one is led, following \cite{Wit92}, to the conclusion that the open-string field theory is holomorphic Chern-Simons theory.

In the case of a brane supported on a curve $C$ in $X$, this analysis does not immediately apply.  The general algorithm for constructing the space of open string states at the cochain level would be to resolve $\Oo_C$ by a complex of holomorphic vector bundles $E^\bullet$ on $X$, and then take the Dolbeault complex of $X$ with coefficients in the sheaf of differential graded algebras given by 
\begin{equation} 
	\underline{\op{End}}(E^\bullet) = (E^\bullet)^\vee \otimes E^\bullet. 
\end{equation}
In each degree, this differential graded algebra is a holomorphic vector bundle on $X$, and the differential is a map of holomorphic vector bundles on $X$.

This algorithm, however, yields something which is too large to work with.  To find a smaller model, we replace $\underline{\op{End}}(E^\bullet)$  by its cohomology, which is a graded algebra in the category of coherent sheaves on $C$.  This cohomology is denoted
\begin{equation} 
	\underline{\op{Ext}}^\bullet(\Oo_C,\Oo_C) =H^\ast\left(  \underline{\op{End}}(E^\bullet) \right). 
\end{equation}
Then, a nice model for the space of open strings stretched from $C$ to $C$ is given by the differential graded algebra
\begin{equation} 
	\Omega^{0,\ast}\left(C,  \underline{\op{Ext}}^\bullet(\Oo_C,\Oo_C) \right). 
\end{equation}
A standard computation tells us that 
\begin{equation} 
	  \underline{\op{Ext}}^\bullet(\Oo_C,\Oo_C) = \wedge^\ast N_C 
\end{equation}
where $N_C$ is the normal bundle to $C$ in $X$. 

There is a subtle point in this analysis which we have skipped over. When passing from the sheaf of differential graded algebras $\underline{\op{End}}(E^\bullet)$ to its cohomology sheaves, we should remember that the cohomology sheaves can have an induced $A_\infty$ structure. This $A_\infty$ structure makes $\wedge^\ast N_C$ into an $A_\infty$ algebra in the dg category of complexes of vector bundles on $C$, meaning that there can be $A_\infty$ operations \begin{equation} 
	\mu_n \in \Omega^{0,\ast}(C, \underline{\Hom} (\wedge^\ast N_C)^{\otimes n}, \wedge^\ast N_C ) ). 
\end{equation}
Fortunately, these $A_\infty$ operations all vanish in the case that $X$ is the total space of a vector bundle $V = N_C$ on $C$.  (For $X$ to be Calabi-Yau, we need $\wedge^2 V = K_C$).  In what follows we will focus on the special case
\begin{equation} 
	V = K_C^{1/2} \otimes \C^2. 
\end{equation}
In this case, the space of open-string states for strings stretched from $\Oo_C$ to itself is the differential graded algebra
\begin{equation} 
	\mathcal{A}_C=	\Omega^{0,\ast}(C, \wedge^\ast (K_C^{1/2} \otimes \C^2) ).
\end{equation}
More generally, if we take $N$ copies of the brane $\Oo_C$, we would find the algebra $\mathcal{A}_C \otimes \mf{gl}_N$.
\subsection{The open-string field theory}
Following Witten's prescription, the fields of the open-string field theory living on the brane $\Oo_C^{\oplus N}$ consist of the ghost number $1$ elements of the open-string differential graded algebra $\mathcal{A}_C \otimes \mf{gl}_N$, equipped with the Chern-Simons action 
\begin{equation} 
	\int_C 	\tfrac{1}{2} \op{Tr} (\alpha \d \alpha) + \tfrac{1}{3} \op{Tr} \alpha^3.\label{eqn_CS_action} 
\end{equation}
Here by $\op{Tr}$ we mean the trace map
\begin{equation} 
	\mathcal{A}_C^3 \otimes \gl_N = \Omega^{1,1}(C) \otimes \gl_N \to \Omega^{1,1}(C). 
\end{equation}
If we work in the BV formalism, the space of fields, ghosts, anti-fields and anti-fields to ghosts is $\mathcal{A}_C \otimes\mf{gl}_N[1]$, where $[1]$ indicates a shift in ghost number by $1$.  The BV action is given by same formula \eqref{eqn_CS_action}, but where $\alpha$ is now taken to be an open string state of arbitrary ghost number.

Let us now translate this into more familiar terms. The fields of ghost number zero consist of a partial connection $A \in \Omega^{0,1}(C) \otimes \mf{gl}_N$ and two fields $X,Y \in \Omega^{0,0}(C, K_C^{1/2})$. Despite the fact that the fields $X,Y$ are spinors, they are \emph{bosonic} fields.  The Lagrangian is
\begin{equation} 
	\int_C X \dbar_A Y. 
\end{equation}
Thus, the fields $X,Y$ are a system of symplectic bosons valued in $\C^2 \otimes \gl_N$.  These are coupled to a gauge field $A \in \Omega^{0,1}(C) \otimes \mf{gl}_N$, and acted on in the usual way by gauge transformations which at the infinitesimal level are in $\Omega^{0,0}(C) \otimes \mf{gl}_N$.

Coupling to the gauge field $A$ has the effect of performing BRST reduction to the system of symplectic bosons. We will describe the algebra of local operators of our gauge theory, with its BRST differential. We will find that this algebra is obtained from the algebra of the system of symplectic bosons by adjoining a $\b$ and $\c$ ghost, with the usual BRST differential.

The $\b$ and $\c$ ghost operators are of non-zero ghost number, and so must come from the ghosts and anti-fields of the gauge theory.  In the BV formalism, the fields of the gauge theory consist of
\begin{equation} 
	\Omega^{0,\ast}(C, \wedge^\ast (K^{1/2} \otimes \C^2) \otimes \mf{gl}_N[1]. 
\end{equation}
The fields can be written in terms of four copies of the Dolbeault complex of $C$:
\begin{align} 
	\mbf{C} \in	\Omega^{0,\ast}(C)\otimes \gl_N[1] \\
	\mbf{X} \in  \Omega^{0,\ast}(C, K^{1/2}) \otimes \gl_N\\
	\mbf{Y}	\in \Omega^{0,\ast}(C, K^{1/2}) \otimes \gl_N
\\
	\mbf{B} \in  \Omega^{0,\ast}(C,K)\otimes \gl_N[-1] \\
\end{align}
The linearized BRST operator is the $\dbar$ operator on each copy of the Dolbeault complex.  The various fields have the following interpretation:
\begin{align}
	\mbf{C}^0 &= \c & \mbf{C}^1 &= A \\
	\mbf{X}^0 &= X & \mbf{X}^1 &= Y^\ast\\
	\mbf{Y}^0 &= Y & \mbf{Y}^1 &= -X^\ast \\
	\mbf{B}^0 &= A^\ast & \mbf{B}^1 &= \c^\ast	
\end{align}
where the $Y^\ast$ means the anti-field to $Y$, etc.  We will identify $A^\ast$ with the $\b$-ghost.

The algebra of local operators at a point $p \in C$ is generated by the operators:
\begin{align} 
	\partial_z^n \c(p) &  \text{ of spin } n\\ 
	\partial_z^n X(p) &  \text{ of spin } n+\tfrac{1}{2} \\ 
	\partial_z^n Y(p)  &  \text{ of spin } n+ \tfrac{1}{2} \\ 
	\partial_z^n \b(p) &    \text{ of spin } n +1. 
\end{align}
(where $n \ge 0$). The remaining operators we can write down are either BRST exact or not BRST closed. Operators involving $\zbar$ derivatives of the fields $\c,X,Y,\b$ are BRST exact, whereas operators involving the fields $A,Y^\ast, X^\ast, \c^\ast$ are not BRST closed.     Later, we will see that a more careful analysis tells us we should discard the operators built from the $\c$-ghost without any derivatives. 

Because the kinetic term in the action functional contains the $\dbar$ operator, and the cubic term involves fields which are in $\Omega^{0,1}$, the only non-trivial OPEs are given by
\begin{align} 
	\c(0) \cdot \b(z)& \simeq \frac{1}{z} \\
	X(0) \cdot Y(z) & \simeq \frac{1}{z}.
\end{align}
To verify that the algebra of operators of the gauge theory is the BRST reduction of that of a system of symplectic bosons, we need to analyze the non-linear term in the BRST operator of the gauge theory.  This is associated to the cubic term in the action, which is 
\begin{equation} 
	\tfrac{1}{3}	\int \op{Tr} \alpha^3 
\end{equation}
for $\alpha \in \Omega^{0,\ast}(C, \wedge^\ast(K_C^{1/2} \otimes \C^2) )$.  If we expand $\alpha$ into components, this reads
\begin{equation} 
     S^{cubic} = 	\int A X Y + \tfrac{1}{2} \int \c \c \c^\ast + \int \c Y^\ast Y  
	+ \int \c X X^\ast + \int A \c A^\ast. 
\end{equation}
Recalling that we identify $A = \b^\ast$, and so $\b = A^\ast$, from this expression we can read of the BRST transformations of any operator. The formula is
\begin{equation} 
	Q_{BRST} O = \{O,S^{cubic}\} 
\end{equation}
where $\{-,-\}$ is the BV anti-bracket.  We find explicitly:
\begin{align} 
	Q_{BRST} \b &= XY + [\c, \b] &  Q_{BRST} \c &= \tfrac{1}{2} [\c, \c] \\	
	Q_{BRST} X &=  [\c, X]   & Q_{BRST} Y &= [\c, Y] .
\end{align}
This is the BRST operator for the BRST reduction of the symplectic boson.  

We have shown that the chiral algebra living on the $B$-brane $\Oo_C^{\oplus N}$ inside the Calabi-Yau three-fold which is the total space of $K_C^{1/2} \oplus K_C^{1/2}$ is the chiral algebra of $\mscr{N}=4$ Yang-Mills, with $(X,Y) \equiv (Z^1,Z^2)$.

\subsection{Including space-filling branes}
This analysis can be generalized to give a description of the chiral algebra living on a $1$-complex dimensional brane in a Calabi-Yau $3$-fold, when we also include space-filling branes.

The theory on $K$ space-filling brane is holomorphic Chern-Simons theory for the group $\mf{gl}(K)$. More generally, one can consider $K$ branes and $L$ anti-branes, giving rise to holomorphic Chern-Simons for the supergroup $\mf{gl}(K \mid L)$. As we already discussed, the only situation in which this theory is consistent at the quantum level is when $K = L$.

Because of this, we will focus on the case when we have the same number of space-filling branes and anti-branes, and we will analyze the effect this has on the chiral algebra supported on a brane on an algebraic curve.  The coherent sheaf corresponding to $K$ space-filling branes is $\Oo_X \otimes \C^K$. The sheaf corresponding to $K$ space-filling anti-branes is $\Oo_X \otimes \Pi \C^K$. Here we work with sheaves which have two gradings, one corresponding to fermion number and the other two ghost number.

We will continue to restrict ourselves to the case when the Calabi-Yau three-fold $X$ is the total space of $K^{1/2} \oplus K^{1/2}$ over a curve $C$. The space of states for an open string stretched between the space-filling branes $\Oo_X \otimes \C^{K \mid K}$ and the branes $\Oo_C \otimes \C^N$ is 
\begin{equation} 
	\Omega^{0,\ast}(C, \op{Hom}(\C^{K \mid K}, \C^N) ). 
\end{equation}
A short calculation of $\op{Ext}$-sheaves tells us that open strings going from $\Oo_C\otimes \C^N$ to $\Oo_X \otimes \C^{K \mid K}$ contribute
\begin{equation} 
	\Omega^{1,\ast}(C, \op{Hom}(\C^N, \C^{K \mid K} )[-2]. 
\end{equation}
We conclude that the open-string field theory in the presence of
$K \mid K$ space-filling branes is obtained by adjoining to the theory we discussed earlier the bi-fundamental matter
\begin{equation} 
	\Omega^{0,\ast}(C, \op{Hom}(\C^{K \mid K}, \C^N) )[1] \oplus 	\Omega^{1,\ast}(C, \op{Hom}(\C^N, \C^{K \mid K} )[-1]  . 
\end{equation}
We have described the matter in the BV formalism, including fields and anti-fields.  To give a description in more familiar terms, we need to modify the ghost numbers of the fields in a procedure familiar from topological twisting.  We give the fields a grading coming from the $\C^\times$ action on $\C^{K \mid K}$ under which all elements have weight $1$. Then, we can shift the ghost number and fermion number of the fields by saying that a field of charge $l$ under this $U(1)$ action gets ghost number shifted by $l$ and fermion number shifted by $l$ modulo $2$.  

If we do this, we find the same bi-fundamental matter but with a different grading:
\begin{equation} 
		\Omega^{0,\ast}(C, \op{Hom}(\C^{K \mid K}, \C^N) ) \oplus 	\Omega^{1,\ast}(C, \op{Hom}(\C^N, \C^{K \mid K} )  .  
\end{equation}
With this grading, the fields of ghost number $0$ consist of 
\begin{align} 
	\gamma_{i\alpha} & \in \Omega^{0,0}(C, \op{Hom}(\C^{K \mid K}, \C^N) ) \\
\beta^{i\alpha} & \in \Omega^{1,0}(C, \op{Hom}(\C^N, \C^{K \mid K} )  ).  
\end{align}
The Lagrangian is $\int \beta \dbar_A \gamma$, where $A$ is the $\mf{gl}_N$ gauge field.  We find that the bifundamental matter we are coupling is a $\beta-\gamma$ system valued in bifundamental the super-representation $\op{Hom}(\C^{K \mid K}, \C^N)$ of $\mf{gl}_N \oplus \mf{gl}_{K \mid K}$.  

\section{D-brane sources} \label{app:source}
Consider now the topological $B$-model on a Calabi-Yau $3$-fold $X$ with $N$ branes living on a holomorphic curve $C \subset X$.  The closed-string field
\begin{equation} 
	\alpha \in \op{Ker} \partial \subset \oplus_{i+j = 2}\PV^{i,j}(X) 
\end{equation}
is coupled to the brane via $ N \int_{C} \partial^{-1} \alpha$.  Here we are interpreting $\partial^{-1} \alpha$ as a $2$-form on $X$, and integrating this over the curve $C$. The closed-string Lagrangian in the presence of this term becomes
\begin{equation} 
	\tfrac{1}{2} \int \partial^{-1} \alpha \dbar  \alpha + \tfrac{1}{6} \int \alpha^3 + N \int_C \partial^{-1} \alpha. 
\end{equation}
Let us solve the equations of motion in the presence of this term. Letting $\alpha = \partial \beta$ and varying $\beta$, we find 
\begin{equation} 
	\dbar \partial \beta + \tfrac{1}{2}\partial \left(  (\partial \beta)(\partial \beta) \right) + N \delta_C = 0.	
\end{equation}
Reinserting $\alpha = \partial \beta$, and using the identity
\begin{equation} 
	\partial(\alpha^2) = \{\alpha,\alpha\}, 
\end{equation}
we find
\begin{equation} 
	\dbar \alpha + \tfrac{1}{2}\{\alpha,\alpha\} + N \delta_C = 0. \label{eqn_source_appendix} 
\end{equation}
Here we are treating $\delta_C$ as an element of $\PV^{1,2}(X)$, using the isomorphism between $\PV^{1,2}(X)$ and $\Omega^{2,2}(X)$. 
This equation tells us that the obstruction to $\alpha$ satisfying the Maurer-Cartan equation defining an integrable deformation of complex structure of $X$ is given by $-N\delta_C \in \PV^{1,2}$.

\section{Some computations of cohomology of coherent sheaves}
\label{appendix_coherent_cohomology}
In this section we will perform certain computations of the cohomology of coherent sheaves on the compactified deformed conifold $\br{X}$, which we used at various points in the body of the paper.  These computations are all quite easy and use standard techniques of homological algebra. 

The first computation we will do is the following.
\begin{lemma} 
The cohomology $H^\ast(\br{X}, \Oo(-D))$ vanishes. 
 \end{lemma}
\begin{proof} 
Recall that $\br{X}$ is a quadric in $\CP^4$, and the divisor is the intersection of $\br{X}$ with a copy of $\CP^3$.  There is a short exact sequence of coherent sheaves on $\CP^4$
\begin{equation} 
0 \to \Oo_{\CP^4}(-3) \to \Oo_{\CP^4}(-1) \to \Oo_{\br{X}}(-1) \to 0. 
 \end{equation}
On any $\CP^n$, the cohomology of $\Oo(-k)$ vanishes for $k = 1,\dots,n$.  From this we see that the cohomology of $\Oo_{\br{X}}(-1)$ vanishes.

 \end{proof}

Next we will show that:
\begin{lemma} 
	The cohomology $H^\ast_{\dbar}(\br{X}, \Omega^{2}_{\br{X}}(\log D))$ vanishes. 
 \end{lemma}
\begin{proof} 
There is a short exact sequence of coherent sheaves on $\br{X}$
	\begin{equation} 
		0 \to \Omega^2(\br{X}) \to \Omega^2(\br{X},\log D) \xto{Res} \Omega^1(D) \to 0 	 
	\end{equation}
It is easy to compute that the topological Euler characteristic of $\br{X}$ is $4$, and by the Lefschetz theorem $H^i(\br{X})$ is that of $\CP^4$ in degrees $\le 2$. We conclude by Poincar\'e duality that $H^i(\br{X}) = H^i(\CP^3)$, and in particular that $H^i(\br{X}, \Omega^2_{\br{X}})$ is $\C$ in degree $2$ and $0$ in other degrees.  Similarly, $H^i(D, \Omega_D)$ is $\C^2$ in degree $1$ and zero in other degrees.   
We find from the long exact sequence in cohomology
\begin{equation} 
\dots \to H^1(\br{X}, \Omega^2_{\br{X}} (\log D)) \to H^1(D, \Omega^1_D) \to H^2(\br{X}, \Omega^2_{\br{X}} ) \to \dots 
 \end{equation}
 that $H^i(\br{X}, \Omega^2_{\br{X}} (\log D)) $ is $\C$ if $i = 1$ and zero otherwise.

\begin{equation} 
0 \to T_X(-2 D) \to T_{\CP^4}(-2) \mid_{\br{X}} \to \Oo_X \to 0. 
 \end{equation}
 Further, there is a short exact sequence of coherent sheaves on $\CP^4$
 \begin{equation} 
 0 \to T_{\CP^4}(-4) \to T_{\CP^4}(-2) \to T_{\CP^4}(-2)\otimes \Oo_{\br{X}} \to 0. 
  \end{equation}
Finally, we have the Euler sequence
\begin{equation} 
0 \to \Oo_{\CP^4} \to \Oo(1)^4 \to T_{\CP^4} \to 0. 
 \end{equation}
By tensoring the Euler sequence with $\Oo(-k)$, we find that the coherent cohomology of $T_{\CP^4}(-k)$ vanishes for $k = 2,3,4$.  
 \end{proof}

\section{Consistency of the boundary conditions for Kodaira-Spencer theory and holomorphic Chern-Simons theory}\label{app:bc}
The first thing to check is that the boundary conditions we have proposed \eqref{eqn_bc} are consistent boundary conditions.  We need to check that there is a propagator which satisfies these conditions, and that the interaction term for holomorphic Chern-Simons theory and Kodaira-Spencer theory is well-defined for fields satisfying these boundary conditions.

Let us first verify consistency of the boundary condition for holomorphic Chern-Simons theory for the group $\mf{psl}(K \mid K)$. Let $A^{0,\ast} \in \Omega^{0,\ast}(X) \otimes \mf{psl}(K \mid K)[1]$ be the field for holomorphic Chern-Simons theory, in the BV formalism.  Our boundary condition is that  $A^{0,\ast}$ extends to an element
\begin{equation} 
	A^{0,\ast} \in \Omega^{0,\ast}(\br{X}, \Oo(- D) ) \otimes \mf{psl}(K \mid K)[1]. 
\end{equation}
Since $\Omega_X$ has a pole of order $3$ on the boundary divisor, the cubic interaction
\begin{equation} 
\int \Omega_X (A^{0,\ast} \wedge A^{0,\ast} \wedge A^{0,\ast}) 
 \end{equation}
 is well-defined.  

Next, let us check that there is a propagator satisfying this boundary condition. To do this, we need to understand what equations a propagator on $\br{X}$ must satisfy. Let us first describe the corresponding equations for the propagator on $X$.

The $\delta$-function on the diagonal in $X \times X$ is naturally a section of the determinant of the conormal bundle to the diagonal, which is $K_X$.  We can use the holomorphic volume form $\Omega_X$ to view this $\delta$-function as a Dolbeault form valued in the trivial bundle:
\begin{equation} 
	\Omega_X^{-1} \delta_{\op{Diag}} \in \br{\Omega}^{0,3}(X \times X). 
\end{equation}
(Here $\br{\Omega}^{0,\ast}$ indicates Dolbeault forms whose coefficients can be distributions).  Stripping off the Lie algebra factor, the propagator is an element in 
\begin{equation} 
	P_X \in \br{\Omega}^{0,2}(X \times X) 
\end{equation}
which satisfies
\begin{equation} 
	\dbar P_X = \Omega_X^{-1} \delta_{\op{Diag}}. 
\end{equation}
We can also impose a gauge condition, for instance, asking that 
\begin{equation} 
	\dbar^\ast P_X = 0. 
\end{equation}

If we work on $\br{X}$, then the volume form $\Omega_{\br{X}}$ has poles, so that $\Omega_{\br{X}}^{-1}$ has zeroes.  The delta-function $\Omega_{\br{X}}^{-1} \delta_{\op{Diag}}$ on the diagonal is then a distribution twisted by the bundle $\Oo_{\br{X}}(-3D))$ on the diagonal.

If we have a function on $X \times X$ which has a first order zero on $D \times X$ and on $X \times D$, then when we restrict it to the diagonal it has a second order zero on the divisor $D$ in the diagonal.  The $\delta$-function on the diagonal is a distribution which has a third order zero (and in particular a second order zero) along the divisor $D$ in the diagonal. We can therefore view it as a distribution on $X \times X$ twisted by the line bundle $\Oo(-D) \boxtimes \Oo(-D)$: 
\begin{equation} 
	\Omega_X^{-1} \delta_{\op{Diag}} \in \br{\Omega}^{0,3}(\br{X} \times \br{X}, \Oo(- D) \boxtimes \Oo(- D) ). 
\end{equation}

A propagator for the theory with fields extend to $\br{X}$ is then defined to be an element
\begin{equation} 
	P_{\br{X}} \in \br{\Omega}^{0,2}(\br{X} \times \br{X}, \Oo(-D) \boxtimes \Oo(-D))  
\end{equation}
which satisfies the equation
\begin{equation} 
	\dbar P_{\br{X}} =\Omega_X^{-1} \delta_{\op{Diag}}.\label{eqn_propagator} 
\end{equation}
To fix the gauge freedom, we will also ask that $P_{\br{X}}$ satisfies the gauge condition
\begin{equation} 
	\dbar^\ast P_{\br{X}} = 0\label{eqn_gauge} 
\end{equation}
(where $\dbar^\ast$ is defined with respect to some K\"ahler metric on $\br{X}$). 

\emph{A priori}, it is not obvious that such a propagator exists, or that if it exists it is unique. However, a cohomology calculation given in the appendix \ref{appendix_coherent_cohomology} shows that
\begin{equation} 
	H^\ast_{\dbar}(\br{X}, \Oo(-D)) = 0. 
\end{equation}
This, together with standard Hodge theory results, immediately implies that there is a unique propagator $P_{\br{X}}$ which satisfies equation \eqref{eqn_propagator} and equation \eqref{eqn_gauge}.

\subsection{Boundary condition for the fields of Kodaira-Spencer theory}
Now we will repeat the above analysis in the more difficult case of Kodaira-Spencer theory. 
We will focus on the Beltrami differential field of KS theory, which we view as a $(2,1)$ form.  We impose the equation $\partial \alpha = 0$ homologically, so that the full field content in the BV formalism is an element of $\Omega^{\ge 2,\ast}(X)$.  The linearized BRST operator is the de Rham operator $\dbar + \partial$.

The boundary conditions we impose are that $\alpha$ extends to a form on $\br{X}$ with logarithmic poles on the boundary ,and that the reside of $\alpha$ in the boundary divisor integrates to $0$ over  the curve $z = 0$.  It is not difficult to check that the natural map
\begin{equation}
H^\ast (\Omega^{\ge 2,\ast}(\br{X}, \log D) ) \to H^\ast(X)  
\end{equation}
is an isomorphism in degrees $> 0$, so that the only cohomology of the left hand side corresponds to the class in $H^3(X)$. Imposing this residue condition on the field means that there are no zero modes.

Now let us explain how to construct the propagator. The $\delta$-function on the diagonal in $X \times X$ is an element
\begin{equation} 
	\delta_{\op{Diag}} \in \br{\Omega}^{3,3}(X \times X). 
\end{equation}
We can apply the operator $\partial$ on one factor to get a distribution
\begin{equation} 
	(\partial \otimes 1) \delta_{\op{Diag}} \in \br{\Omega}^{4,3}(X \times X). 
\end{equation}
We can project this onto the components which is in $\Omega^{2,\ast}$ on each factor of $X$:
\begin{equation} 
	\pi_{2,2} (\partial \otimes 1) \delta_{\op{Diag}} \in \br{\Omega}^{(\ge 2, \ge 2),3}(X \times X).
\end{equation}
Here by $\Omega^{(\ge 2,\ge 2),\ast}$ we mean it lives in $\Omega^{\ge 2,\ast}$ on each factor.

A propagator for Kodaira-Spencer theory, written in the language of differential forms instead of polyvector fields, is then an element
\begin{equation} 
	P_X \in \Omega^{(\ge 2, \ge 2),\ast}(X \times X) 
\end{equation}
which satisfies the conditions
\begin{align} 	
	P_X &= \sigma P_X \\
	(\partial + \dbar) P_X &= \pi_{2,2}(\partial \otimes 1) \delta_{\op{Diag}}.\label{eqn_ks_propagator} 
\end{align}
(In the first equation, $\sigma$ is the operator which switches the two factors of $X$). We could also impose a particular gauge and ask that $\dbar^\ast P_X = 0$ if we like. 

This makes sense if we extend everything to $\br{X}$ and have logarithmig poles on the boundary. Firstly, we note that
\begin{equation} 
	\pi_{2,2} (\partial \otimes 1) \delta_{\op{Diag}} \in \br{\Omega}^{(\ge 2, \ge 2),3}(\br{X} \times\br{X}, \log D) 
\end{equation}
where $D$ here is the boundary of $\br{X} \times \br{X}$. 

The question is then whether we can find an element
\begin{equation} 
	P_{\br{X}} \in   \br{\Omega}^{(\ge 2, \ge 2),2}(\br{X} \times\br{X}, \log D) 
\end{equation}
satisfying equations \ref{eqn_ks_propagator}.  Because the cohomology of  $\br{\Omega}^{(\ge 2, \ge 2),2}(\br{X} \times\br{X}, \log D)$ is one dimensional, spanned by a $6$-form class, there exists such a $P_{\br{X}}$.  It is not unique even up to the addition of an  exact element, because we can always add a multiple of the $6$-form class; however, if we further ask that the integral of the residue of $P_{\br{X}}$ in the boundary vanishes, it is unique up to the addition of an exact element.

\noindent {\bf Acknowledgements.} K.C. and D.G. are supported by the NSERC
Discovery Grant program and by the Perimeter Institute for Theoretical
Physics. Research at Perimeter Institute is supported by the
Government of Canada through Industry Canada and by the Province of
Ontario through the Ministry of Research and Innovation.

\providecommand{\href}[2]{#2}\begingroup\raggedright\endgroup



\begin{thebibliography}{10}

\bibitem{Bershadsky:1993ta}
M.~Bershadsky, S.~Cecotti, H.~Ooguri, and C.~Vafa, {\it {Holomorphic anomalies
  in topological field theories}},  {\em Nucl. Phys.} {\bf B405} (1993)
  279--304, [\href{http://arxiv.org/abs/hep-th/9302103}{{\tt hep-th/9302103}}].
  [AMS/IP Stud. Adv. Math.1,655(1996)].

\bibitem{Bershadsky:1993cx}
M.~Bershadsky, S.~Cecotti, H.~Ooguri, and C.~Vafa, {\it {Kodaira-Spencer theory
  of gravity and exact results for quantum string amplitudes}},  {\em Commun.
  Math. Phys.} {\bf 165} (1994) 311--428,
  [\href{http://arxiv.org/abs/hep-th/9309140}{{\tt hep-th/9309140}}].

\bibitem{Maldacena:1997re}
J.~M. Maldacena, {\it {The Large N limit of superconformal field theories and
  supergravity}},  {\em Int. J. Theor. Phys.} {\bf 38} (1999) 1113--1133,
  [\href{http://arxiv.org/abs/hep-th/9711200}{{\tt hep-th/9711200}}]. [Adv.
  Theor. Math. Phys.2,231(1998)].

\bibitem{Witten:1998qj}
E.~Witten, {\it {Anti-de Sitter space and holography}},  {\em Adv. Theor. Math.
  Phys.} {\bf 2} (1998) 253--291,
  [\href{http://arxiv.org/abs/hep-th/9802150}{{\tt hep-th/9802150}}].

\bibitem{Gopakumar:1998ki}
R.~Gopakumar and C.~Vafa, {\it {On the gauge theory / geometry
  correspondence}},  {\em Adv. Theor. Math. Phys.} {\bf 3} (1999) 1415--1443,
  [\href{http://arxiv.org/abs/hep-th/9811131}{{\tt hep-th/9811131}}]. [AMS/IP
  Stud. Adv. Math.23,45(2001)].

\bibitem{Dijkgraaf:2002dh}
R.~Dijkgraaf and C.~Vafa, {\it {A Perturbative window into nonperturbative
  physics}},  \href{http://arxiv.org/abs/hep-th/0208048}{{\tt hep-th/0208048}}.
  
\bibitem{Dijkgraaf:2002fc} 
  R.~Dijkgraaf and C.~Vafa,
  {\it {Matrix models, topological strings, and supersymmetric gauge theories}},
  {\em Nucl.\ Phys.\ B} {\bf 644}, 3 (2002)
 [\href{http://arxiv.org/abs/hep-th/0206255}{{\tt hep-th/0206255}}].

\bibitem{Aganagic:2003qj} 
  M.~Aganagic, R.~Dijkgraaf, A.~Klemm, M.~Marino and C.~Vafa,
  {\it {Topological strings and integrable hierarchies}},
  {\em Commun.\ Math.\ Phys.}  {\bf 261}, 451 (2006)
  [\href{http://arxiv.org/abs/hep-th/0312085}{{\tt hep-th/0312085}}].
  
\bibitem{Dijkgraaf:2007sx}
R.~Dijkgraaf and C.~Vafa, {\it {Two Dimensional Kodaira-Spencer Theory and
  Three Dimensional Chern-Simons Gravity}},
  \href{http://arxiv.org/abs/0711.1932}{{\tt arXiv:0711.1932}}.

\bibitem{Witten:1991zd}
E.~Witten, {\it {Ground ring of two-dimensional string theory}},  {\em Nucl.
  Phys.} {\bf B373} (1992) 187--213,
  [\href{http://arxiv.org/abs/hep-th/9108004}{{\tt hep-th/9108004}}].
  
\bibitem{Ghoshal:1995wm} 
  D.~Ghoshal and C.~Vafa,
  {\it{C = 1 string as the topological theory of the conifold}},
  {\em Nucl.\ Phys.\ B} {\bf 453}, 121 (1995)
  [\href{http://arxiv.org/abs/hep-th/9506122}{{\tt hep-th/9506122}}].
  
\bibitem{Ooguri:1999bv}
H.~Ooguri and C.~Vafa, {\it {Knot invariants and topological strings}},  {\em
  Nucl. Phys.} {\bf B577} (2000) 419--438,
  [\href{http://arxiv.org/abs/hep-th/9912123}{{\tt hep-th/9912123}}].

\bibitem{IMZ2018}
N.~Ishtiaque, S.~F. Moosavian, and Y.~Zhou, {\it Topological holography: The
  example of the d2-d4 brane system},
  \href{http://arxiv.org/abs/arXiv:1809.00372}{{\tt arXiv:1809.00372}}.

\bibitem{Beem:2013sza}
C.~Beem, M.~Lemos, P.~Liendo, W.~Peelaers, L.~Rastelli, and B.~C. van Rees,
  {\it {Infinite Chiral Symmetry in Four Dimensions}},  {\em Commun. Math.
  Phys.} {\bf 336} (2015), no.~3 1359--1433,
  [\href{http://arxiv.org/abs/1312.5344}{{\tt arXiv:1312.5344}}].

\bibitem{Costello:2016mgj}
K.~Costello and S.~Li, {\it {Twisted supergravity and its quantization}},
  \href{http://arxiv.org/abs/1606.00365}{{\tt arXiv:1606.00365}}.
  
\bibitem{Bonetti:2016nma}
F.~Bonetti and L.~Rastelli,
{\it {Supersymmetric localization in AdS$_{5}$ and the protected chiral algebra}},
{\em JHEP} \textbf{08}, 098 (2018)
 \href{http://arxiv.org/abs/1612.06514}{{\tt arXiv:1612.06514}}.


\bibitem{SenZwi96}
A.~Sen and B.~Zwiebach, {\it Background independent algebraic structures in
  closed string field theory},  {\em Comm. Math. Phys.} {\bf 177} (1996), no.~2
  305--326.

\bibitem{Cos05}
K.~Costello, {\it The partition function of a topological field theory},  {\em
  Journal of {T}opology} {\bf 2} (2009), no.~4
  [\href{http://arxiv.org/abs/arXiv:math.QA/0509264}{{\tt
  arXiv:math.QA/0509264}}].

\bibitem{Sen:2015uaa}
A.~Sen, {\it {BV Master Action for Heterotic and Type II String Field
  Theories}},  {\em JHEP} {\bf 02} (2016) 087,
  [\href{http://arxiv.org/abs/1508.05387}{{\tt arXiv:1508.05387}}].

\bibitem{McGreevy:2000cw}
J.~McGreevy, L.~Susskind, and N.~Toumbas, {\it {Invasion of the giant gravitons
  from Anti-de Sitter space}},  {\em JHEP} {\bf 06} (2000) 008,
  [\href{http://arxiv.org/abs/hep-th/0003075}{{\tt hep-th/0003075}}].

\bibitem{Grisaru:2000zn}
M.~T. Grisaru, R.~C. Myers, and O.~Tafjord, {\it {SUSY and goliath}},  {\em
  JHEP} {\bf 08} (2000) 040, [\href{http://arxiv.org/abs/hep-th/0008015}{{\tt
  hep-th/0008015}}].

\bibitem{Biswas:2006tj}
I.~Biswas, D.~Gaiotto, S.~Lahiri, and S.~Minwalla, {\it {Supersymmetric states
  of N=4 Yang-Mills from giant gravitons}},  {\em JHEP} {\bf 12} (2007) 006,
  [\href{http://arxiv.org/abs/hep-th/0606087}{{\tt hep-th/0606087}}].

\bibitem{Maloney:2007ud}
A.~Maloney and E.~Witten, {\it {Quantum Gravity Partition Functions in Three
  Dimensions}},  {\em JHEP} {\bf 02} (2010) 029,
  [\href{http://arxiv.org/abs/0712.0155}{{\tt arXiv:0712.0155}}].

\bibitem{Dijkgraaf:2000fq}
R.~Dijkgraaf, J.~M. Maldacena, G.~W. Moore, and E.~P. Verlinde, {\it {A Black
  hole Farey tail}},  \href{http://arxiv.org/abs/hep-th/0005003}{{\tt
  hep-th/0005003}}.

\bibitem{Yin:2007gv}
X.~Yin, {\it {Partition Functions of Three-Dimensional Pure Gravity}},  {\em
  Commun. Num. Theor. Phys.} {\bf 2} (2008) 285--324,
  [\href{http://arxiv.org/abs/0710.2129}{{\tt arXiv:0710.2129}}].

\bibitem{Witten:1998xy}
E.~Witten, {\it {Baryons and branes in anti-de Sitter space}},  {\em JHEP} {\bf
  07} (1998) 006, [\href{http://arxiv.org/abs/hep-th/9805112}{{\tt
  hep-th/9805112}}].

\bibitem{Gaiotto:2009gz}
D.~Gaiotto and J.~Maldacena, {\it {The Gravity duals of N=2 superconformal
  field theories}},  {\em JHEP} {\bf 10} (2012) 189,
  [\href{http://arxiv.org/abs/0904.4466}{{\tt arXiv:0904.4466}}].

\bibitem{Costello:2018fnz}
K.~Costello and D.~Gaiotto, {\it {Vertex Operator Algebras and 3d $\mathcal
  N=4$ gauge theories}},  \href{http://arxiv.org/abs/1804.06460}{{\tt
  arXiv:1804.06460}}.

\bibitem{Tsy83}
B.~Tsygan, {\it Homology of matrix algebras over rings and the hochschild
  homology},  {\em Uspeki Math. Nauk.} {\bf 38} (1983) 217--218.

\bibitem{LodQui84}
J.-L. Loday and D.~Quillen, {\it Cyclic homology and the lie algebra homology
  of matrices},  {\em Comment. Math. Helv.} {\bf 59} (1984), no.~4 569--591.

\bibitem{Karch:2002sh}
A.~Karch and E.~Katz, {\it {Adding flavor to AdS / CFT}},  {\em JHEP} {\bf 06}
  (2002) 043, [\href{http://arxiv.org/abs/hep-th/0205236}{{\tt
  hep-th/0205236}}].

\bibitem{Kachru:1998ys}
S.~Kachru and E.~Silverstein, {\it {4-D conformal theories and strings on
  orbifolds}},  {\em Phys. Rev. Lett.} {\bf 80} (1998) 4855--4858,
  [\href{http://arxiv.org/abs/hep-th/9802183}{{\tt hep-th/9802183}}].

\bibitem{Lawrence:1998ja}
A.~E. Lawrence, N.~Nekrasov, and C.~Vafa, {\it {On conformal field theories in
  four-dimensions}},  {\em Nucl. Phys.} {\bf B533} (1998) 199--209,
  [\href{http://arxiv.org/abs/hep-th/9803015}{{\tt hep-th/9803015}}].

\bibitem{Bershadsky:1998cb}
M.~Bershadsky and A.~Johansen, {\it {Large N limit of orbifold field
  theories}},  {\em Nucl. Phys.} {\bf B536} (1998) 141--148,
  [\href{http://arxiv.org/abs/hep-th/9803249}{{\tt hep-th/9803249}}].

\bibitem{CosLi15}
K.~Costello and S.~Li, {\it Quantization of open-closed bcov theory, i},
  \href{http://arxiv.org/abs/arXiv:1505.06703}{{\tt arXiv:1505.06703}}.

\bibitem{Nekrasov:2004js}
N.~Nekrasov, H.~Ooguri, and C.~Vafa, {\it {S duality and topological strings}},
   {\em JHEP} {\bf 10} (2004) 009,
  [\href{http://arxiv.org/abs/hep-th/0403167}{{\tt hep-th/0403167}}].

\bibitem{Witten:1992fb}
E.~Witten, {\it {Chern-Simons gauge theory as a string theory}},  {\em Prog.
  Math.} {\bf 133} (1995) 637--678,
  [\href{http://arxiv.org/abs/hep-th/9207094}{{\tt hep-th/9207094}}].

\bibitem{Cos11}
K.~Costello, {\em Renormalization and effective field theory}.
\newblock Surveys and monographs. American Mathematical Society, 2011.

\bibitem{Li16}
S.~Li, {\it Vertex algebras and quantum master equation},
  \href{http://arxiv.org/abs/1612.01292}{{\tt arXiv:1612.01292}}.

\bibitem{Deligne1971}
P.~Deligne, {\it Th{\'e}orie de hodge, ii},  {\em Publications de l'Institute
  des Haute 'Etude Scientifiques} {\bf 40} (1971), no.~1.

\bibitem{Loday2013}
J.-L. Loday, {\it Cyclic homology}, .

\bibitem{Aspinwall:2004jr}
P.~S. Aspinwall, {\it {D-branes on Calabi-Yau manifolds}},  in {\em {Progress
  in string theory. Proceedings, Summer School, TASI 2003, Boulder, USA, June
  2-27, 2003}}, pp.~1--152, 2004.
\newblock \href{http://arxiv.org/abs/hep-th/0403166}{{\tt hep-th/0403166}}.

\bibitem{Wit92}
E.~Witten, {\it Chern-simons gauge theory as a string theory},
  \href{http://arxiv.org/abs/hep-th/9207094}{{\tt hep-th/9207094}}.

\end{thebibliography}

\end{document}